\documentclass[aos]{imsart-custom}

\RequirePackage{amsthm,amsmath,amsfonts,amssymb}
\RequirePackage[authoryear]{natbib}%% uncomment this for author-year citations
\RequirePackage[colorlinks,citecolor=blue,urlcolor=blue]{hyperref}%% uncomment this for coloring bibliography citations and linked URLs
\RequirePackage{graphicx}%% uncomment this for including figures

\startlocaldefs
\theoremstyle{plain}
\newtheorem{theorem}{Theorem}
\newtheorem{prp}{Proposition}
\newtheorem{lem}{Lemma}

\newtheorem{defi}{Definition}
\newtheorem{clm}{Claim}

\theoremstyle{remark}
\newtheorem{rem}{Remark}
\usepackage{graphicx}
\usepackage{subcaption}
\usepackage{enumitem}
\usepackage{booktabs}
\usepackage{placeins}
\usepackage{amsthm}
\usepackage[utf8]{inputenc}
\usepackage{bbm}

\DeclareMathOperator*{\argmax}{arg\, max}

\def\cN{\mathcal{N}}
\def\cP{\mathcal{P}}
\def\cS{\mathcal{S}}

\def\bR{\mathbf{R}}
\def\bX{\mathbf{X}}

\def\br{\mathbf{r}}

\def\bbN{\mathbb{N}}
\def\bbR{\mathbb{R}}

\renewcommand{\equiv}{\vcentcolon =}
\def\Rb{\overline R}

\renewcommand{\P}[1]{\mathbb{P}\left(#1\right)}
\newcommand{\Phn}[1]{\mathbb{P}_{H_0}\left(#1\right)}
\newcommand{\Pha}[1]{\mathbb{P}_{H_1}\left(#1\right)}
\newcommand{\PP}{\mathbb{P}}

\newcommand{\E}[1]{\mathbb{E}\left[ #1 \right]}

\newcommand{\Eha}[1]{\mathbb{E}_{H_1}\left[ #1 \right]}
\newcommand{\Var}[1]{\text{Var}\left( #1 \right)}
\newcommand{\Cov}[1]{\text{Cov}\left( #1 \right)}
\newcommand{\Varha}[1]{\text{Var}_{H_1}\left( #1 \right)}
\newcommand{\Covha}[1]{\text{Cov}_{H_1}\left( #1 \right)}
\renewcommand{\exp}[1]{\mathrm{exp}\left(#1\right)}
\newcommand{\ind}[1]{\mathbbm{1}\left\{ #1 \right\}}
\newcommand{\given}{\;\middle|\;}

\usepackage{mathtools}
\DeclarePairedDelimiter\ceil{\lceil}{\rceil}

\usepackage{mathtools}
\DeclarePairedDelimiter\abs{\lvert}{\rvert}

\newcommand{\bigO}{\ensuremath{\mathcal{O}}}
\newcommand{\bigOp}[1][\mathbb{P}]{\ensuremath{\bigO_{\scriptscriptstyle{}#1}}}
\newcommand{\smallO}{\ensuremath{o}}
\newcommand{\smallOp}[1][\mathbb{P}]{\ensuremath{\smallO_{\scriptscriptstyle{}#1}}}
\newcommand{\smallOmega}{\ensuremath{\omega}}

\renewcommand{\d}{{\rm d}}

\usepackage{mathtools}
\DeclarePairedDelimiter\floor{\lfloor}{\rfloor}

\usepackage{anyfontsize}

\newcommand\numberthis{\addtocounter{equation}{1}\tag{\theequation}}
\newcommand{\myparagraph}[1]{\textit{#1}.}

\endlocaldefs

\begin{document}

\begin{frontmatter}
\title{Sparse Anomaly Detection Across Referentials: A Rank-Based Higher Criticism Approach}
\runtitle{Sparse Anomaly Detection Across Referentials}

\begin{aug}
%%%%%%%%%%%%%%%%%%%%%%%%%%%%%%%%%%%%%%%%%%%%%%%
%% Only one address is permitted per author. %%
%% Only division, organization and e-mail is %%
%% included in the address.                  %%
%% Additional information can be included in %%
\author[A]{\fnms{Ivo V.}~\snm{Stoepker}\ead[label=e1]{i.v.stoepker@tue.nl}\orcid{0000-0001-9579-7259}},
\author[A]{\fnms{Rui M.}~\snm{Castro}\ead[label=e2]{rmcastro@tue.nl}\orcid{0000-0003-4398-0718}}
\and
\author[B]{\fnms{Ery}~\snm{Arias-Castro}\ead[label=e3]{eariascastro@ucsd.edu}}
%%%%%%%%%%%%%%%%%%%%%%%%%%%%%%%%%%%%%%%%%%%%%%
%% Addresses                                %%
%%%%%%%%%%%%%%%%%%%%%%%%%%%%%%%%%%%%%%%%%%%%%%
\address[A]{Department of Mathematics and Computer Science, Technische Universiteit Eindhoven, Eindhoven, The Netherlands\printead[presep={,\ }]{e1,e2}}

\address[B]{Department of Mathematics and Halıcıoğlu Data Science Institute, University of California, San Diego, La Jolla, CA, USA\printead[presep={,\ }]{e3}}
\end{aug}

\begin{abstract}
Detecting anomalies in large sets of observations is crucial in various applications, such as epidemiological studies, gene expression studies, and systems monitoring. We consider settings where the units of interest result in multiple independent observations from potentially distinct referentials. Scan statistics and related methods are commonly used in such settings, but rely on stringent modeling assumptions for proper calibration. We instead propose a rank-based variant of the higher criticism statistic that only requires independent observations originating from ordered spaces. We show under what conditions the resulting methodology is able to detect the presence of anomalies. These conditions are stated in a general, non-parametric manner, and depend solely on the probabilities of anomalous observations exceeding nominal observations. The analysis requires a refined understanding of the distribution of the ranks under the presence of anomalies, and in particular of the rank-induced dependencies. The methodology is robust against heavy-tailed distributions through the use of ranks. Within the exponential family and a family of convolutional models, we analytically quantify the asymptotic performance of our methodology and the performance of the oracle, and show the difference is small for many common models. Simulations confirm these results. We show the applicability of the methodology through an analysis of quality control data of a pharmaceutical manufacturing process.
\end{abstract}



\end{frontmatter}
%%%%%%%%%%%%%%%%%%%%%%%%%%%%%%%%%%%%%%%%%%%%%%
%% Please use \tableofcontents for articles %%
%% with 50 pages and more                   %%
%%%%%%%%%%%%%%%%%%%%%%%%%%%%%%%%%%%%%%%%%%%%%%
\tableofcontents

%%%%%%%%%%%%%%%%%%%%%%%%%%%%%%%%%%%%%%%%%%%%%%
%%%% Main text entry area:				    %%
%%%%%%%%%%%%%%%%%%%%%%%%%%%%%%%%%%%%%%%%%%%%%%
\section{Introduction}\label{sec:intro}
We study anomaly detection among a large number of subjects, each giving rise to multiple observations. The observations within each subject are assumed independent, but they can potentially originate from different \textit{referentials} --- i.e.~from a different distribution. For example, one entry within the subject may record some continuous measurement, while another entry may record a discrete count variable. Importantly, we assume the observations within each referential take values in an ordered space (for example, real numbers or ordinal categories). This allows observation ranks to be defined. Finally, we are primarily interested in a \emph{sparse} setting where only a small fraction of the subjects are potentially affected by anomalies.

Formally, we consider monitoring $n$ subjects, where each subject is associated with $t$ independent measurements, not necessarily within the same referential. Specifically, letting $[n] \equiv \{1, \dots, n\}$ and $[t] \equiv \{1, \dots, t\}$, we observe $\mathbf{X}\equiv(X_{ij})_{i\in[n], j\in[t]}$, where $X_{ij}$ is the value from subject $i \in [n]$ at measurement $j \in [t]$.

We cast our anomaly detection problem as a hypothesis testing problem. Under the null hypothesis, all variables $X_{ij}$ are independent, with a distribution function that depends only on the referential index $j$, denoted by $F_{0j}$ and assumed unknown. Under the alternative hypothesis a majority of observations follow these common null distributions $F_{0j}$ as well. However, there exists a small subset of anomalous subjects $\cS \subset [n]$ such that, for each $i\in\cS$ and each $j\in[t]$ the distribution of $X_{ij}$, denoted by $F_{ij}$ and also assumed unknown, differs from the null distribution $F_{0j}$. The subset $\cS$ is unknown and, importantly, is not assumed to have a particular structure. Our interest in sparse alternatives corresponds to the implicit assumption that $\abs{\cS}$ is much smaller than $n$. Succinctly, the hypothesis testing problem we consider in this paper is therefore:
\begin{align*}
H_0:\qquad &\forall_{i\in[n], j\in[t]} \ X_{ij} \stackrel{\text{indep.}}{\sim} F_{0j} \ , \numberthis \label{hyp:general}\\
H_1:\qquad & \exists_{\cS \subset [n]}\ : \forall_{i\in\cS,j\in[t]} \ X_{ij}\stackrel{\text{indep.}}{\sim} F_{ij} \ , \text{ and } \forall_{i\notin\cS,j\in[t]} \ X_{ij}\stackrel{\text{indep.}}{\sim} F_{0j}\ ,
\end{align*}
where none of the distributions above are assumed known.

Our main contribution is a novel rank-based methodology, based on a variant of the higher criticism statistic \citep{Donoho2004}, requiring no distributional knowledge, but which is able to asymptotically distinguish between the two hypotheses above provided the anomalous distributions $F_{ij}$ are sufficiently different from the null distributions $F_{0j}$. These results are formalized in Section~\ref{sec:res_analytic}. Although we focus on the scenario where anomalous observations are \emph{typically} larger than nominal ones, it is not assumed that the distribution $F_{ij}$ stochastically dominates the null distribution $F_{0j}$. 

Although not directly apparent, the setting we consider can be seen as a high-dimensional extension of the scenario typically envisioned when applying the classical Friedman test \citep{Friedman1937}. To see this, note that the null hypothesis in~\eqref{hyp:general} is composite and we do not make any assumptions on the null distributions $F_{0j}$, such that we can particularize~\eqref{hyp:general} by restricting
\begin{equation}\label{eq:friedman}
X_{ij} = \mu_i + A_j + Z_{ij} \ ,
\end{equation}
with the $\mu_i$ representing unknown constants, and $A_j\sim G_j$ independently from unknown distribution $G_j$ (they need not be identically distributed) and the $Z_{ij}$ being i.i.d.~with unknown distribution $F_0$ respectively. Then testing:
\begin{equation}\label{hyp:anova}
H_0: \text{all $\mu_i$ are equal,} \quad \text{ vs. } \quad H_1 : \text{not all $\mu_i$ are equal,}
\end{equation}
is the setting envisioned for the Friedman test.\footnote{The Friedman test is classically applied to settings where treatments are applied to subjects under study. In the terminology we have introduced, our subjects correspond to these classical treatments, and our referentials correspond to these classical subjects.} Further restricting $G_j$ and $F_0$ to zero-mean normal distributions particularizes the setting to a high-dimensional extension of ANOVA, resembling the classical normal location model considered in \cite{Donoho2004} where the higher criticism statistic was introduced.

Since we are not assuming any particular set of null distributions, our formulation in~\eqref{hyp:general} can be applied in a variety of settings without the need to accurately model the null distributions~$F_{0j}$. This is especially advantageous in settings where conclusions may change depending on the models for $F_{0j}$ assumed. Furthermore, as observations within each subject may originate from different referentials, various data sources on the subjects under study may be combined for the purposes of anomaly detection if each source provides independent ordinal data.

Applications can then be found in a variety of fields. We briefly discuss some possibilities below.
\begin{itemize}
\item In industrial settings, one may observe a large number of production endpoints (which could correspond to machines, or individual components thereof) for which various failure metrics are observed. Under nominal conditions, these can be modeled as independent samples from metric-specific unknown null distributions, whereas under anomalous conditions these metrics are raised and should be detected. We present such an example in Section~\ref{sec:application} based on pharmaceutical data from \cite{Zagar2022}. 
\item In cosmological settings, where the aim is to detect signals against the cosmic microwave background noise, a parametric application of the higher criticism statistic is considered in \cite{Jin2005}. Our non-parametric approach may be useful in these contexts, as it can be employed without the need to explicitly model the noise and signal distributions. In Section~\ref{sec:exp-family} we discuss a convolution model as a particularization of~\eqref{hyp:general} which is commonly hypothesized in this context, and show the (asymptotic) power difference with respect to a parametric approach. Depending on the signal distribution, our rank-based approach may even be more powerful than an approach based on subject means.
\item The above formulation may also be appropriate in some gene expression studies. For example, consider studies like \cite{Gordon2002, Wang2005}, where a small subset of genes is found to have predictive power with respect to presence of a disease. It may be relevant to test if there exists a (small) subset of patients for which the gene expressions within this small subset of genes are different compared to the complete (diseased) patient body; this heterogeneity among the diseased group may be important in subsequent studies. The formulation in~\eqref{hyp:general} then allows for the expression levels to have a gene-specific null distribution. 
\item In a similar vein, in copy number variation analysis our methodology may be used to detect rare segment variants without the need to assume a parametric model on the observations, and allowing for either a probe-specific or subject-specific null distribution. For such settings, parametric methods on a particularization of~\eqref{hyp:general} using the higher criticism statistic have been discussed in \cite{Jeng2013} for example.
\end{itemize}

Our methodology requires the observations to originate from an ordered space (within their referential) such that observation ranks can be defined. This precludes directly incorporating vector-valued and categorical variables. Nevertheless, depending on the specific context, one may appropriately transform such non-ordered observations to an ordered space prior to application of our methodology. For vector-valued observations this could amount to either using vector coordinates as separate measurements, or computing an adequate (ordered) summary statistic for each vector-valued observation --- for example a suitable vector norm. For categorical variables various approaches exist which transform categorical variables to meaningful numerical distance measures or outlying scores. See \cite{Taha2020} and Chapter 8 of \cite{Aggarwal2017} and references therein.

\myparagraph{Contributions} In this work we propose an exact test for the null hypothesis~\eqref{hyp:general}. Our procedure replaces the observations by their ranks within each referential $F_{0j}$, and relies on a variant of the higher criticism statistic for inference. In this way, it extends the classic Friedman test \citep{Friedman1937} to a high-dimensional, sparse setting. Note that the proposed procedure itself does not rely on asymptotic considerations. As far as we are aware, there are no other works that propose a distribution-agnostic methodology suitable for the high-dimensional setting~\eqref{hyp:general} with finite-sample calibration guarantees. Due to the use of ranks, calibration of the statistic can be done through Monte-Carlo simulation, and since the calibration distribution depends only on the sample size, critical values can be tabulated in advance. 

We draw an unexpected parallel with parametric work, where we show that given sufficiently many referentials, the problem~\eqref{hyp:general} viewed through the lens of ranking boils down to a heteroskedastic normal location model. The analysis requires a refined understanding of the distribution of the ranks under the presence of anomalies. This calls for a very sharp characterization going significantly further than those in related literature. Furthermore, the analysis requires a careful treatment of rank-induced dependencies.

Our analysis is not restricted to identical and weak anomalies. This restriction is natural in parametric settings where results can be naturally extrapolated to stronger signals. However, in our rank-based setting, such extrapolation is no longer straightforward. Furthermore, our general statements elucidate how anomalies of varying magnitudes and heteroskedasticity should be weighed in order to determine asymptotic detection difficulty. This insight allows us to show asymptotic power guarantees (in the form of a detection boundary) in an elegant nonparamatric way: asymptotically, this only depends on the probability of observing an anomalous observation larger than a null observation, and the probability of observing an anomalous observation simultaneously larger than two independent null observations. These quantities are easily particularized to specific parametric settings to obtain the corresponding detection boundaries. We compare our methodology with an oracle test in a one-parameter exponential family and a convolutional setting, and characterize the performance loss due to the use of ranks. Finally, through simulation we confirm this theoretical loss of power closely matches the loss in finite samples.

\myparagraph{Related work} Our work exists on the boundary of two branches of literature; on the one hand, our approach connects well with the high-dimensional signal detection literature, particularly with works that use the higher criticism statistic in nonparametric settings. On the other hand, the rank transformation considered and subsequent scope of application closely connects with the classical Friedman test \citep{Friedman1937} when viewed as a high-dimensional sparse extension of this two-way nonparametric ANOVA methodology. 

As far as we are aware, the unsuitability of the Friedman test in high-dimensional sparse settings has not been discussed in previous literature, and approaches originating from this angle have not been pursued. Nevertheless, the usage of ranks as a means to construct distribution-agnostic tests has a long history; for a discussion see \cite{Lehmann1975, MR758442} for classical references.

In the statistics literature, previous authors \citep{Ingster2010, Arias-Castro2011a, mukherjee2015hypothesis} have considered high-dimensional testing problems motivated from ANOVA and regression models, but with an important distinction; the authors assume normally distributed residuals and consider a large number of covariates (much larger than the number of subjects under study) where they assume the effects of the covariates are sparse (i.e.~nonzero for a sparse subset of the covariates), without a main effect that is potentially nonzero for a sparse subset of the subjects. When we particularize our general setting~\eqref{hyp:general} to a repeated measurements ANOVA~\eqref{eq:friedman}, we are not assuming any sparsity over the block effects $A_j$, but a potentially sparse subsets of subjects may have a nonzero main effect $\mu_i$. And, importantly, we do not require the residuals to be normal. 
In \cite{Wu2014} the authors investigate an extension of the aforementioned sparse linear regression models with dependencies that naturally arise in genomics settings, assuming normality of residuals. However, in the numerics section calibration by permutation is mentioned, which hints at the applicability of such an approach in a distribution-agnostic manner with finite-sample calibration guarantees.
Similar comments apply to related work in \cite{liu2018powerful}.
In \cite{Sabatti2009} this approach is considered in an applied setting, and \cite{Donoho2015} also discusses applications of the higher criticism statistic in genomics. Here, permutations are mentioned as a means to calibrate tests in finite samples. In \cite{Liu2019b} an approximation to such permutation-based $p$-values is proposed which has similar accuracy in large samples, but is computationally advantageous. 

In \cite{Butucea2013} even sparser settings related to such regression problems are considered. The authors consider a setting where one observes data akin to our setting (a matrix $\bX$ consisting of $n$ subjects with $t$ measurements) where only a sparse subset of the measurements is affected for a sparse subset of the subjects, assuming residuals originate from some known exponential family. This setting then boils down to detecting a submatrix for which observations have nonzero mean. While we do not principally require all referentials to exhibit anomalous behavior and our methodology may still have some power in this case (we formalize this in Section~\ref{sec:analytic_calpow}), when faced with such sparsity one naturally requires the set of nonzero effects to exhibit strong anomalous behavior in order for detection to be successful. In \cite{Ma2015} a similar problem is discussed with an eye for computational costs. A distribution-agnostic approach for this problem based on scan statistics is proposed in \cite{Arias-Castro2017a}. 

We continue with a discussion relating to our work in the context of the high-dimensional distribution-agnostic statistical signal detection literature. Many works exist in this domain for parametric settings, with a particular interest for higher criticism-based methodologies initiated by \cite{Donoho2004}. However, there are only a few works that propose a distribution-agnostic approach; we will discuss these first.

In \cite{stoepker2021} a distribution-agnostic approach is proposed for a particularization of the setting in~\eqref{hyp:general} in which the nominal distributions are identical for each referential, i.e.~$F_{0j} = F_0$. The higher criticism statistic on the raw data is used for inference, and permutations are used both in the definition of the statistic as well as for calibration, leading to an exact test. The authors prove the permutation test has risk tending to zero under the same signal strengths as its parametric counterpart. Recent works have developed general frameworks for studying the power of permutation tests, proving minimax rate optimality in various settings --- see \cite{Kim2022} and \cite{Dobriban2022}. In particular, examples in the latter include the linear regression setting of \cite{Ingster2010} where the noise distribution is assumed symmetrical, which is then used to propose a test based on sign-flips. In \cite{Delaigle2011} the same particularization of~\eqref{hyp:general} is considered and another distribution-agnostic approach is proposed. The authors use studentized subject means for inference: their distribution is estimated through bootstrapping, subsequently allowing a higher criticism statistic to be computed. Asymptotic calibration guarantees are given.

Another setting close to~\eqref{hyp:general} is considered in \cite{arias-castro2018a}, where the authors assume $t=1$ and the anomalies are assumed to exhibit some structure. A distribution-agnostic approach is constructed through a scan statistic, either on raw data and their ranks, and calibration is done by permutation.

In \cite{Arias-Castro2017b} a similar setting to~\eqref{hyp:general} is considered with $t=1$. The null distribution is assumed to be symmetric, but apart from that the methodology is distribution-agnostic. A sign-based test tailored to sparse alternatives is proposed.
In \cite{Arias-Castro2020} a bivariate setting is considered where the goal is to detect a sparse contamination, which presents itself as positive dependencies between the bivariate observations. The authors adapt Spearman's rank correlation to make it powerful against sparse alternatives.

In \cite{Zou2017} a distribution-agnostic methodology is proposed for a the particularization of~\eqref{hyp:general} where the distributions are identical for each for each referential, i.e.~$F_{0j} = F_0$. However, instead of detection, the authors consider identifying the set of anomalous subjects. In \cite{delaigle2009} a setting quite different from ours is considered, where observations with unspecified distribution are tested to have mean equal to the mean of auxiliary observations. The higher criticism is used, and calibration guarantees are given asymptotically. In \cite{Kurt2020}, yet another different setting is considered in change-point detection with unknown distributions, where the subjects are observed in real-time and a period is guaranteed to be free of anomalies. The statistics used are shown to be bounded under the null, and power properties are shown empirically.

The hypothesis testing problem~\eqref{hyp:general} has been discussed for various parametric settings for $t=1$, initially for normal settings by \cite{MR1456646} (see also Chapter 8 in \cite{Ingster2003}) and \cite{Donoho2004}. For our setting, the parametric heteroskedastic normal location model considered by \cite{TonyCai2011} --- later extended in \cite{Cai2014} --- is particularly relevant, as our detection boundary is shown to have the same form as in that work. In this model the variance of the anomalous observations constitutes an unknown nuisance parameter. 

In \cite{Laurent2012} a more extensive model is considered where anomalies may have varying magnitudes and observations may have subject-specific variance, but this variance structure is assumed known. 
The authors study the minimax separation rate of the problem, which they define as the smallest $L^2$ distance of the anomalous magnitudes to zero such that a test exists with prescribed power (under sparsity constraints). The authors provide upper and lower bounds of this separation rate, depending on the structure of the problem.
The same setting is considered in \cite{Chhor2024}, but the analysis extended for general $L^p$ norms, and the provided upper and lower bounds of the separation rate are always matching.

Heteroskedasticity has also been considered in a different parametric setting in \cite{Arias-Castro2018b}, but here the anomalies are assumed to have the same mean as the null observations, instead allowing the variance to change with the sample size. As we do not make the latter restriction, the detection boundary related to that in \cite{TonyCai2011} arises. We assume no dependencies between subjects, unlike \cite{Hall2008, HJ09} that consider the effect of dependencies on the higher criticism in a parametric setting.

\myparagraph{Organization} Section~\ref{sec:intro} introduces the problem and its relationship with common settings encountered in practice. 
Section~\ref{sec:meth} proposes a novel rank-based higher criticism test. 
Section~\ref{sec:res_analytic} presents some theory developed for our methodology. 
Section~\ref{sec:sim-settings} describes some numerical experiments carried out to assess the finite-sample performance of the methodology, elucidating the power characteristics of the proposed test and confirming our theoretical results. 
We showcase the applicability of the proposed methodology on a real dataset in Section~\ref{sec:application}. 
A discussion in Section~\ref{sec:discussion} closes the main body of the paper.
Proofs are deferred to Sections~\ref{sec:proof-th1} and \ref{sec:particularizations}. 
Some auxiliary technical results which shed further light on remarks made in the main text, and further derivations to make the manuscript self-contained, are deferred to Appendix~\ref{app:math}. 
Supplemental simulation results are included in Appendix~\ref{app:supp-sim}. Supplemental results to the application setting considered are included in Appendix~\ref{app:supp-app}.

\myparagraph{Notation} Throughout the paper we use standard asymptotic notation. Let $n\to\infty$, then $a_n = \bigO(b_n)$ when $| a_n / b_n |$ is bounded, $a_n = \smallO(b_n)$ when $a_n / b_n \to 0$, and $a_n=\smallOmega( b_n)$ when $b_n = \smallO(a_n)$. We also use probabilistic versions: $a_n = \bigOp(b_n)$ when $| a_n / b_n |$ is stochastically bounded\footnote{That is, for any $\varepsilon>0$ there is a $C_\varepsilon$ and $n_\varepsilon$ such that $\forall n>n_\varepsilon\ \P{| a_n / b_n |>C_\varepsilon}< \varepsilon$.} and $a_n = \smallOp(b_n)$ when $a_n / b_n$ converges to $0$ in probability. Unless otherwise stated, we consider asymptotic behavior with respect to $n \to \infty$. We use the subscripts $H_0$ and $H_1$ when considering an expectation and a probability to explicitly indicate whether we are considering the null or alternative hypothesis.

\section{Methodology}\label{sec:meth}

\subsection{Ranking observations}\label{sec:ranking}
Our methodology replaces the observations by their ranks within their respective referential, in the same fashion as the ranking procedure used in the classical Friedman test \citep{Friedman1937}. For the purpose of our analysis ties are broken at random, which has the advantage that one can regard the subsequent testing procedure as a test calibrated by permutations sampled uniformly over the set of all permutations.

To allow us to state our results in a general way (and not focus on the particular case where the involved distributions are continuous) we must formalize the tie breaking procedure, which we do here by adding a small independent perturbation to the original observations, prior to ranking:
\begin{equation}\label{eq:def_tiebreak}
X_{ij}^{(\delta)} \equiv X_{ij} + \delta W_{ij} \ , 
\end{equation}
where $W_{ij} \overset{\text{i.i.d.}}{\sim} \text{Uniform}([-1,1])$. With probability one there are no ties in the transformed observations, and when $\delta$ is sufficiently small the original ordering is preserved among unique observations (in particular taking $0 < \delta < \tfrac{1}{2}\min_{i,k,j}\{\abs{X_{ij} - X_{kj}} : \abs{X_{ij} - X_{kj}} > 0\}$ suffices). We can now compute the rank of observation $X_{ij}$ by
\begin{equation}\label{eq:def_ranks}
R_{ij} \equiv \lim_{\delta\to 0} R_{ij}(X_{1j}^{(\delta)},\dots,X_{nj}^{(\delta)}) = \lim_{\delta\to 0} \sum_{k\in[n]}\ind{X_{ij}^{(\delta)} \geq X_{kj}^{(\delta)}} \ .
\end{equation}
We denote the resulting set of ranks by $\bR \equiv \bR(\bX) = (R_{ij})_{i\in[n], j\in[t]}$.

\subsection{Rank-based higher criticism test}\label{sec:hc_rank}

We adapt a version of the higher criticism statistic which was initially proposed for analytic purposes in \cite{Donoho2004}. This version is also convenient from a practical standpoint \citep{stoepker2021, Arias-Castro2011a, Wu2014, Arias-Castro2020}.

Define the rank mean of subject $i$ as:
\[
Y_i(\bR) \equiv \frac{1}{t}\sum_{j\in[t]} R_{ij} \ .
\]
The version of the higher criticism that we consider makes use of the following quantity:
\begin{equation}\label{eq:def_Nq}
N_q(\bR) \equiv \sum_{i\in[n]} \ind{ \frac{Y_i(\bR) - \Rb}{\sigma_R} \geq \sqrt{\frac{2q\log(n)}{t}}}\ ,
\end{equation}
which denotes the number of subject centered rank means exceeding a threshold based on parameter $q > 0$, with 
\begin{equation} \label{eq:def_Rbar_sigmaR}
\Rb\equiv\frac{1}{nt}\sum_{i\in[n]}\sum_{j\in[t]} R_{ij} = \frac{n+1}{2}, \text{ and } \sigma_R^2\equiv\frac{1}{nt} \sum_{i\in[n]}\sum_{j\in[t]} (R_{ij}-\overline R)^2 = \frac{n^2-1}{12} \ .
\end{equation}
Define also
\begin{equation}\label{eq:def_pq}
p_{q} \equiv \Phn{\frac{Y_1(\bR) - \Rb}{\sigma_R} \geq \sqrt{\frac{2q \log(n)}{t}} }  \ ,
\end{equation}
which corresponds to the probability of a subject rank mean to exceed the same threshold as in~\eqref{eq:def_Nq} under the null hypothesis in~\eqref{hyp:general}. Note that this quantity is independent of the data apart from the sample size $(n,t)$ and can be computed by Monte-Carlo simulation. We are now ready to define our higher criticism statistic and test:
\begin{defi}[Rank-based higher criticism test]\label{def:rank_hc_stat}
Define
\begin{equation}\label{eq:def_Vq}
V_q(\bR)\equiv \frac{N_q(\bR) - np_q}{\sqrt{np_q(1-p_q)}}\ ,
\end{equation}
where we take the convention that $0/0=0$, with $N_q(\bR)$ and $p_q$ defined as in~\eqref{eq:def_Nq} and~\eqref{eq:def_pq} respectively. Let
\begin{equation}\label{eq:def_original_grid}
Q_n \equiv \left\{\frac{2}{k_n}, \frac{4}{k_n}, \dots, 2 \right\}\ ,
\end{equation}
for some $k_n \in \bbN$. Then our rank-based higher criticism statistic is defined as
\begin{equation}\label{eq:def_rank_HC}
T(\bR) \equiv \max_{q\in Q_n} V_q(\bR)\ .
\end{equation}
Referring to the hypothesis testing problem~\eqref{hyp:general}, denote the set of observations by $\mathbf{x}$ and the set of ranks of our observations by $\mathbf{r} \equiv \mathbf{R}(\mathbf{x})$.\footnote{As is common, we use lowercase symbols to emphasize these are nonrandom.} Then, our rank-based higher criticism $p$-value is defined as:
\begin{equation}\label{eq:def_rank_HC_p_value}
\cP_{\text{rank-hc}}(\mathbf{r})  \equiv \Phn{ T(\bR) \geq T(\mathbf{r})} \ .
\end{equation}
Our rank-based higher criticism test with significance level $\alpha$ is consequently defined as:
\[
\psi_{\text{rank-hc}}(\mathbf{X}) = \ind{ \cP_{\text{rank-hc}}(\mathbf{R}) \leq \alpha}  \ ,
\]
where $\psi(\mathbf{X}) = 1$ implies rejection of the null hypothesis.
\end{defi}

The maximization over the thresholds is what lends the statistic its adaptivity over a diverse range of anomalous regimes; for moderately sparse and weak signals, lower threshold values are most discriminative between the null and alternative hypothesis, whereas for stronger but sparser signals it is instead more powerful to consider larger thresholds.

Note that the null distribution of the observation ranks in~\eqref{eq:def_pq} and~\eqref{eq:def_rank_HC_p_value} is only a function of the sample size $(n,t)$. Therefore, estimation of the probability $p_q$ and the distribution $T(\bR) \mid H_0$ may be done in advance and tabulated, such that future test executions for same-sized problems can be greatly expedited.

The supplementary code in \cite{CodeSupplement} includes an implementation of the test described in Definition~\ref{def:rank_hc_stat} in R \citep{RCoreTeam2024}.

\begin{rem}\label{rem:discrete_grid}
Given the discussion above, there is seemingly no need for the discrete grid $Q_n$, and one could instead maximize $V_q(\bR)$ over a continuous interval. However, discretization of the grid allows us to prove our analytical results below without the use of advanced tools from empirical process theory. The same strategy has been employed in several papers in the same line of work \citep{stoepker2021, Arias-Castro2011a, Wu2014, Cai2007}. The extension of the grid $Q_n$ in~\eqref{eq:def_original_grid} up until the value 2 may be unexpected given prior literature, but this extension is included as a convenient way to asymptotically deal with highly anomalous observations, ultimately allowing the statement of our main results in a general way.
\end{rem}

\section{Analytic results}\label{sec:res_analytic}

\subsection{Normal location model}
We pause to discuss results of the parametric higher criticism in a particularization of~\eqref{hyp:general}. In the next section, when we show analytic results for the test in Definition~\ref{def:rank_hc_stat} when applied to the general hypothesis test~\eqref{hyp:general}, we will then be able to connect our results back to this restricted setting.

Our methodology uses a variant of the higher criticism statistic introduced by \cite{Donoho2004}. The original motivation for the statistic is the detection, within a sample of univariate independent observations hypothesized to come from the same null distribution, of a small subset of observations with a different distribution. Consider the following particularization of the general hypothesis test~\eqref{hyp:general} for $t = 1$ which is commonly referred to as the (heteroskedastic) normal location model:
\begin{align*}
H_0:\qquad &\forall_{i\in[n]} \quad X_{i} \stackrel{\text{i.i.d.}}{\sim} \cN(0,1) \ , \\
H_1:\qquad &\exists_{\cS \subset [n]} : \forall_{i\in\cS} \quad X_{i}\stackrel{\text{i.i.d.}}{\sim}\cN(a,\sigma^2), \text{ and } \forall_{i\notin\cS} \quad X_{i} \stackrel{\text{i.i.d.}}{\sim} \cN(0,1)\ . \numberthis\label{hyp:norm}
\end{align*}
It has been shown that the higher criticism statistic can be used to construct an ``optimal'' test for the above hypothesis test without requiring knowledge of the parameters $\abs{\cS}$, $a$ and $\sigma$; first by \cite{Donoho2004} for fixed and known $\sigma = 1$ and generalized by \cite{TonyCai2011} for arbitrary unknown $\sigma > 0$. This optimality statement has an asymptotic nature, and is done through a so-called detection boundary which comes into play as follows; if $a$ is parameterized using a constant $r$ as
\begin{equation}\label{eq:mu_donoho}
a \equiv \sqrt{2r\log(n)} \ ,
\end{equation}
and $\abs{\cS}$ is parameterized using a constant $\beta\in(1/2,1)$ as
\begin{equation}\label{eq:n-anomalies}
\abs{\cS} \equiv \lceil n^{1-\beta}\rceil \ ,
\end{equation}
then the detection boundary, which we denote by $\rho(\beta,\sigma)$, captures for a given sparsity level~$\beta$ and heteroskedasticity $\sigma$ the minimum size of $r$ required for the hypotheses in~\eqref{hyp:norm} to be asymptotically distinguishable. For the hypotheses in~\eqref{hyp:norm}, the detection boundary $\rho(\beta,\sigma)$ is given by:
\begin{equation}\label{eq:def_detection_boundary}
\rho(\beta,\sigma) \equiv \begin{cases}
\begin{cases} 
(2-\sigma^2)(\beta - \frac{1}{2}) \ , & 1/2 < \beta \leq 1-\sigma^2/4\ , \\ 
(1 - \sigma\sqrt{1-\beta})^2 \ ,		&  1-\sigma^2/4 < \beta < 1 \ ,
\end{cases} & 0\leq\sigma<\sqrt{2} \ ,\vspace{0.5cm}\\
\begin{cases}
0 \ , 								& 1/2 < \beta \leq 1-1/\sigma^2 \ ,\\
(1 - \sigma\sqrt{1-\beta})^2 \ ,		& 1-1/\sigma^2 < \beta < 1 \ ,
\end{cases} & \sigma\geq\sqrt{2} \ .
\end{cases}
\end{equation}
A visualization of the detection boundary $\rho(\beta,\sigma)$ is given in Figure~\ref{fig:detection_boundary}. 

\begin{figure}[htb]
\centering
\begin{subfigure}{0.7\textwidth}
\hspace{0.5cm}\includegraphics[width=\linewidth]{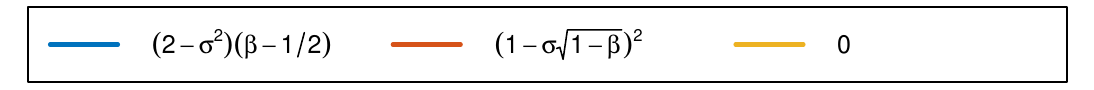}\vspace{0.3cm}
\end{subfigure}

\begin{subfigure}{0.32\textwidth}
\includegraphics[width=\linewidth]{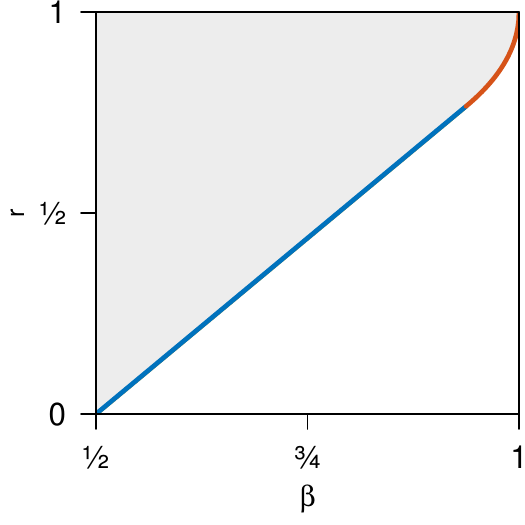}
\caption{Boundary for $\sigma = 0.5$.}
\end{subfigure}
\begin{subfigure}{0.32\textwidth}
\includegraphics[width=\linewidth]{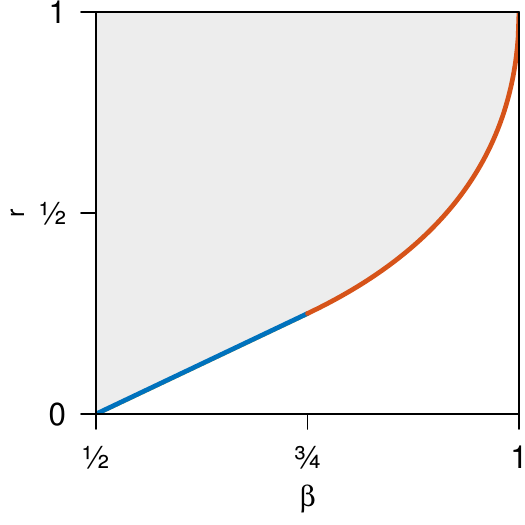}
\caption{Boundary for $\sigma = 1$.}
\end{subfigure}
\begin{subfigure}{0.32\textwidth}
\includegraphics[width=\linewidth]{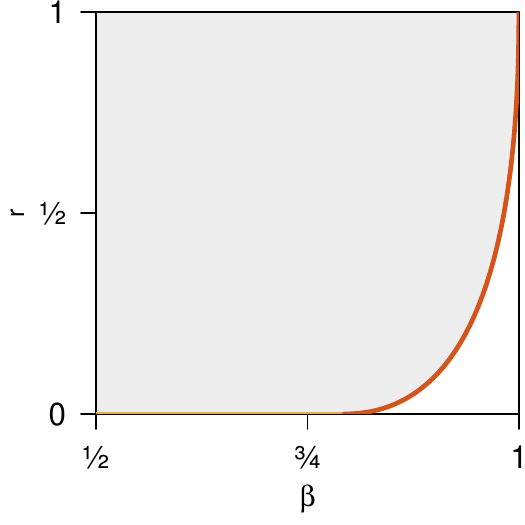}
\caption{Boundary for $\sigma = 2.2$.}
\end{subfigure}
\caption{The detection boundary $\rho(\beta,\sigma)$ from \eqref{eq:def_detection_boundary} for different values of $\sigma$.}\label{fig:detection_boundary}
\end{figure}

In \cite{MR1456646} it is shown that for the setting~\eqref{hyp:norm} with known $\sigma=1$, any test is powerless, meaning the sum of type I and type II errors converges to 1 when $r < \rho(\beta,\sigma)$. Conversely, in \cite{Donoho2004} it is shown that for known $\sigma=1$ the test based on the higher criticism statistic has its sum of type I and type II errors converge to zero when $r > \rho(\beta,\sigma)$; thus, it is ``optimal'' in the sense that any asymptotically detectable signal is detectable by the higher criticism statistic as well --- apart from signals existing on the boundary $r=\rho(\beta,\sigma)$, for which behavior is hard to describe and are typically excluded in this context. In \cite{TonyCai2011} the preceding statements are generalized for unknown $\sigma > 0$.

Note that the function $\rho(\beta,\sigma)$ is decreasing in $\sigma$ when $\beta$ is held constant; in other words, larger heteroskedasticity decreases the required magnitude for detection. In the next section it will be made clear how our setting relates to this heteroskedastic normal location model.

\subsection{Calibration and power guarantees}\label{sec:analytic_calpow}
We proceed with analytic results for the test in Definition~\ref{def:rank_hc_stat} in the general hypothesis testing scenario~\eqref{hyp:general}. Calibration results are stated regardless of the sample size. On the other hand, our power results are of an asymptotic nature; we consider the setting where the number of subjects~$n$ grows, while the number of anomalous subjects~$\abs{\cS}$ grows at a slower pace. Note that in general the number of observations for each subject $t$ and the anomalous distributions $F_{ij}$ also change with increasing $n$. However, to avoid cluttering the presentation, this dependency is not explicitly stated apart from the quantities in the upcoming Definition~\ref{def:rank_trans_char}, which will be instrumental to state our power guarantees.

Our methodology is based strictly on the observation ranks (as defined in Section~\ref{sec:ranking}, after tie-breaking). Note that this ranking may distort evidence against the null; to understand which characteristics of the anomalous distributions $F_{ij}$ may be relevant for successful detection, it is useful to consider a transformation of our observations. Consider the transformation:
\begin{align*}
U_{ij} &\equiv \lim_{\delta \to 0} F_{0j}^{(\delta)}(X_{ij}^{(\delta)}) = F_{0j}(X_{ij}) + \left(\tfrac{W_{ij}-1}{2}\right)\bigg(F_{0j}(X_{ij}) - F_{0j}^{-}(X_{ij}) \bigg)\ ,  \numberthis\label{eq:def_Uij}
\end{align*}
where $F_{0j}^{(\delta)}$ is the convolution of $F_{0j}$ and the $\delta$-dilated distribution of $W_{ij}$ resulting from the tie-breaking transformation~\eqref{eq:def_tiebreak}, and $F_{0j}^-(x) = \lim_{t\uparrow x} F_{0j}(t)$. Note that for null observations (i.e.~for $i\not\in\cS$) we have that $U_{ij}$ is uniform on $[0,1]$. For continuous observations the tie-breaking scheme is irrelevant, and the above definition reduces to $U_{ij} = F_{0j}(X_{ij})$. Through this lens, the degree of \emph{non-uniformity} of the anomalous observations after the transformation~\eqref{eq:def_Uij} dictates the testing difficulty.

Perhaps unexpectedly, given a sufficient amount of referentials, the \textit{only} characteristics of the distribution of $U_{ij}$ (for $i\in\cS$) that determine asymptotic power are its mean and variance. A similar phenomenon has been observed in \cite{stoepker2021} in a different setting as well. Essentially, given sufficient averaging of the ranks over the referentials, the subject means are sufficiently normal at a specific point in the tails, such that the asymptotic detection boundary matches the boundary if $Y_i(\bR)$ were \textit{exactly} normal. In this normal regime, the only relevant quantities for detection are the signal mean and heteroskedasticity. At this point we can connect our work with that of \cite{TonyCai2011} on the heteroskedastic normal location model~\eqref{hyp:norm}. While this is initially not obvious, when the nonparametric hypothesis testing problem~\eqref{hyp:general} is viewed through the lens of ranking, given sufficient averaging over the subjects, its detection boundary asymptotically exhibits the same structure as the heteroskedastic normal location model~\eqref{hyp:norm} considered in \cite{TonyCai2011}. This insight has serious implications: anomalies that manifest themselves in the first and second moment of $U_{ij}$ play an asymptotic role, but any other anomalous behavior is asymptotically irrelevant!

We highlight the importance of the mean and variance of $U_{ij}$ with the following definition:
\begin{defi}[Anomalous characteristics after rank-transform]\label{def:rank_trans_char}
Let $U$ be a continuous uniform random variable supported on $[0,1]$. Through the lens of ranking, the anomalous signal \emph{magnitude} of subject $i$ is given by:
\begin{equation}\label{eq:def_mu}
\mu_{n,i} \equiv \frac{1}{t}\sum_{j\in[t]}\frac{\E{U_{ij}} - \E{U}}{\sqrt{\Var{U}}} =  2\sqrt{3}\left( \frac{1}{t}\sum_{j\in[t]}\E{U_{ij}} - \frac{1}{2} \right)\ ,
\end{equation}
and the anomalous signal \emph{heteroskedasticity} of subject $i$ is given by:
\begin{equation}\label{eq:def_var}
\sigma^2_{n,i} \equiv \frac{1}{t}\sum_{j\in[t]}\frac{\Var{U_{ij}}}{\Var{U}} = \frac{12}{t}\sum_{j\in[t]}\Var{U_{ij}} \ .
\end{equation}
\end{defi}

\begin{rem}
The quantities of Definition~\ref{def:rank_trans_char} relate to simple probabilities of ``correctly ranking'' anomalous observations with higher ranks than nominal observations, which we further detail in Appendix~\ref{app:moments-u}.
\end{rem}

In stating our analytic results, it will be critical to distinguish between ``weak'' anomalies, ``strong'' anomalies, and a category of anomalies that exists on the boundary between them. Concretely, letting $\eta \in (0,2]$ be fixed, we consider the following partition of $\cS$ indexing strong, weak, and boundary signals respectively:
\begin{align*}
\cS_{T,n} &\equiv \left\{ i\in\cS : \mu_{n,i} \geq \sqrt{\frac{2(1+\eta)\log(n)}{t}} \right\} \ , \\
\cS_{W,n} &\equiv \left\{ i\in\cS: \mu_{n,i} < \sqrt{\frac{2\log(n)}{t}} \right\} \ , \\
\cS_{B,n} &\equiv \cS \setminus (\cS_{T,n}\cup\cS_{W,n}) \ . \numberthis\label{eq:def_partition_anomalies} 
\end{align*}

Note that the definitions of $\cS_{T,n}$ and $\cS_{B,n}$ depend on $\eta$, but to avoid complicating the notation this dependency is not explicitly stated. The following sequence will play a role in characterizing power properties that hinge on the strength of the boundary signals in $\cS_{B,n}$:
\begin{equation}
\nu_n \equiv n^{\frac{1}{2}\log(n)^{-1/4} +\frac{1}{4k_n} + \frac{1}{2}\sqrt{\frac{\log(n)}{t}}}\big(\sqrt{k_n} + n^{\frac{1}{2}\log^{-1/2}(n)}\big)\ .
\end{equation}
Importantly, under the assumptions of our upcoming Theorem~\ref{th:rank_hc_test}, $\nu_n = n^{o(1)}$ and $\nu_n = \omega(1)$. We are now ready to state our main result:

\begin{theorem}[Performance of permutation-rank-based higher criticism test]\label{th:rank_hc_test}
Referring to the hypothesis testing problem~\eqref{hyp:general}, consider the rank-based higher criticism test $\psi_\text{rank-hc}(\mathbf{X})$ from Definition~\ref{def:rank_hc_stat} for arbitrary $\alpha\in (0,1)$. Then:

\begin{enumerate}[label=(\arabic*)]
\item This test has level at most $\alpha$. \label{en:th1_null}
\item Consider the alternative hypothesis and the parameterizations of the anomalous subjects in Definition~\ref{def:rank_trans_char}, partitioned as in~\eqref{eq:def_partition_anomalies}. Assume $k_n = n^{o(1)}$ and $k_n \to \infty$. Assume $t = \omega(\log(n))$ and $t = o(n)$. The test has power converging to one if either one of the following conditions hold: \label{en:th1_alt}
\begin{enumerate}[label=(\roman*)]
\item $\abs{\cS} = \bigO(\sqrt{n})$ and $\liminf_{n\to\infty} \ \abs{\cS_{T,n}} \geq 1$. \label{en:th1_strong}
\item $\abs{\cS} \leq \sqrt{n}$ and $\abs{\cS_{B,n}} = \omega(\nu_n)$. \label{en:th1_boundary}
\item Assume constants $\beta > 1/2$, $r$ and $\gamma$  exist such that $\abs{\cS} = n^{1-\beta + o(1)}$ and
\begin{equation}\label{eq:def_min_r_gamma}
\liminf_{n\to\infty}\min_{i\in\cS_{W,n}}  \frac{\mu_{n,i}^2 t}{2\log(n)} = r \ , \quad  \liminf_{n\to\infty} \min_{i\in\cS_{W,n}} \sigma^2_{n,i} = \gamma^2 \ ,
\end{equation}
when $r > \rho(\beta,\gamma)$. If $\gamma^2 > 2$ and $\beta\in(1/2,1-1/\gamma^2)$ then $r = 0$ suffices. \label{en:th1_weak}
\end{enumerate}
\end{enumerate}
\end{theorem}

The proof of Theorem~\ref{th:rank_hc_test} is provided in Section~\ref{sec:proof-th1}, consisting of two main parts. The first part requires a characterization of the $p$-value of the test from Definition~\ref{def:rank_hc_stat} --- this is provided in Lemma~\ref{lem:bound-p-value} which resembles approaches from other works~\citep{arias-castro2018a, stoepker2021}. The second part is significantly more challenging. Most notably, we require a tight control on the difference between the tail behavior of the distribution of the rank subject means under the null, and that of the null subjects under the alternative. Obtaining such tight control, given in Lemma~\ref{lem:prob-norm-null-alt-approx}, requires a refined understanding of the dependencies induced by the ranks, which is captured in Lemma~\ref{lem:bound-Nq}.  

It is important to reflect on the assumptions on $k_n$, $t$ and $\abs{\cS_{B,n}}$ in the theorem above. The requirement that $k_n \to \infty$ enables us to assume that, for any $q\in(0,2]$, there exists a gridpoint in $Q_n$ that approximates this point arbitrarily well as $n$ increases. The upper bound on $k_n$ is, however, of a technical nature, and we conjecture it is a proof artifact induced by the use of a union bound (see Remark~\ref{rem:discrete_grid}). Similarly, we conjecture that the required minimal rate for $\abs{\cS_{B,n}}$ under condition~\ref{en:th1_boundary} can be lowered and is a proof artifact induced through similar reasons. We conjecture that a sharper analysis could lower the required rate to $\omega(\sqrt{\log\log(n)})$. The lower bound on $t$ ensures sufficiently Gaussian tail behavior of the subject means. This lower bound requirement is relatively mild, as the ranks have a bounded support, and thus the desired tail behavior can be ensured even for small $t$ (through a Bernstein bound for bounded random variables). The reasons for the upper bound requirement $t = \smallO(n)$ are more subtle; most notably, it ensures that we can accurately relate the distribution of the ranks from nominal subjects (including their dependencies) under the alternative hypothesis with the distribution of ranks under the null hypothesis. Obtaining a sufficiently accurate approximation is far from trivial as one requires a high degree of accuracy in the distribution tails. It is especially challenging for large values of $t$. A restriction to $t = n^{o(1)}$ simplifies some of the arguments, but through careful treatment of higher order terms one can show that for much larger $t$ the theorem holds true as well. Similar dynamics are also encountered and discussed in the permutation setting of \cite{stoepker2021} where imposing the stronger restriction $t = n^{o(1)}$ seems unavoidable without modifying the methodology itself or imposing restrictions on the distributions involved. 

Prior literature on parametric settings highlight only the most intricate setting where $\cS = \cS_{W,n}$, since one can typically extrapolate results to stronger signal regimes, as often the power of such tests is clearly monotonic in the signal magnitude. In our rank-based test such monotonicity is no longer obvious, in part due to the use of ranking. Similar considerations arise in related works~\citep{arias-castro2018a, stoepker2021}. Therefore understanding the power properties of the test for stronger signal regimes requires a more detailed investigation in this setting, and thus have been included here to paint a complete picture.  

\begin{rem} Through the assumption on $t = \omega(\log(n))$ and the parameterization of $\mu_{n,i}$, the theorem implies detection is possible even if $\mu_{n,i} \to 0$, corresponding to the probability of ranking an anomalous observation larger than a nominal observation converging to $1/2$, provided this convergence is sufficiently slow and the constants involved are sufficiently large. Moreover, if the anomalous asymptotic minimal variance is large enough ($\gamma^2 > 2$) and the signal is not too sparse (i.e.~$\beta < 1-1/\gamma^2$), this constant can be 0 as well (corresponding to $\sum_{j\in[t]}\E{U_{ij}} = 0$) and detection is then possible through the heteroskedasticity $\sigma_{n,i}^2$ of the anomalous signals alone.
\end{rem}

\begin{rem}\label{rem:variance-zero}
The case of $\gamma = 0$ may be considered uninteresting when restricting to the normal location model~\eqref{hyp:norm}, but it is undesirable to exclude this case in our general setting~\eqref{hyp:general}; doing so precludes scenarios where anomalous observations have asymptotically vanishing variance (under parameterization in Definition~\eqref{def:rank_trans_char}). An example of a case where both $\mu_{n,i} \to 0$ and $\sigma_{n,i} \to 0$ is included in Appendix~\ref{app:variance-zero}.
\end{rem}

\begin{rem}\label{rem:nonconverging_sigma}
In Theorem~\ref{th:rank_hc_test} the power characterizations for the weakest signals are shown to be dependent on the limit inferior of $\sigma_{n,i}^2$. One might wonder if this characterization is needlessly broad, as $\mu_{n,i} \to 0$ may imply that $\sigma_{n,i}$ has a well-defined limit. In Appendix~\ref{app:nonconverging-sigma} we include an example where $\mu_{n,i} \to 0$, but $\sigma_{n,i}^2$ does not converge.
\end{rem}

\begin{rem}\label{rem:varying-signal}
Although we have aimed to present Theorem~\ref{th:rank_hc_test} as general as possible, we do present guarantees based on minimal signal strength under case~\ref{en:th1_weak}. Since the subjects which attain the minima for $\mu_{n,i}$ and $\sigma_{n,i}$ may not coincide, this may potentially paint a pessimistic picture of the asymptotic power requirements if there is a large amount of heterogeneity in the characteristics of the anomalies. Understanding how the anomalous subject means $\mu_{n,i}$, variances $\sigma_{n,i}$, and their potential multiplicity within the set $\cS_{W,n}$ should be weighed against each other is intricate, and a general statement can only be given in an implicit manner. A more general result is possible, but it would come at the expense of clarity of presentation. We direct the reader to Equation~\eqref{eq:final-req} and Lemma~\ref{lem:prob-alt-char} in the proof of Theorem~\ref{th:rank_hc_test} in Section~\ref{sec:proof-th1}, which together form a less explicit requirement for asymptotic power that does not rely on minimal signal magnitude and heteroskedasticity.
\end{rem}

\subsection{Two special cases}\label{sec:particular}
In this section we compare the asymptotic power properties of our rank-based test with respect to minimax lower bounds, with the aim of quantifying the potential power loss incurred due to the use of ranks. As is standard in the analysis of rank based tests (see, e.g. \citet[Chapter 8]{Lehmann1975}; \citet[Chapter 4.3]{MR758442}; \citet[Chapter 14]{VanDerVaart1998}; \citet[Chapter 2]{Hajek2010}) such a comparison is made through particularization to parametric families. We present numerical results which illustrate power loss, complementary to these analytic results, in Section~\ref{sec:sim}.

The particularization serves two purposes: first, it allows us to effectively connect the results of Theorem~\ref{th:rank_hc_test} to related literature. Second, this focuses the discussion on settings where detailed quantification of the power loss of our rank-based test is meaningful. To see this, note that for settings in which nominal observations arise from a distribution with bounded support and anomalies manifest as a location shift, usage of subject rank mean statistics is in general suboptimal. We therefore focus on settings for which the subject rank means are a meaningful statistic, for which such a comparative analysis is compelling. 

The results in this section refer to parametric families with non-varying referentials, and where potential anomalies all have identical anomalous distributions. This allow us to effectively connect our results to related literature. Moreover, such results also give an upper bound on the power loss due to ranks when varying referentials are induced within the parametric family through a monotone data transformation:
\begin{equation}\label{eq:transform-monotone}
\tilde X_{ij} \equiv f_j(X_{ij}) \ ,
\end{equation}
where $f_j$ is strictly monotone increasing and possibly random. This can be motivated through an invariance argument; conditionally on the data $\bX$, the $p$-value of the rank-based test in Definition~\ref{def:rank_hc_stat} on $(\tilde X_{ij})_{i\in[n],j\in[t]}$ is invariant of the choice of $f_j$ and as such the minimal signal strength required for asymptotic power is identical for any $f_j$. A prototypical example for $f_j(x)$ is the classical Friedman setting~\eqref{eq:friedman} in which $f_j(x) = A_j + x$ for some $A_j\sim G$. Furthermore, if $f_j$ is independent of the data, then the minimal signal strength required for asymptotic power of the oracle test is at least the requirement when $f_j$ is the identity, since $f_j$ does not provide any information to reject the null hypothesis in~\eqref{hyp:general}. Stated differently, the lower bound on the minimal signal strength for the original data are a (possibly loose) lower bound for any such transformation, and the degree of this looseness highly depends on the nature of $f_j$. These two observations together imply that the upcoming results also apply to observations with varying referentials induced as in~\eqref{eq:transform-monotone}.

\subsubsection{The exponential family}\label{sec:exp-family}
Like earlier works in this domain~\citep{arias2011detection, arias-castro2018a, stoepker2021, Arias-Castro2017a} we consider the one-parameter exponential family in natural form as a natural benchmark to study. This family is defined through its density as an exponentially tilted density with respect to a base probability measure. Specifically, with reference to the general hypothesis test in~\eqref{hyp:general}, let $F_{0j}=F_0$ be a probability distribution on the real line (which can be continuous, discrete, or a combination thereof) with all moments finite, and let $\mu_0$ and $\sigma_0^2$ denote its mean and variance respectively. Then, the anomalous distributions $F_{ij} = F_{\theta_i}$ are defined through their density $f_{\theta_i}$ with respect to $F_0$, and parameterized by $\theta_i \in [0,\theta_\star)$ as $f_{\theta}(x) = \exp{\theta x  - \log \varphi_0(\theta)}$, where $\varphi_0(\theta) = \int e^{\theta x} {\rm d}F_0(x)$, and $\theta_\star = \sup\{\theta > 0: \varphi_0(\theta) < \infty\}$.

As remarked in \cite{stoepker2021, arias-castro2018a} this family encompasses many models interesting in practice: the normal location model~\eqref{hyp:norm} common in signal detection literature, Poisson models arising from syndromic surveillance~\citep{kulldorff2005stp}, and various Bernoulli models~\citep{walther2010optimal, mukherjee2015hypothesis, Wu2014}.

Assume that for all observations, $X_{ij} \sim F_{\theta_i}$ from the one-parameter exponential family as defined above. In this parameterization, the hypothesis test in~\eqref{hyp:general} particularizes to:
\begin{align*}
H_0:\qquad &\forall_{i\in[n]} \quad \theta_{i}=0\ , \numberthis \label{hyp:exp}\\
H_1:\qquad &\exists_{\cS \subset [n]} : \forall_{i\in\cS} \quad \theta_{i} = \theta > 0\ , \text{ and } \forall_{i\notin\cS} \quad \theta_{i}=0 \ .
\end{align*}
The collection of subject sums constitute sufficient statistics for this hypothesis test. In the alternative hypothesis, constraining the anomalies to equal magnitude is motivated from a minimax stance \citep{stoepker2021} and leads to the following lower bound:
\begin{theorem}[\cite{stoepker2021}, Theorem~1]\label{th:exp-lowerbound}
Refer to the hypothesis testing problem~\eqref{hyp:exp} and parameterization $\theta=\tau\sqrt{2\rho(\beta,1)\log(n)/(\sigma_0^2t)}$ with $\tau$ constant, and $\abs{\cS}=n^{1-\beta}$ with $\beta \in(1/2,1)$. Provided $t = \smallOmega(\log^3(n))$, any test will be asymptotically powerless if $\tau<1$.
\end{theorem}

To particularize Theorem~\ref{th:rank_hc_test} to this one-parameter exponential setting define
\[
\Upsilon_0 \equiv \left(\sqrt{3}\, \E{\max\left\{\frac{X_1-\mu_0}{\sigma_0}, \frac{X_2-\mu_0}{\sigma_0}\right\}}\right)^{-1} \ ,
\]
where $X_1, X_2 \sim F_0$ independently. The following proposition characterizes the loss due to ranks in this parametric family:

\begin{prp}\label{prp:rank-perf-exp}
Refer to the hypothesis testing problem~\eqref{hyp:exp} and parameterization ${\theta=\tau\sqrt{2\rho(\beta,1)\log(n)/(\sigma_0^2t)}}$ with $\tau$ constant, and $\abs{\cS}=n^{1-\beta}$ with $\beta \in(1/2,1)$. Provided $t = \smallOmega(\log(n))$, the rank-based higher criticism test from Definition~\ref{def:rank_hc_stat} for any $\alpha\in (0,1)$ has power converging to one if $\tau > \Upsilon_0$.
\end{prp}

The proof is an extension of Proposition 2 from \cite{arias-castro2018a}, and given in Section~\ref{sec:rank-perf-exp}. Contrasting Theorem~\ref{th:exp-lowerbound} with Proposition~\ref{prp:rank-perf-exp} we see that within this parametric family the loss in power due to ranks is given by $\Upsilon_0$. For many common models, this quantity is small; for the normal model, for example, it corresponds to $\sqrt{\pi/3} \approx 1.023$. For the exponential distribution, this quantity corresponds to $2/\sqrt{3} \approx 1.15$. Furthermore, $\Upsilon_0 \geq 1$, and equality is attained if and only if $F_0$ is the continuous uniform distribution.
The magnitude of this asymptotic power loss arises in other contexts as well; see, for example, \citet[Chapter 4.3]{MR758442} for a statement pertaining to the normal location model.

Although this parametric family is relevant from a practical context, it does not allow us to explore the performance gap of our rank based test when both the signal magnitude and heteroskedasticity change. The reason is that, although in general parameter $\theta$ above influences both the mean and variance of $F_\theta$, the variance is only affected at higher order terms which are asymptotically irrelevant. Interestingly, an extension of Theorem~\ref{th:exp-lowerbound} to a two-parameter exponential family is far from trivial (see discussion in Section~\ref{sec:discussion}). 

\begin{rem}
The setting in~\eqref{hyp:exp} particularizes all anomalies to have equal magnitude, motivated from a minimax perspective \citep{stoepker2021}. A full characterization in the presence of varying anomalies --- akin to Theorem~\ref{th:rank_hc_test} --- requires an analysis going far beyond those presented in Theorem~\ref{th:exp-lowerbound} of \cite{stoepker2021}. To draw some connections between the setting of~\eqref{hyp:exp} and the setting of Theorem~\ref{th:rank_hc_test}, consider a setting with varying anomalies of magnitude $\theta_i$ for $i \in\cS$ where, in the spirit of categorizing ``weak'' and ``strong'' signals akin to~\eqref{eq:def_partition_anomalies}, one signal is ``strong'' with magnitude $\theta_i = \tau_i\sqrt{2\log(n)/(\sigma_0^2 t)}$. Then a sufficient condition for the likelihood ratio test to have power converging to one is $\tau_i > 1$, irrespective of the magnitude of the weaker signals. In that setting, a sufficient condition for our rank-based test to have power converging to one is $\tau_i > \Upsilon_0$.
\end{rem}

\subsubsection{Heteroskedastic convolutions}\label{sec:particular-conv}

Consider a setting where the signal manifests itself through a convolution with a normal distribution. This setting is of interest in cosmological applications \citep{Jin2005, Cai2014}. Specifically, with reference to~\eqref{hyp:general}, let $F_{0j}$ be the standard normal distribution, whereas for $i\in\cS$ distribution $F_{ij}$ is given by the convolution of a normal distribution with variance $\sigma^2$ and an unspecified distribution $G$ dilated by a scalar $\theta$. The hypothesis test~\eqref{hyp:general} then particularizes to:
\begin{align*}
H_0:\qquad &\forall_{i\in[n],j\in[t]} \quad X_{ij} \overset{\text{i.i.d.}}{\sim} \cN(0,1) \ , \numberthis \label{hyp:conv}\\
H_1:\qquad &\exists_{\cS \subset [n]} : \forall_{i\in\cS,j\in[t]} \quad X_{ij} = Z_{ij} + \theta Q_{ij} \ \text{ with } Z_{ij} \overset{\text{i.i.d.}}{\sim} \cN(0,\sigma^2) \\
&\text{ and } Q_{ij} \overset{\text{i.i.d.}}{\sim} G \text{ , and } \forall_{i\notin\cS,j\in[t]} \quad  X_{ij} \overset{\text{i.i.d.}}{\sim} \cN(0,1) \ .
\end{align*}

Note that if $G$ is the mass distribution at $1$, then this this precisely the setting in~\eqref{hyp:norm} with repeated observations. As shown in \cite{Cai2014}, if $G$ is unbounded then the dilation $\theta$ needs to be sufficiently small with respect to the specific distribution $G$ in order for the problem to exhibit nontrivial behavior. Therefore, we restrict ourselves to the setting where $G$ has bounded support on $[0,1]$. The essential supremum of $Q_{ij}\sim G$ plays a prominent role, and is defined as:
\[
s(G) \equiv \inf\{a \in \bbR : G(a) = 1\} \ .
\]
The \emph{optimal} test for some particularizations of~\eqref{hyp:general} may not rely on subject means, as these are generally not sufficient statistics. For example, when the anomalous signal heteroskedasticity is large a test based on subject variances instead may perform much better, and this may be the case for the setting~\eqref{hyp:conv}. To allow for a meaningful comparison that captures the effect of ranking, we contrast with the performance of the best possible test based on subject means only as the rank-based test in Definition~\ref{def:rank_trans_char} is based on rank subject means as well.

Using results from \cite{Cai2014}, we obtain the following proposition, proven in Section~\ref{sec:prp:means-perf-conv_proof}.
\begin{prp}\label{prp:means-perf-conv}
Refer to the hypothesis testing problem in~\eqref{hyp:conv} with $G$ having support on $[0,1]$, and parameterization $|\cS|=n^{1-\beta}$ with $\beta \in(1/2,1)$ and $\theta=s(G)^{-1}\sqrt{2r\log(n)/t}$ with $r > 0$ constant. Suppose that only the subject sums $\sum_{j=1}^tX_{ij}$ are observed. Then, any test is asymptotically powerless if $r < \rho(\beta,\sigma)$.
\end{prp}

The following quantity captures the limit of the variance of the anomalous signal after rank transform as in Definition~\ref{def:rank_trans_char} within the particularization~\eqref{hyp:conv}: 
\[
\xi_\sigma^2 \equiv \lim_{n\to\infty} \sigma_{n,i}^2 =  12\left(\E{\Phi\Big(\cN(0,\sigma^2)\Big)^2}-\frac{1}{4}\right) \ ,
\]
where $i\in\cS$. Letting $\mu(G)$ denote the mean of distribution $G$, the following proposition characterizes the loss due to ranks in this parametric family:

\begin{prp}\label{prp:rank-perf-conv}
Refer to the hypothesis testing problem in~\eqref{hyp:conv} with $G$ having support on $[0,1]$, and parameterization $|\cS|=n^{1-\beta}$ with $\beta \in(1/2,1)$ and $\theta=s(G)^{-1}\sqrt{2r\log(n)/t}$ with $r>0$ constant. Let $t = \smallOmega(\log(n))$. Then, the rank-based higher criticism test from Definition~\ref{def:rank_hc_stat} for arbitrary $\alpha\in (0,1)$ has power converging to one if
\[
r > \sqrt{\frac{s(G)}{\mu(G)}} \tilde \rho(\beta,\sigma) \ \text{ , where } \tilde\rho(\beta,\sigma) \equiv \sqrt{\tfrac{1}{6}\pi(\sigma^2+1)}\rho(\beta,\xi_\sigma) \ .
\]
\end{prp}

The proof is given in Section~\ref{sec:rank-perf-conv}. A visualization of $\tilde\rho(\beta,\sigma)$ is given in Figure~\ref{fig:detection_boundary-rank}. Note that the results above imply that for some type of signals, the rank based test may have power converging to one while no test based on subject means exists which is asymptotically powerful! This is due to the variance after rank transform $\xi_\sigma$ being larger than $\sigma$ when the Gaussian component of the alternative is sufficiently underdispersed with respect to the null noise. Essentially, the ranks ``push-out'' the relatively low-variance signals, which eases detection. Referring to Figure~\ref{fig:detection_boundary-rank}, if $G$ is a point mass at 1, then the rank-based higher criticism test from Definition~\ref{def:rank_hc_stat} is powerful when $r$ is in the red region, but in that region no asymptotically powerful test based on subject means exists. However, for larger $\sigma$, we can expect rank-based tests to perform poorer than their oracle counterparts, as they cannot leverage heteroskedasticity as efficiently. Note that these remarks concern the best possible test based on subject means, and not the overall best possible test.

\begin{figure}[htb]
\centering
\begin{subfigure}{0.38\textwidth}
\hspace{0.4cm}\includegraphics[width=\linewidth]{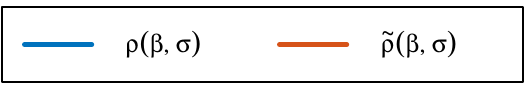}\vspace{0.3cm}
\end{subfigure}

\begin{subfigure}{0.32\textwidth}
\includegraphics[width=\linewidth]{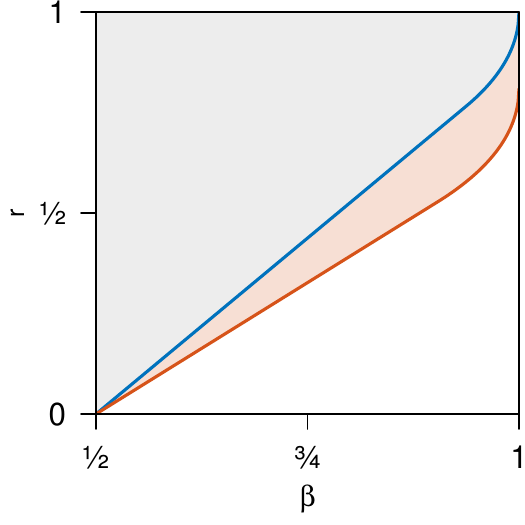}
\caption{Boundary for $\sigma = 0.5$.}
\end{subfigure}
\begin{subfigure}{0.32\textwidth}
\includegraphics[width=\linewidth]{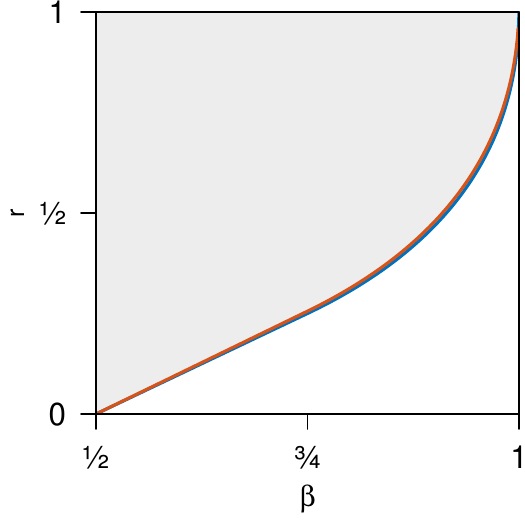}
\caption{Boundary for $\sigma = 1$.}
\end{subfigure}
\begin{subfigure}{0.32\textwidth}
\includegraphics[width=\linewidth]{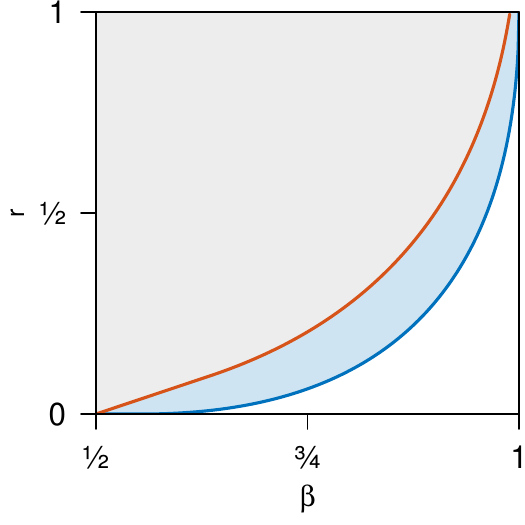}
\caption{Boundary for $\sigma = 1.5$.}
\end{subfigure}
\caption{The detection boundary $\rho(\beta,\sigma)$ and $\tilde\rho(\beta,\sigma)$ for different values of~$\sigma$.}\label{fig:detection_boundary-rank}
\end{figure}

\section{Simulation study}\label{sec:sim-settings}

In this section we conduct several numerical simulations to gain further insight and illustrate some of the theoretical results. This furthermore allows us to give practical recommendations regarding our methodology. Though our methodology is more flexible, we only consider settings where the null and alternative distributions are identical for each referential $j\in[t]$ for clarity and to allow comparison with other works.

To confirm the theoretical results established in Proposition~\ref{prp:rank-perf-exp} and Proposition~\ref{prp:rank-perf-conv}, we consider the settings described in Table~\ref{tb:sim-settings1} and Table~\ref{tb:sim-settings2} respectively.

\begin{table}[htb]
\centering
\begin{tabular}{@{}llll@{}}
\toprule
Name & $F_0$                            & $F_\theta$                                       								& $\Upsilon_0$ \\ \midrule
Normal        & $\cN(0,1)$                                & $\cN(\theta,1)$                                            	& $\sqrt{\pi/3}$ \\
Exponential   & $\text{Exponential}(\tfrac{3}{2})$ (rate) & $\text{Exponential}(\tfrac{3}{2}-\theta)$ (rate)          	& $\tfrac{2}{3}\sqrt{3}$\\
Uniform       & $\text{Uniform}(0,1)$                     & $F_\theta(x) = \frac{\exp{\theta x} - 1}{\exp{\theta}-1}$ 	& 1 \\ \bottomrule
\end{tabular}
\caption{Simulation settings fitting the exponential hypothesis testing scenario~\eqref{hyp:exp} considered. The third column contains the distributions from the exponential family with $F_0$ as the base measure as described in Section~\ref{sec:exp-family}. The derivation of the values of $\Upsilon_0$ are provided in Appendix~\ref{app:upsilon}.}\label{tb:sim-settings1}
\end{table}

\begin{table}[htb]
\centering
\begin{tabular}{@{}lll@{}}
\toprule
Name 						& $G$   															& $\zeta_G(\beta,\sigma)$ 			\\ \midrule
Normal($\sigma$)		& Point mass at 1    											& $\sqrt{\frac{\pi(\sigma^2+1)\rho(\beta,\xi_\sigma)}{6\rho(\beta,\sigma)}}$				 		\\
Triangular($\sigma$)   				& Triangular distribution on $[0,1]$ with mode at $\tfrac{1}{2}$ 	& $\sqrt{\frac{2\pi(\sigma^2+1)\rho(\beta,\xi_\sigma)}{3\rho(\beta,\sigma)}}$          			\\ \bottomrule
\end{tabular}
\caption{Simulation settings fitting the normal convolution hypothesis testing scenario~\eqref{hyp:conv} considered, where in the remainder we indicate within parentheses the value of $\sigma$. Quantity $\zeta_G$ is defined in~\eqref{eq:def_zetaG}. Note that Normal($\sigma=1$) corresponds to the first setting of Table~\ref{tb:sim-settings1}.}\label{tb:sim-settings2}
\end{table}

In the settings of Table~\ref{tb:sim-settings1} corresponding to the exponential family setting~\eqref{hyp:exp}, we further parameterize $\theta$ using a parameter $\tau$ as in Theorem~\ref{th:exp-lowerbound} and Proposition~\ref{prp:rank-perf-exp}:
\begin{equation}\label{eq:theta-param-exp}
\theta_\tau = \tau\sqrt{\frac{2\rho(\beta,1)\log(n)}{\sigma_0^2 t}} \ .
\end{equation}

Similarly, in the settings of Table~\ref{tb:sim-settings2} corresponding to the convolution setting~\eqref{hyp:conv}, we further parameterize $\theta$ using a parameter $\tau$ as:
\begin{equation}\label{eq:theta-param-conv}
\theta_\tau=\tau \cdot s(G)^{-1}\sqrt{\frac{2\rho(\beta,\sigma)\log(n)}{t}} \ .
\end{equation}
Note that if $\sigma < \sqrt{2}$ then under the parameterization above Proposition~\ref{prp:means-perf-conv} implies no test based on stream means is asymptotically powerful when $\tau < 1$, and the rank-based test is asymptotically powerful (under suitable restrictions of $k_n$ and $t$) if
\begin{equation}\label{eq:def_zetaG}
\tau > \zeta_G(\beta,\sigma) \equiv \sqrt{\frac{\pi(\sigma^2+1)s(G)^{2}\rho(\beta,\xi_\sigma)}{6\mu(G)^2\rho(\beta,\sigma)}} \ .
\end{equation}
Finally, we consider a setting involving Cauchy distributions to highlight how our test performance depends on the parameters $\mu_{n,i}$ and $\sigma_{n,i}$, and not on the original distribution mean and variance. Specifically, referring to the general hypothesis test~\eqref{hyp:general} with nonvarying referentials, we particularize $F_{0j} = F_0$ to a Cauchy distribution with location parameter 0 and scale parameter 1, and for $i\in\cS$ we particularize $F_{ij} = F_i$ to a Cauchy distribution with location parameter $\theta$ and scale parameter 1. We further parameterize $\theta$ using $\tau$ as:
\begin{equation}\label{eq:theta-param-cauchy}
\theta_\tau = \tau\cdot\pi \sqrt{\frac{2\rho(\beta,1)\log(n)}{3t}} \ ,
\end{equation}
such that Theorem~\ref{th:rank_hc_test} implies the rank test from Definition~\ref{def:rank_hc_stat} has power converging to one if $\tau > 1$. Details are deferred to Appendix~\ref{app:sim-cauchy}.

\subsection{Grid choice}\label{sec:exp-grid}
The grid in~\eqref{eq:def_original_grid} is sufficient for the asymptotic power guarantees given by Theorem~\ref{th:rank_hc_test}. In finite samples, however, there may be evidence against the null hypothesis that is most clearly observed outside of this grid. Specifically, for extremely strong anomalies, the discrepancy between $N_q(\bR)$ and $np_q$ may be highest at a value $q > 2$. To avoid disregarding this evidence, we introduce an extension of the grid in~\eqref{eq:def_original_grid} that includes a point such that $N_q(\bR) = 0$. Since $Y_i(\bR) \leq n$, we have that
\[
\frac{Y_i(\bR) - \overline R}{\sigma_R} \leq \frac{n - \frac{n+1}{2}}{\sqrt{\frac{n^2-1}{12}}} \leq \sqrt{3} \ .
\]
Therefore, if $q > 3t/2\log(n)$, then $N_q(\bR) = 0$. We thus extend the grid in~\eqref{eq:def_original_grid} as:
\begin{equation}\label{eq:extended-grid}
\tilde Q_n \equiv \left\{\frac{1}{k_n},\frac{2}{k_n},\dots, \frac{3t}{2\log(n)} \right\} \ .
\end{equation}
Note that $|\tilde Q_n| = \bigO(k_nt/\log(n))$. In the proof of Theorem~\ref{th:rank_hc_test}, it was critical that the gridsize~$k_n$ was subpolynomial. The extended grid is no longer of this size unless we restrict $t=n^{o(1)}$. As we have previously remarked, we conjecture that the upper bound requirement of the gridsize seems an artifact of the proof. Our numerical experiments do not show deteriorating performance when the gridsize is increased, and therefore we use the extended grid in~\eqref{eq:extended-grid} in subsequent experiments.

This leaves us to choose an appropriate value of $k_n$. From a practical standpoint a smaller grid resolution allows for a faster implementation and it is therefore useful to investigate what value of $k_n$ is sufficient. To recommend a choice of $k_n$, we compare the performance of the test for different values of $k_n$. Figure~\ref{fig:grid} depicts the results for three of the simulation settings given in Table~\ref{tb:sim-settings1} and Table~\ref{tb:sim-settings2}, with $n=10^3$ and $t=\ceil{\log(n)} = 7$, and sparsity $|\cS| = \ceil{ n^{1-0.85}} = 3$. Here, we observe that $k_n = \ceil{\log^2(n)}$ is sufficient. Supplemental results given in Appendix~\ref{app:sim-grid} indicate this choice is adequate for other settings as well, and we therefore use this choice for the remainder of the numerical experiments.

\begin{figure}
\centering
\begin{subfigure}{0.6\textwidth}
    \includegraphics[width=\textwidth]{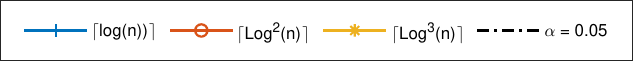}\vspace{0.3cm}
\end{subfigure}

\centering
\begin{subfigure}{0.32\textwidth}
    \includegraphics[width=\textwidth]{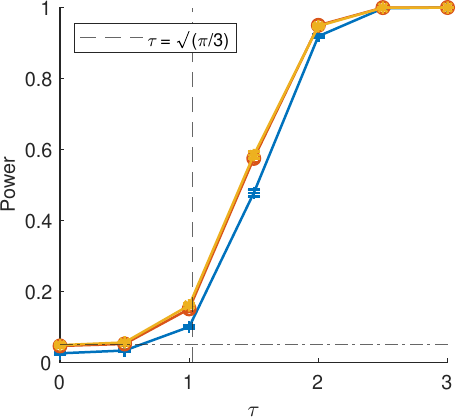}
    \caption{Normal}
\end{subfigure} 
\begin{subfigure}{0.32\textwidth}
    \includegraphics[width=\textwidth]{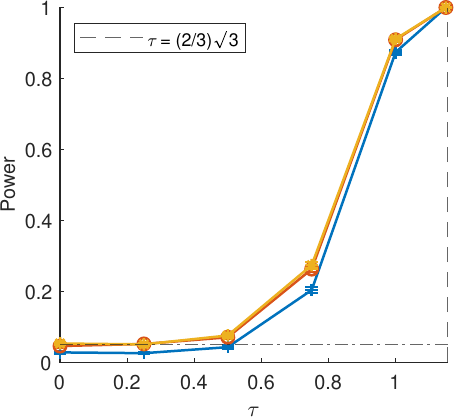}
    \caption{Exponential}
\end{subfigure}
\begin{subfigure}{0.32\textwidth}
    \includegraphics[width=\textwidth]{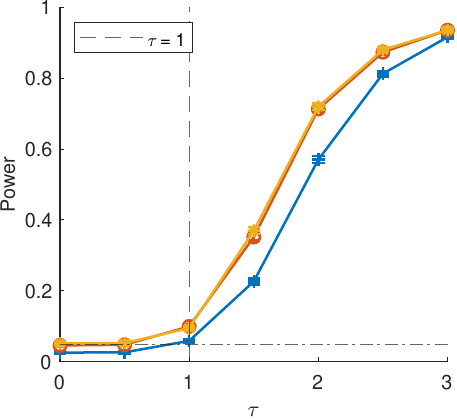}
    \caption{Cauchy}
\end{subfigure} 
\caption{Simulated power at $5\%$ significance for the rank test of Theorem~\ref{th:rank_hc_test} under different choices of $k_n$ as in~\eqref{eq:extended-grid} as a function of the signal strength. The subcaptions refer to the settings of Table~\ref{tb:sim-settings1} and Table~\ref{tb:sim-settings2}, with relevant parameters further parameterized using $\tau$ as in Equations~\eqref{eq:theta-param-exp},~\eqref{eq:theta-param-conv},~\eqref{eq:theta-param-cauchy}. In the above, $n=1000$, $t=7$ and $\lvert\cS\rvert = 3$. Monte-Carlo simulation for the null statistics is based on $10^5$ samples. Each signal level was repeated $10^4$ times, leading to the $95\%$ confidence bars depicted.}
\label{fig:grid}
\end{figure}

\subsection{Simulation results}\label{sec:sim}
We compare the performance of our methodology with three relevant alternatives; the permutation test as described in \cite{stoepker2021} and two tests calibrated using information on the involved distributions; a test based on a higher criticism statistic involving the subject means, and the Friedman test.

\myparagraph{Permutation test from \cite{stoepker2021}} We use the extended data-dependent grid with resolution $\log(n)$ as recommended in that work. Unlike our rank-based methodology, the validity of the inference from this permutation-based test hinges crucially on the assumption of identical referentials. This permutation test is therefore not a suitable alternative to our rank-based methodology in all settings --- the latter is more flexible. Since we have constrained ourselves to simulation settings with identical referentials (i.e.~$F_{0j}=F_0$), this permutation test is also applicable which allows for a meaningful comparison. Nevertheless, when contrasting the upcoming numerical results one should account for the relative flexibility of the rank-based methodology.

\myparagraph{Distribution-aware higher criticism test} The distribution-aware higher criticism test depends on the null distribution. We use the same form of the higher criticism statistic as we used to define the methodology of Theorem~\ref{th:rank_hc_test} and base the statistic on the subject means. Let $\mu_0$ and $\sigma_0$ be the null mean and variance respectively (when it exists), and define:
\[
T_{\text{dist-hc}}(\bX) = \max_{q\in Q(\bX)}\left\{ \frac{ \sum_{i\in[n}\ind{ (Y_i(\bX)-\mu_0)/\sigma_0 \geq \sqrt{2q\log(n)/t} } - nw_q}{\sqrt{nw_q(1-w_q)}} \right\},
\]
where $Y_i(\bX) = \frac{1}{t}\sum_{j\in[t]} X_{ij}$ denote the subject means and $Q(\bX) = \{\frac{1}{k_n},\frac{2}{k_n},\dots,(M(\bX)^2-\mu_0)t/(2\sigma_0\log(n)) \}$ with $M(\bX) = \frac{1}{t}\sum_{j\in[t]} \bX_{(j)}$ where $\bX_{(j)}$ denotes the $j$th overall order statistic of the data (i.e.~$M(\bX)$ is the average of the $t$ largest observations), and finally
\[
w_q = \Phn{\frac{Y_i(\bX)-\mu_0}{\sigma_0} \geq \sqrt{\frac{2q\log(n)}{t}}} \ .
\]
For the simulations, we use $k_n = \log(n)^2$. For the Cauchy distribution, the first two moments do not exist. Instead, the null median is used for $\mu_0 = 0$ and the null scale parameter is used to set $\sigma_0 = 1$.

\myparagraph{Friedman test \citep{Friedman1937}} Monte-Carlo simulation from the true null distribution is used to calibrate the test.

We investigate and compare the performance of the test of Theorem~\ref{th:rank_hc_test} under varying signal strengths in the exponential family settings of Table~\ref{tb:sim-settings1}, the convolution settings of Table~\ref{tb:sim-settings2}, and the Cauchy setting described in Section~\ref{sec:sim-settings} for various sample sizes. The results are given in Figure~\ref{fig:power}. Supplemental results are provided in Appendix~\ref{app:sim-perf}. The results of the permutation test are not computed for the settings where $n=10^4$ due to computational costs.

\begin{figure}
\centering
\begin{subfigure}{0.6\textwidth}
    \includegraphics[width=\textwidth]{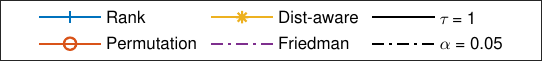}\vspace{0.3cm}
\end{subfigure}

\begin{subfigure}{0.32\textwidth}
    \includegraphics[width=\textwidth]{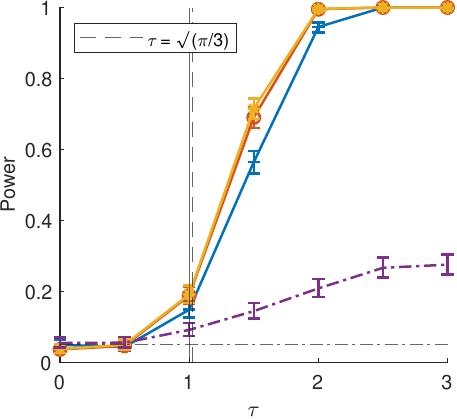}
    \caption{Normal($\sigma = 1$)}\label{fig:power-first-smalln}
\end{subfigure}
\begin{subfigure}{0.32\textwidth}
    \includegraphics[width=\textwidth]{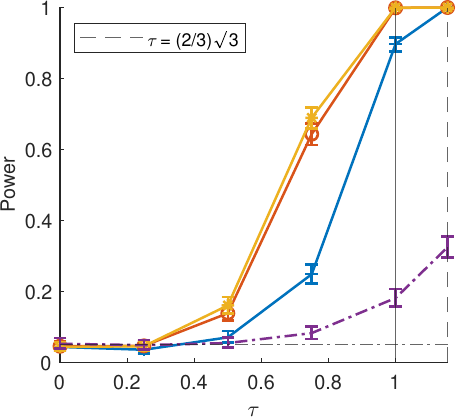}
    \caption{Exponential}
\end{subfigure} 
\begin{subfigure}{0.32\textwidth}
    \includegraphics[width=\textwidth]{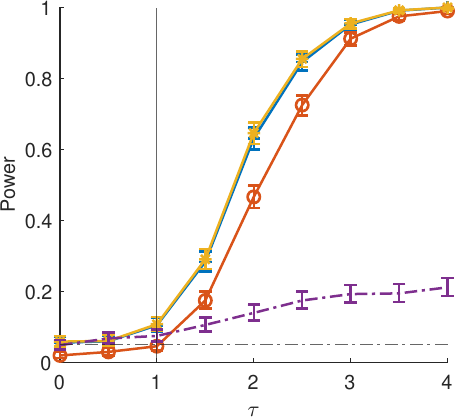}
    \caption{Uniform}
\end{subfigure}
\begin{subfigure}{0.32\textwidth}
    \includegraphics[width=\textwidth]{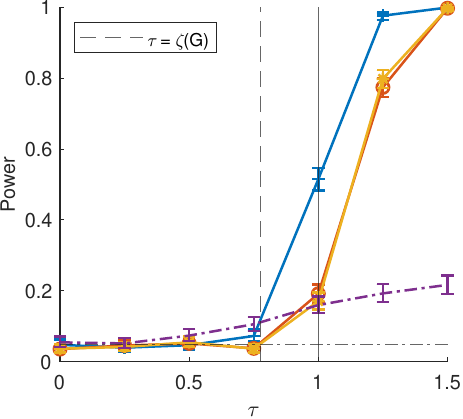}
    \caption{Normal($\sigma = \tfrac{1}{2}$)}
\end{subfigure}
\begin{subfigure}{0.32\textwidth}
    \includegraphics[width=\textwidth]{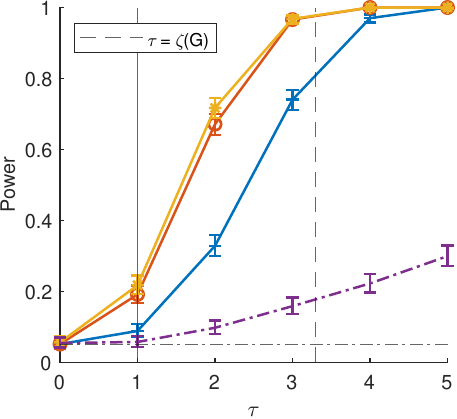}
    \caption{Normal($\sigma = \tfrac{3}{2}$)}
\end{subfigure} 
\begin{subfigure}{0.32\textwidth}
    \includegraphics[width=\textwidth]{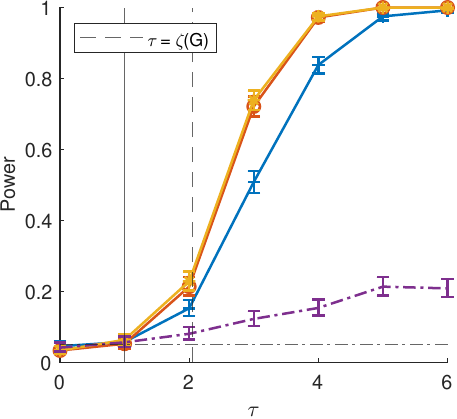}
    \caption{Triangular($\sigma = 1$)}
\end{subfigure}
\begin{subfigure}{0.32\textwidth}
    \includegraphics[width=\textwidth]{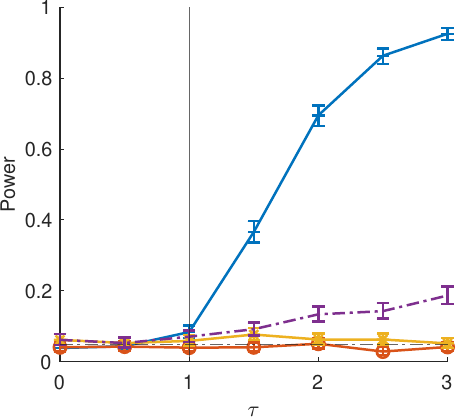}
    \caption{Cauchy}\label{fig:power-last-smalln}
\end{subfigure}  
\begin{subfigure}{0.32\textwidth}
    \includegraphics[width=\textwidth]{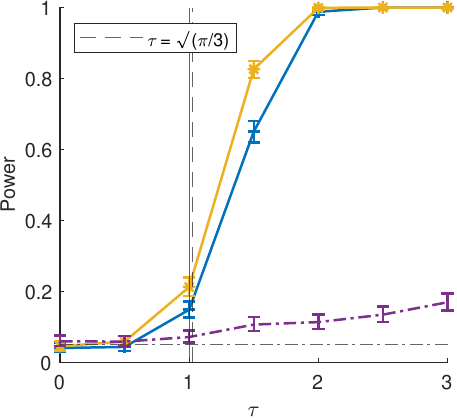}
    \caption{Normal($\sigma=1$)}\label{fig:power-first-bign}
\end{subfigure}  
\begin{subfigure}{0.32\textwidth}
    \includegraphics[width=\textwidth]{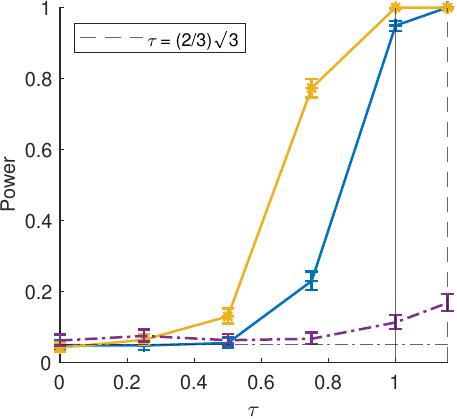}
    \caption{Exponential}\label{fig:power-last-bign}
\end{subfigure}
\caption{Simulated power at $5\%$ significance for the rank test of Theorem~\ref{th:rank_hc_test} and the three tests described in Section~\ref{sec:sim}, as a function of the signal strength. The subcaptions refer to the settings described in Section~\ref{sec:sim-settings} with signals further parameterized as in~\eqref{eq:theta-param-exp},~\eqref{eq:theta-param-conv} and~\eqref{eq:theta-param-cauchy}. Figures~\ref{fig:power-first-smalln}--\ref{fig:power-last-smalln} concern $n=10^3, \ t=7$ and $|\cS| = 3$. Figures~\ref{fig:power-first-bign} and~\ref{fig:power-last-bign} concern $n=10^4, \ t=10$ and $|\cS| = 4$. Monte-Carlo simulation for the null statistics is based on $10^5$ samples. The permutation test was executed with $10^3$ permutations for each simulation. Each signal level was repeated $10^3$ times, leading to the $95\%$ confidence bars depicted.}\label{fig:power}
\end{figure}

On the whole, the theoretical asymptotic performance gaps from Proposition~\ref{prp:rank-perf-exp} and Proposition~\ref{prp:rank-perf-conv} closely match what we observe through simulation. However, for the convolution setting with $G$ particularized to a triangular distribution, the theoretical gaps are not so closely observed in this finite sample. Although Proposition~\ref{prp:means-perf-conv} indicated that the only driving force of distribution $G$ is its essential supremum, since this notion is true only in an asymptotic sense, this is not observed in our experimental setup. Therefore, the factor in the performance gap $s(G)/\mu(G)$ from Proposition~\ref{prp:rank-perf-conv} gives rise to a larger performance gap between the best possible test (based on subject means) and our rank-based test than what is observed in practice.

The superior performance of the rank-based test compared to the means-based oracle test for underdispersed signals as predicted by Proposition~\ref{prp:rank-perf-conv} is observed as well. This is thus not merely an asymptotic peculiarity, but rather a reason to consider this test apart from its flexibility. 

Perhaps surprisingly, for the uniform setting the rank-based test performs better (and not on-par) with the permutation test. However, within this setting the rank-based test is essentially an oracle test, so this is not unexpected.

The results of the Cauchy setting confirm that our (asymptotic) testing performance only depends on the quantities $\mu_{n,i}$ and $\sigma_{n,i}$, and show how the ``regularizing'' effect of the ranks can be useful. The permutation and distribution-aware tests are naturally weak, as they are based on subject means. As the Cauchy distribution has heavy-tails these are rather uninformative summary statistics. Conversely, the rank-transformation prior to averaging in our methodology provides power in our methodology. Essentially, if large observations are informative (such as is the case for the exponential distribution setting) then using ranks incurs a loss of information and thus power; conversely, if large observations are uninformative (such as is the case for the Cauchy distribution) then using ranks regularizes and increases power.

Contrasting the performance of our rank-based test of Definition~\ref{def:rank_hc_stat} with that of Friedman test, the unsuitability of the latter is quite apparent, particularly for large samples.

\section{Application to quality control in pharmaceutical manufacturing}\label{sec:application}
To illustrate how our methodology may be applied in practice, we examine pharmaceutical manufacturing data of a batch-produced drug. This industry is subject to strict regulations and demands rigorous quality monitoring. Therefore, prompt notification of possible production issues is crucial, and the proposed anomaly detection methodology is valuable in this regard.

\myparagraph{Description of the data and measurements used}
We consider the laboratory data from \cite{Zagar2022}, which contain quality control measures of production batches at three production phases; before production on the incoming raw materials, at an intermediate production step, and after finalization of the drug. The data consist of measurements on a total of 25 different subfamilies of the drug. The quality control measures all exist within their own referential. For example, one measure corresponds to the active content of the drug released in 30 minutes time, while another measure corresponds to the number of impurities in the drug. The goal is to detect the presence of a subset of batches for which these quality control measures are distributed differently from the remaining batches. Our methodology is flexible and can jointly deal with these diverse referentials in a nonparametric manner.

We have selected two groups of two or more quality measurements which will be used for the purposes of anomaly detection over the batches. Since our calibration is only exact if there are no dependencies under nominal circumstances, we only consider groups of quality measures for which this is a reasonable assumption given the described context. 

To avoid dependencies, firstly we only consider quality control measurements within a single production phase. Secondly, some measurements within a single phase may be highly dependent, such as the minimum, average, and maximum weight of the tablets. In such cases, we select the averages. Thirdly, some measurements are done both before and after a film is applied, and in such cases we select the measurement based on the film covered drugs. Finally, some measurements may be dependent through latent variables, such as weight and thickness of the tablets. In such cases, we have selected a subset of the remaining measurements. 

To avoid incorporating missing values, we exclude some measures where a considerable amount of missing values is present. Finally, measurements on the raw material phase contain a considerable number of ties. While our methodology is exact under ties, one can expect signals within such measurements to have little power and we thus exclude them. The final two groups of measurements we consider are given in Table~\ref{tb:measurements}. 

\begin{table}[htb]
\centering
\begin{tabular}{l@{\hskip 1cm}p{8cm}@{\hskip 1cm}l}
\toprule
\textbf{Phase} & \textbf{Measurements considered}                                                                       & \textbf{Identifier} \\ \midrule
Intermediate   & • Film-coated tablet weight relative standard deviation measured after coating process.                                          & fct\_rsd\_weight    \\
               & • Average film-coated tablets hardness measured after coating; in Newtons.                  & fct\_av\_hardness   \\ \midrule
Final product  & • Drug release from final tablet in defined time: average (calculated) \% of API released in 30 minutes. & dissolution\_av     \\
			   & • Residual solvent content in final product measured with gas chromatography (GC) method.				& residual\_solvent \\
               & • Total impurities content in final product measured with HPLC method.                                   & impurities\_total  \\ \bottomrule
\end{tabular} \caption{Considered measurements in the case study of Section~\ref{sec:application}.}\label{tb:measurements}
\end{table}

\myparagraph{Comparison with normality-based approach}
We contrast our methodology with an approach based on an assumption of normality of the observations. For each measurement, we center the observations using their sample mean and standardize them using the sample standard deviations. The resulting observations are then summed and the parametric higher criticism is applied to their sums. A selection of the resulting subject means are provided in Figure~\ref{fig:application-histogram}, and supplemental figures are provided in Appendix~\ref{app:supp-app}.

\begin{figure}[htb]
\centering
\begin{subfigure}{0.32\textwidth}
    \includegraphics[width=\textwidth]{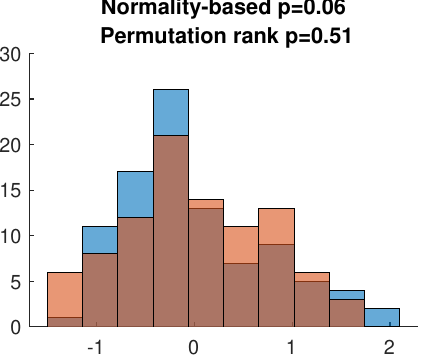}
    \caption{Subfamily 1, intermediate.}
\end{subfigure}
\begin{subfigure}{0.32\textwidth}
    \includegraphics[width=\textwidth]{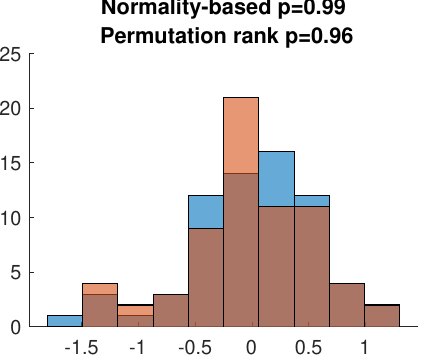}
    \caption{Subfamily 21, intermediate.}
\end{subfigure}
\begin{subfigure}{0.32\textwidth}
    \includegraphics[width=\textwidth]{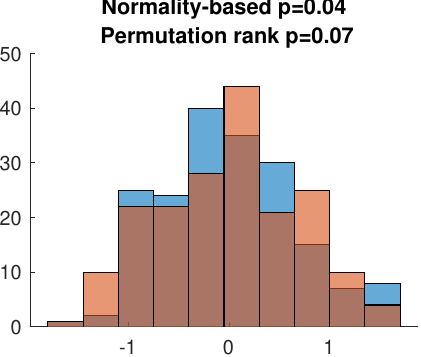}
    \caption{Subfamily 23, final product.}
\end{subfigure}
\caption{Histogram of (normalized) subjects means. The blue bars represent subject means based on standardized observations, the red bars represent subject means based on standardized observation fractional ranks. }\label{fig:application-histogram}
\end{figure}

\myparagraph{Dealing with ties}
Using our test with random tie-breaking comes with the guarantees in Theorem~\ref{th:rank_hc_test}, particularly conservativeness and near-exactness of the test. However, the randomness of the tie-breaking procedure results in a random $p$-value which may be undesirable in practice. Dealing with ties is a common challenge in rank-based testing. Various starting points to adapt rank-based tests for ties have been discussed in the literature: for example in Section 5.6 in \cite{Gibbons2003}. A common approach to avoid the test outcome being random is through the use of \textit{midranks} (see, for example, Section 1.4 in \cite{Lehmann1975}, sometimes referred to as fractional ranks or average ranks) rather than random tie-broken ranks. Midranking amounts to assigning groups of tied observations the expected value of the rank under random tie breaking. We suggest two avenues in which midranks can be used in our testing approach; a near-exact approach, and a naive approach.

The first approach computes the value of $N_q(\bR)$ based on the midranks. Estimating the value of $p_q$, as well as the calibration of the statistic, is done through column-wise permutation. This comes at a higher computational cost due to the calibration no longer depending on merely the sample size. One can prove conservativeness of such a methodology by arguing that the set of columnwise permutations (defined in the proof of Theorem~\ref{th:rank_hc_test} in Equation~\eqref{eq:def-permset}) forms a group with operation the composition of permutations --- see Theorem 15.2.1 in \cite{MR2135927}.

The second approach likewise computes the value of $N_q(\bR)$ based on the midranks. However, estimating the value of $p_q$, as well as calibration of the statistic, is done using the null distribution assuming ties are broken at random. This second approach comes without calibration guarantees, but no higher computational costs are incurred. 

As it turns out, in the present situation the $p$-values that result from these two approaches are similar.

\myparagraph{Results}
We apply our methodology on data of each drug subfamily separately and use the sets of measurements reported in Table~\ref{tb:measurements}. We only consider drug subfamilies whenever the data concerns more than 100 production batches with no missing data over the measurements. This means we ultimately consider the drug subfamilies 1, 13, 15, 17, 21, and 23, as labeled in the data source. The results are given in Table~\ref{tb:p-val-res}.

\begin{table}[htb]
\centering
\begin{tabular}{l|lll|lll}
\toprule
\textbf{Subfamily} & \textbf{Intermediate} & \textbf{}      & \textbf{}         & \textbf{Final}       &                &                   \\
\textbf{}               & \textit{Permutation}  & \textit{Naive} & \textit{Normality} & \textit{Permutation} & \textit{Naive} & \textit{Normality} \\ \midrule
1                       & 0.51                  & 0.45           & 0.06              & 0.15                 & 0.12           & 0.08              \\
13                      & 0.83                  & 0.91           & 0.27              & 0.74                 & 0.92           & 0.04              \\
15                      & 0.92                  & 0.89           & 0.22              & 0.65                 & 0.88           & 0.56              \\
17                      & 0.54                  & 0.69           & 0.54              & 0.47                 & 0.79           & 0.12              \\
21                      & 0.96                  & 0.99           & 0.99              & 0.90                 & 0.93           & 0.19              \\
23                      & 0.53                  & 0.64           & 0.37              & 0.07                 & 0.05           & 0.04  \\ \bottomrule           
\end{tabular}\caption{Resulting $p$-values from application of our rank-based methodology with fractional ranks and the higher-criticism statistic assuming normality, on the measurements of Table~\ref{tb:measurements} separately for each drug family. The $p$-values of our rank-based approach are computed based on calibration by permutation (``\textit{Permutation}'') or using the null distribution under random tie-breaking (``\textit{Naive}''), both based on $10^4$ samples. Values are rounded to two decimal places.}\label{tb:p-val-res}
\end{table}

Our test indicates potential anomalous behavior based of the final measurements of the subfamily with code 25, and based on this finding it would be advisable to further investigate this set of batches for production anomalies. Discovery of evidence for potential anomalous behavior in the final measurements is in line with comments made by the authors of the dataset \citep{Zagar2022}, who note that for the quality measures at the final production stage, the deviations from the averages (when considering all products simultaneously) are largest among all quality measures.

\myparagraph{Discussion}
Note that our methodology does not identify the subset of anomalous batches directly. Identification of anomalies comes with its own set of challenges and typically requires stronger signals for reliable detection \citep{Donoho2004}. It may be natural to consider the subject-level $p$-values as:
\[
p_i(\br) = \Phn{ Y_i(\bR) \geq Y_i(\br)} \ ,
\]
where $\br$ is the currently observed (fixed) data in the expression above. We resume this discussion in Section~\ref{sec:discussion}.

\begin{rem}
Note that our methodology is one-sided, and we principally have monitored for large values in all measurements. In case one would prefer to monitor low values instead, one could compute the ranks on the observations after multiplication by $-1$. A two-sided version (which we have not analyzed) could be constructed by computing ranks of the observations after centering and computing absolute values, or by suitably combining two one-sided tests.
\end{rem}

\section{Discussion}\label{sec:discussion}
In this paper, we have introduced a rank-based higher criticism statistic. Among others, it is useful as a high-dimensional extension of the classical Friedman test setting. The characterization of the regimes under which the test has asymptotic power converging to one were given in a nonparametric way and are easily particularizable to common parametric models. The underlying analysis requires careful treatment of rank-induced dependencies such that a sharp characterization of the tails of the rank distributions could be obtained. A surprising connection with a parametric testing scenario was drawn, showing that under moderate amount of referentials the testing dynamics boil down to those of the heteroskedastic normal location model. Finally we have shown that particularization to relevant parametric families allows us to elegantly capture the difference in power between the best possible test and our rank-based test. Perhaps surprisingly, a test with distributional knowledge based on subject means may actually perform worse than our rank-based test in this setting.

The one-parameter exponential family was considered in Section~\ref{sec:exp-family} as a particularization of our hypothesis test~\eqref{hyp:general}. This family has a drawback of only parameterizing signal magnitude and not heteroskedasticity. Although a two-parameter exponential family would allow us to parameterize both these quantities, finding a sharp lower bound result analogous to Theorem~\ref{th:exp-lowerbound} is far from trivial. Note that in our context it is most informative to present lower bound results on the best possible test solely observing subject sums, as our rank-based methodology is based solely on subject rank sums (as discussed in Section~\ref{sec:particular-conv}). Under this restriction application of \cite{Cai2014} or direct proofs akin to the proof of Theorem~\ref{th:exp-lowerbound} are complicated due the subject means no longer constituting as sufficient statistics. Although the individual anomalous observations are exponentially tilted counterparts of the null distributions, when summing anomalous observations we can no longer characterize their sum distributions as exponentially tilted counterparts of the null convolution distributions --- unlike in the one parameter exponential family. This nontrivially complicates the analysis and is an interesting avenue for future work; along with a statement that makes no restriction on solely observing subject means.

Our main result requires an upper bound on the number of observations per subject as $t = o(n)$. As discussed in Section~\ref{sec:analytic_calpow}, overcoming this upper bound is challenging since it requires an even sharper understanding of the nominal rank distribution under the presence of anomalies. We conjecture this upper bound is an artifact of the proof, however, and not indicative of potential asymptotic power deterioration when $t$ grows large. This conjecture is in line with exploratory simulation results which we include in Figure~\ref{fig:streamlengths} in Appendix~\ref{app:streamlengths}.

In the presentation of our asymptotic power results, we have not assumed that all referentials are affected in the presence of anomalies. Instead the (averaged) quantities in Definition~\ref{def:rank_trans_char} must be sufficiently large --- which is a result of basing our test on subject rank means --- and this is also possible when a fraction of referentials are unaffected. Nevertheless, when only a small fraction of referentials are affected (placing the setting closer to that of detecting a sparse submatrix akin to \cite{Butucea2013, Ma2015, Arias-Castro2017a}), a test based on quantities other than subject rank means (e.g. subject rank quantiles) may perform much better and this is an interesting area for future research. An approach based on scan statistics is discussed in \cite{Arias-Castro2017a}. 

For the proof of the results of Theorem~\ref{th:rank_hc_test}, it was not needed to explicitly derive the limiting null distribution of the rank-based higher criticism statistic in~\eqref{eq:def_rank_HC}. While its null distribution can be approximated for any sample size through Monte-Carlo sampling, exploring the limiting distribution of the statistic may be an interesting avenue for future research, which must carefully account for dependencies induced by the ranks. For the classical higher criticism statistic introduced in \cite{Donoho2004} this has been discussed in \cite{Gontscharuk2015}.

Our methodology is based on the higher criticism statistic. It might be fruitful to explore the use of other statistics such as those proposed by \cite{Berk1979} --- which \cite{moscovich2016exact} apply to the homoscedastic version of normal testing problem~\eqref{hyp:norm} --- or \cite{Stepanova2014}, but when observation ranks are used. 

Another direction, already discussed in Section~\ref{sec:application} is related to identification of anomalous subjects based on the rank $p$-values. Depending on the context it may be more appropriate to control the false discovery rate \citep{benjamini1995} instead of the more stringent family wise error rate in such problems --- which the higher criticism already allows to control. Dependencies induced by the ranks should be accounted for, and since the joint null distribution of the collection of subject $p$-values is known after rank-transform, an approach exploiting this knowledge is likely more powerful than general methods that can deal with dependencies such as that of \cite{benjamini2001}.

\section{Proof of Theorem~\ref{th:rank_hc_test}}\label{sec:proof-th1}

\myparagraph{Statement~\ref{en:th1_null}} the arguments for conservativeness of the test hold trivially by the definition of the $p$-value in~\eqref{eq:def_rank_HC_p_value}.

\myparagraph{Statement~\ref{en:th1_alt}} we must show that for any fixed $\alpha\in (0,1)$ and for any of the instances of the alternative described in the theorem, for the $p$-value in~\eqref{eq:def_rank_HC_p_value} we have that
\begin{equation}\label{eq:p-value-convergence}
\Pha{\cP_{\text{rank-hc}}(\bR) \leq \alpha } \to 1 \ .
\end{equation}
We show this in two parts; first, we bound the $p$-value $\cP_{\text{rank-hc}}(\bR)$ by the statistic $T(\bR)$. This argument is common over all three substatements~\ref{en:th1_strong},~\ref{en:th1_boundary}, and~\ref{en:th1_weak}, and rests on a permutation argument. In the second part, we show that the statistic $T(\bR)$ is sufficiently large such that the above convergence holds. This part is much more intricate, and the proof strategy diverges between the three substatements. For substatement~\ref{en:th1_strong}, a sharp result is most easily obtained by focusing on a single subject from $\cS_{T,n}$, and requires characterizing the tail probability of a single subject rank mean $Y_i(\bR)$. Zooming in on a single subject is not sufficient for statements~\ref{en:th1_boundary} and~\ref{en:th1_weak}, which subsequently requires characterizing rank-induced dependencies. This is done through a bound with a surrogate statistic that carries no problematic dependencies. The proof then requires a sharp analysis of the size of the surrogate statistic. While the analysis showing the statistic $T(\bR)$ is sufficiently large is different for all three substatements, a common part is the fundamental need to characterize the size of the normalizing quantity $p_q$, and in each substatement we must carefully contrast quantities with this normalization constant.

For convenience throughout the proof we define:
\[
s \equiv \abs{\cS} \ .
\]
We start by bounding $\cP_{\text{rank-hc}}(\bR)$ by the statistic $T(\bR)$. This bound can be conveniently constructed through a permutation argument, so we first restate our rank-based $p$-value as a permutation $p$-value. To do so, consider the set of permutations that only permute observations within their respective referential as:
\begin{equation}\label{eq:def-permset}
\Pi' \equiv \{\pi \in \Pi: \forall_{i\in[n],j\in[t]} \  \exists_{k\in[n]} : \pi(i,j) = (k,j) \} \ ,
\end{equation}
where $\Pi$ is the set of all permutations of index set $\{(i,j), i\in[n], j\in[t]\}$. Denote the ranks permuted with $\pi\in\Pi'$ by $\bR^\pi$. Then, if $\pi$ is sampled uniformly at random from $\Pi'$ and if $\br_1$ is any fixed set of ranks, we have that
\begin{equation}\label{eq:dist-null-perm}
\left( \bR \mid H_0 \right) \overset{d}{=} \br_1^\pi \ ,
\end{equation}
where $\overset{d}{=}$ indicate these two quantities have the same distribution. With this is mind, for any fixed set of ranks $\br_1$ we have 
\[
p_q = \P{\frac{Y_1(\br_1^\pi) - \Rb}{\sigma_R} \geq \sqrt{\frac{2q \log(n)}{t}}}  \ ,
\]
and
\[
\cP_{\text{rank-hc}}(\bR) = \P{ T(\bR^\pi)\geq T(\bR) \given \bR} \ .
\]
Now, we start by simplifying the role of $\pi$ in the above expression, resembling the approach taken in \cite{arias-castro2018a, stoepker2021}. This leads to the following bound on our $p$-value depending on the size of the grid $k_n$ which is proven in Section~\ref{sec:bound-p-value}:

\begin{lem}\label{lem:bound-p-value}
The following inequality holds almost surely:
\[
\cP_{\text{rank-hc}}(\bR) \leq \frac{k_n}{T(\bR)^2} + \ind{T(\bR) \leq 0} \ .
\]
\end{lem}
We can use the result of this lemma to bound our objective~\eqref{eq:p-value-convergence} as
\begin{align*}
&\Pha{\cP_{\text{rank-hc}}(\bR) \leq \alpha } \\
&\geq \Pha{\frac{k_n}{T(\bR)^2} + \ind{T(\bR) \leq 0} \leq \alpha \given T(\bR) \geq \sqrt{\frac{k_n}{\alpha}}}\Pha{ T(\bR) \geq \sqrt{\frac{k_n}{\alpha}}} \\
&= \Pha{ T(\bR) \geq \sqrt{\frac{k_n}{\alpha}}} \ .\numberthis\label{eq:p-val-to-stat}
\end{align*}
For brevity throughout the remainder define
\begin{equation}\label{eq:def-kn}
\tilde k_n \equiv \sqrt{\frac{k_n}{\alpha}} \ .
\end{equation}
To show~\eqref{eq:p-value-convergence}, first note that for any sequence $q_n\in Q_n$ we have that:
\begin{equation}\label{eq:bound-TR}
\P{T(\bR) \geq \tilde k_n} = \P{\max_{q\in Q_n} V_q(\bR) \geq \tilde k_n} \geq \PP\Big( V_{q_n}(\bR) \geq \tilde k_n\Big) \ .
\end{equation}
To show the test has power converging to one it therefore suffices to show that for any instance of the alternative (under the conditions of the theorem), there exists a sequence $q_n \in Q_n$ such that $V_{q_n}(\bR)$ is sufficiently large.

At this point, it is most convenient to restrict ourselves to the proof of substatement~\ref{en:th1_strong}, and proceed to prove the remaining statements afterwards.

\myparagraph{Statement~\ref{en:th1_strong}} It suffices to show~\eqref{eq:bound-TR} for a lower bound of $V_{q_n}(\bR)$. An obvious lower bound can be constructed by observing that:
\begin{equation}\label{eq:lb-nq-trivial}
N_q(\bR) \geq \ind{\frac{Y_i(\bR) - \overline R}{\sigma_R} \geq \sqrt{\frac{2q\log(n)}{t}}} \equiv M_{i,q}(\bR) \ ,
\end{equation}
where $i\in\cS_{T,n}$ which always exists since $\liminf_{n\to\infty} \abs{\cS_{T,n}} \geq 1$ by assumption. Define for convenience:
\[
b_q \equiv \tilde k_n\sqrt{np_q(1-p_q)} + np_q \ .
\]
Then~\eqref{eq:lb-nq-trivial} implies
\begin{equation}\label{eq:th1-subcaseii-suff}
\P{V_{q}(\bR) \geq \tilde k_n} \geq \P{M_{i,q}(\bR) \geq b_{q}} \ .
\end{equation}
Our aim is to apply Chebyshev's inequality to show the probability above converges to one. To do so, the following two lemmas provide characterizations of the quantities involved which are proven in Section~\ref{sec:bound-pq} and Section~\ref{sec:bound-Emiq} respectively.

\begin{lem}\label{lem:bound-pq}
\[
p_q \leq n^{-q + \frac{\sqrt{6}}{3}\sqrt{\frac{q^3\log(n)}{t}}} \ .
\]
\end{lem}

\begin{lem}\label{lem:bound-EMiq}
Consider the alternative in~\eqref{hyp:general}. If $\frac{s}{n}\sqrt{\frac{t}{\log(n)}} = o(1)$ and $q_n \to q$, and if for $i\in\cS$ we have $\mu_{n,i} \geq \sqrt{2r\log(n)/t}$ for some $r > q$, then
\[
\Pha{\frac{Y_{i}(\bR) - \overline R}{\sigma_R} \geq \sqrt{\frac{2q_n\log(n)}{t}}} \to 1 \ .
\]
\end{lem}

Lemma~\ref{lem:bound-pq} is a consequence of Bernstein's inequality for bounded random variables. Lemma~\ref{lem:bound-EMiq} requires relating the expectation of $Y_i(\bR)$ to the signal magnitudes $\mu_{n,i}$, and is then a consequence of Chebyshev's inequality.

Now, choose $q_n = 1+\eta/2 + o(1)$. Since $k_n \to \infty$ and $1+\eta/2 < 2$, for this choice $q_n \in Q_n$. Note that for $t=\omega(\log(n))$ lemma~\ref{lem:bound-pq} implies $p_q \leq n^{-1+o(1)}$. Since $k_n = n^{o(1)}$ and thus $\tilde k_n = n^{o(1)}$, and since $\eta > 0$, for this choice of $q_n$ we have $b_{q_n} \to 0$. Furthermore, since $t=o(n)$ and $s = \bigO(\sqrt{n})$, Lemma~\ref{lem:bound-EMiq} implies for this choice $\Eha{M_{i,q}(\bR)} \to 1$. All of the above then implies
\[
\Eha{M_{i,q}(\bR)} - b_{q_n} \to 1  \ ,
\]
such that for $n$ sufficiently large the above quantity is positive. Then, an application of Chebyshev's inequality for this choice of $q_n$ gives
\begin{align*}
\Pha{M_{i,q_n}(\bR) \leq b_{q_n}} &= \Pha{M_{i,q_n} - \Eha{M_{i,q_n}} \leq b_{q_n} - \Eha{M_{i,q_n}}} \\
&\leq \Pha{\abs{M_{i,q_n} - \Eha{M_{i,q_n}}} \geq \Eha{M_{i,q_n}} - b_{q_n}} \\
&\leq \frac{\Varha{M_{i,q_n}}}{(\Eha{M_{i,q_n}} - b_{q_n})^2} \\
&= \frac{\Eha{M_{i,q_n}}(1-\Eha{M_{i,q_n}})}{(\Eha{M_{i,q_n}} - b_{q_n})^2} \to 0 \ .
\end{align*}
The above implies~\eqref{eq:th1-subcaseii-suff} converges to 1 which completes the proof of statement~\ref{en:th1_strong}.

\myparagraph{Remaining statements~\ref{en:th1_boundary} and~\ref{en:th1_weak}} For these statements a crude lower bound as in~\eqref{eq:lb-nq-trivial} does not suffice. To lower-bound the probability in~\eqref{eq:bound-TR}, inspiration from earlier works \citep{Donoho2004, stoepker2021} suggests applying Chebyshev's inequality directly on $V_q(\bR)$ could be effective. However, characterizing the variance of $V_q(\bR)$ is not straightforward due to dependencies between the summands of $N_q(\bR)$ induced by ranking. One may hope that the covariances between the summands would be negative using an argument based on negative association \citep{Joag-Dev1983}. However, negative association of the ranks is only true under the null hypothesis, and one can even construct examples where the sum of the summand covariances is orders of magnitude larger than the sum of the summand variances (see Appendix~\ref{app:cov-nq}). This implies that Chebyshev's inequality is loose unless one is willing to impose severe restrictions on the class of distributions considered.

To overcome this issue we first lower-bound $N_q(\bR)$ by a surrogate statistic with similar characteristics, but with independent summands. This can be done by observing that the ranks can be stated as a sum of conditionally independent indicators, which ultimately leads to the following bound which is proven in Section~\ref{sec:bound-Nq}:
\begin{lem}\label{lem:bound-Nq}
Consider the alternative in~\eqref{hyp:general}. Define:
\[
\tilde N_q(\bX) \equiv \sum_{i\in[n]} \ind{\frac{\frac{n-s}{t}\sum_{j\in[t]}U_{ij} - \overline R}{\sigma_R} \geq \sqrt{\frac{2q\log(n)}{t}} + \frac{\varepsilon_{n,i}}{\sigma_R} } \ ,
\]
with $U_{ij}$ defined as in~\eqref{eq:def_Uij} and the sequence
\begin{equation}\label{eq:def-varepsilon-1}
\varepsilon_{n,i} \equiv 2\sqrt{\frac{n\log(n)}{t}} + 1 - \frac{1}{t}\sum_{j\in[t]}\sum_{k\in\cS}\E{\lim_{\delta\to 0} F_{kj}^{(\delta)}(X_{ij}^{(\delta)})} = \bigO\left(\sqrt{\frac{n\log(n)}{t}} + s\right) \ ,
\end{equation}
where $F_{ij}^{(\delta)}$ is the convolution of the distribution $F_{ij}$ and the $\delta$-dilated distribution of $W_{ij}$ resulting from the tie-breaking transformation~\eqref{eq:def_tiebreak}. Analogously, $X_{ij}^{(\delta)}$ is a sample from $F_{ij}^{(\delta)}$.

When $s \leq \sqrt{n}$ there exists a non-negative sequence $\xi_n = o(1)$ for which
\[
\forall x \in \bbR: \quad \Pha{N_q(\bR) \leq x} \leq \Pha{\tilde N_q(\bX) \leq x} + \xi_n \ .
\]
\end{lem}
The beauty of the above lemma is that unlike the summands of $N_q(\bR)$, the summands of $\tilde N_q(\bX)$ are no longer dependent. As such, the variance\footnote{It may be insightful to view the result of Lemma~\ref{lem:bound-Nq} in light of a bias-variance tradeoff argument; essentially, the statistic $\tilde N_q(\bX)$ is biased, but has lower variance than $N_q(\bR)$. Our subsequent analysis capitalizes on this lower variance through Chebyshev's inequality, but must downstream bound the effect of the bias induced. The success of this argument hinges on the fact that the induced bias is small in our context.} of $\tilde N_q(\bR)$ can be easily characterized, and thus we are in good shape to proceed with Chebyshev's inequality. Define for brevity:
\begin{equation}\label{eq:def-aq}
a_q \equiv \tilde k_n\sqrt{np_q(1-p_q)} + \Big(np_q - \Eha{\tilde N_q(\bX)} \Big) \ .
\end{equation}
such that the bound of Lemma~\ref{lem:bound-Nq} and Chebyshev's inequality imply that, if $a_q < 0$, we have:
\begin{align*}
\Pha{V_q(\bR) \leq k_n} &= \Pha{N_q(\bR) \leq k_n\sqrt{np_q(1-p_q)} + np_q} \\
&\leq \Pha{\tilde  N_{q}(\bX) \leq k_n\sqrt{np_q(1-p_q)} + np_q} + o(1) \\
&= \Pha{\tilde N_{q}(\bX) - \Eha{\tilde N_{q}(\bX)} \leq a_q} + o(1) \\
&\leq  \Pha{ \abs{\tilde N_{q}(\bX) - \Eha{\tilde N_{q}(\bX)}} \geq -a_{q}}  + o(1) \\
&\leq \frac{\Varha{\tilde N_{q}(\bX)}}{a_{q}^2}  + o(1)\ . \numberthis\label{eq:bound-Vq}
\end{align*}
Now, combining the bound above with~\eqref{eq:p-val-to-stat} and~\eqref{eq:bound-TR} implies 
\begin{equation}\label{eq:bound-p-value-3}
\Pha{\cP_{\text{rank-hc}}(\bR) \leq \alpha } \geq 1 - \frac{\Varha{\tilde N_{q}(\bX)}}{a_{q}^2} + o(1) \ ,
\end{equation}
provided $a_q < 0$. 
Therefore, to complete the proof, we must show that there exists a $q_n \in Q_n$ such that
\begin{equation}\label{eq:req-endingth1i}
a_{q_n} < 0 \text{ for $n$ sufficiently large, and}
\end{equation}
\begin{equation}\label{eq:req-endingth1ii}
\frac{\Varha{\tilde N_{q_n}(\bX)}}{a_{q_n}^2} \to 0 \ .
\end{equation}
In order to establish~\eqref{eq:req-endingth1i} and~\eqref{eq:req-endingth1ii} for specific sequences $q_n$, we must first characterize the relevant quantities in $a_q$ and the variance of $\tilde N_q(\bX)$ under the alternative. At this stage, it is useful to remark that $\tilde N_{q_n}(\bX)$ can be simplified as: 
\begin{align*}
\tilde N_{q_n}(\bX) &= \sum_{i\in[n]}\ind{ \frac{1}{t}\sum_{j\in[t]} U_{ij} - \frac{1}{2} \geq \frac{\sigma_R}{n-s} \sqrt{\frac{2q_n\log(n)}{t}} + \frac{\varepsilon_{n,i} + \overline R}{n-s} - \frac{1}{2}} \\ \numberthis\label{eq:nq-simplified}
&= \sum_{i\in[n]}\ind{ \frac{2\sqrt{3}}{t}\sum_{j\in[t]} \left(U_{ij} - \tfrac{1}{2}\right) \geq \sqrt{\frac{2q_n\log(n)}{t}}\big(1 + \delta_{n,i,q_n}\big)} \ ,
\end{align*}
where
\begin{align*}
\delta_{n,i,q_n} &\equiv \left(\frac{2\sqrt{3}\sigma_R}{n-s} - 1\right) + 2\sqrt{3}\Bigg(\frac{\varepsilon_{n,i} + \overline R}{n-s} - \frac{1}{2} \Bigg)\sqrt{\frac{t}{2q_n\log(n)}} \numberthis\label{eq:def-delta}\\
&= \bigO\left(\frac{1}{\sqrt{q_n n}} + \frac{s}{n} +\frac{s}{n}\sqrt{\frac{t}{q_n\log(n)}} \right) \ ,
\end{align*}
where we used~\eqref{eq:def_Rbar_sigmaR} and the fact that $\varepsilon_{n,i} = \bigO(\sqrt{n\log(n)/t} + s)$. The above implies that if $q_n \to q > 0$ and $s \leq \sqrt{n}$ and $t = o(n)$, then for $n$ sufficiently large
\begin{equation}\label{eq:bound-delta}
\max_{i\in[n]}\abs{\delta_{n,i,q_n}} \leq \sqrt{\frac{1}{\log(n)}}  = o(1)\ .
\end{equation}
In many places throughout the proof, the simpler implication that $\max_{i\in[n]} \delta_{n,i,q_n} = o(1)$ will suffice. In the proof of statement~\ref{en:th1_boundary} the bound above plays a more prominent role, and for both remaining statements the result of Lemma~\ref{lem:prob-norm-null-alt-approx} is critical, which itself depends critically on the the detailed characterization of $\delta_{n,i,q_n}$ in~\eqref{eq:def-delta}.

It is now useful to introduce the following notation:
\begin{equation}\label{eq:def-vqi}
\Pha{\frac{2\sqrt{3}}{t}\sum_{j\in[t]}U_{ij} -\frac{1}{2} \geq \sqrt{\frac{2q\log(n)}{t}}(1+\delta_{n,i,q})} \equiv \begin{cases} v_{q,i} &\text{ if } i \in\cS \ , \\ \tilde p_q &\text{ if } i\notin\cS \ . \end{cases} 
\end{equation}
Since $\tilde N_q(\bX)$ is a sum of independent indicators, and the variance of indicators are bounded by their first moment, we obtain that
\begin{align*}
\Eha{\tilde N_q(\bX)} &= (n-s)\tilde p_q + \sum_{i\in\cS} v_{q,i}  \ \numberthis\label{eq:expected-Nqtilde} , \\
\Varha{\tilde N_q(\bX)} &\leq n\tilde p_q + \sum_{i\in\cS} v_{q,i}\ . \numberthis\label{eq:variance-Nqtilde}
\end{align*}
So $(n-s)\tilde p_q$ is the contribution to the expected value of $\tilde N_q(\bX)$ originating from the null subjects. For our methodology to work well at the boundary of detection, it is reasonable to suspect that one will require this to be close to the contribution to the expected value of $N_q(\bX)$ of $n-s$ subjects to under the null hypothesis, which is given by $(n-s)p_q$. In other words, we require $\tilde p_q$ to be close to $p_q$. We now formalize this intuitive requirement; using~\eqref{eq:expected-Nqtilde} we can rewrite $a_q$ from~\eqref{eq:def-aq} as:
\begin{equation}\label{eq:def-aq-2}
a_q =  \tilde k_n\sqrt{np_q(1-p_q)} + s\tilde p_q + n(p_q - \tilde p_q) - \sum_{i\in\cS} v_{q,i} \ ,
\end{equation}
which implies that a sufficient condition for requirement~\eqref{eq:req-endingth1i} is that there exists a sequence $q_n \in Q_n$ such that
\begin{equation}\label{eq:req-endingth1-2a}
\frac{\tilde k_n\sqrt{np_{q_n}} + s\tilde p_{q_n} + n(p_{q_n} - \tilde p_{q_n})}{\sum_{i\in\cS} v_{{q_n},i}} \to 0 \ .
\end{equation}
Under the condition above and using~\eqref{eq:variance-Nqtilde} we conclude that~\eqref{eq:req-endingth1ii} is satisfied provided we have for the same sequence $q_n$ that
\begin{equation}\label{eq:req-endingth1-2b}
\frac{n\tilde p_{q_n} + \sum_{i\in\cS} v_{q_n,i}}{\left(\sum_{i\in\cS} v_{{q_n},i}\right)^2} \to 0 \ .
\end{equation}
Showing~\eqref{eq:req-endingth1-2a} and~\eqref{eq:req-endingth1-2b} requires an upper bound on both $p_{q_n}$ and $\tilde p_{q_n}$, as well as their difference $(p_{q_n} - \tilde p_{q_n})$. The following two lemmas provide such characterizations:

\begin{lem}\label{lem:prob-null-char}
Consider the alternative in~\eqref{hyp:general}. If $i\not\in\cS$ then
\[
\tilde p_{q} \leq  n^{-q(1+\delta_{n,i,q}) + (1+\delta_{n,i,q})^2\frac{\sqrt{6}}{3}\sqrt{\frac{q^3\log(n)}{t}}}  \ .
\]
\end{lem}

\begin{lem}\label{lem:prob-norm-null-alt-approx}
Consider the alternative in~\eqref{hyp:general}. If $q_n \to q > 0$ and $t = \omega(\log(n))$ and $s \leq \sqrt{n}$, then
\[
p_{q_n} - \tilde  p_{q_n} \leq n^{-q + \bigO(\sqrt{t/(n^2\log(n))}) + o(1)}\bigO\left(\frac{\sqrt{t}}{n} + \sqrt{\frac{\log(n)}{n}} + \frac{\sqrt{t}}{n}\sum_{i\in\cS}\mu_{n,i}\right) \ ,
\]
where $\mu_{n,i}$ is defined in~\eqref{eq:def_mu}.
\end{lem}
The proofs of these results are deferred to Sections~\ref{sec:prob-null-char} and~\ref{sec:prob-norm-null-alt-approx} respectively. The proof of Lemma~\ref{lem:prob-null-char} follows from an application of Bernstein's inequality, and requiring such asymptotic characterizations is standard in this field. Proving Lemma~\ref{lem:prob-norm-null-alt-approx} is considerably more involved, hinging crucially on the sharp bound obtained in Lemma~\ref{lem:bound-Nq}. Particularly, the higher order terms we elsewhere abstracted in $\delta_{n,i,q}$ play a crucial role in the proof of Lemma~\ref{lem:prob-norm-null-alt-approx}. It is noteworthy to remark that the core of the arguments needed to prove Lemma~\ref{lem:prob-norm-null-alt-approx} are not due to the introduction of the surrogate statistic $\tilde N_q(\bX)$; even without its introduction, we fundamentally need to contrast the tail probabilities of subject rank means under the null hypothesis with the same tail probabilities of null streams under the alternative hypothesis. It is insufficient to have these probabilities be of the same order; to prove our results, we need a sharp characterization of their difference. Lemma~\ref{lem:prob-norm-null-alt-approx} provides this characterization, and additionally bounds additional discrepancies between the quantities that may have been introduced by the surrogate statistic. Lemma~\ref{lem:prob-norm-null-alt-approx} also elucidates how the degree of accuracy of this approximation is influenced by the size of $t$. Its proof clarifies that if one would impose a much sharper constraint that $t = n^{o(1)}$, then one could get the result $p_{q_n} - \tilde p_{q_n} \leq n^{-\frac{1}{2}-q +o(1)}$ (which suffices for the upcoming arguments) without carefully characterizing the higher order terms in $\delta_{n,i,q}$.

At this point, what is left is to show that the contribution to the expected value of $\tilde N_q(\bX)$ by the \textit{anomalous} subjects (i.e. $sv_{q,i}$) is large enough to outweigh the relevant nominal variability and approximation terms captured by Lemma~\ref{lem:bound-pq}, Lemma~\ref{lem:prob-null-char} and Lemma~\ref{lem:prob-norm-null-alt-approx}. This requires particularizing to a well-chosen threshold value $q$ and subsequently characterizing the probability $v_{q,i}$ sufficiently sharply. As this highly depends on the nature of the alternative, the proofs of statements~\ref{en:th1_boundary} and~\ref{en:th1_weak} diverge, and we start with the former. For both statements, it suffices (and will be necessary) to assume that $\abs{\cS_{T,n}} = \bigO(1)$.

\myparagraph{Statement~\ref{en:th1_boundary}} Under this scenario, we choose
\[
q_n = \argmax\left\{q\in Q_n : q < 1-\frac{1}{\log^{1/4}(n)} \right\} \ = \frac{\floor{2k_n\left(1-\frac{1}{\log^{1/4}(n)}\right)}}{2k_n} \leq 1-\frac{1}{\log^{1/4}(n)} \ .
\]
If $k_n \to \infty$ then $q_n \to 1$. Note that
\[
\max_{i\in[n]}\mu_{n,i} \leq \sqrt{3} \ , \text{ and } \max_{i\in\cS_{B,n}\cup\cS_{W,n}} \mu_{n,i} \leq \sqrt{\frac{2(1+\eta)\log(n)}{t}} \ .
\]
Since we have assumed $\abs{\cS_{T,n}} = \bigO(1)$ we can bound
\begin{align*}
\frac{\sqrt{t}}{n}\sum_{i\in\cS} \mu_{n,i} &\leq \frac{\sqrt{t}}{n} \left(\abs{\cS_{T,n}}\sqrt{3} + \abs{\cS_{W,n}\cup\cS_{B,n}} \sqrt{\frac{2(1+\eta)\log(n)}{t}}\right) \\
&\leq \sqrt{t}n^{-1 + o(1)} + n^{-\frac{1}{2} + o(1)} \ , \numberthis\label{eq:bound-sum-mu}
\end{align*}
where we used that $s \leq \sqrt{n}$. Then Lemma~\ref{lem:prob-norm-null-alt-approx} implies that for this choice of $q_n$ we have
\[
n(p_{q_n} - \tilde p_{q_n}) \leq n^{-\frac{1}{2}+\bigO(\sqrt{t/n^2\log(n)})}\left(\sqrt{\frac{t}{n}} + 1\right) \to 0\ ,
\]
where the convergence holds under the assumption $t = o(n)$. Next, we will use the following lemma:
\begin{lem}\label{lem:prob-strong-char}
Consider the alternative in~\eqref{hyp:general}. Let $\mu_{n,i} \geq \sqrt{2r\log(n)/t}$. Then $v_{q_n,i} \to 1$ if either one of the following conditions hold:
\begin{enumerate}
\item $q_n \to q < r\ .$
\item $q_n \leq r - \xi_n$ and $\xi_n > 0$ and $\xi_n = \omega\left(\sqrt{\frac{1}{\log(n)}}\right)\ .$
\end{enumerate}
\end{lem}
Since $\min_{i\in\cS_{B,n}} \mu_{n,i} \geq \sqrt{2\log(n)/t}$, and $q_n \leq 1-\log^{-1/4}(n)$, application of Lemma~\ref{lem:prob-strong-char} implies that $\min_{i\in\cS_{B,n}} v_{q_n,i} \to 1$ such that for $n$ sufficiently large
\[
\sum_{i\in\cS} v_{q_n,i} \geq \frac{\abs{\cS_{B,n}}}{2} \ .
\]
Then the requirements~\eqref{eq:req-endingth1-2a} and~\eqref{eq:req-endingth1-2b} are satisfied if for our choice of $q_n$ we have that
\begin{equation}\label{eq:final-req-statii}
\frac{\tilde k_n\sqrt{np_{q_n}} + s\tilde p_{q_n} + \sqrt{n\tilde p_{q_n}} + 1}{\abs{\cS_{B,n}}} \to 0 \ .
\end{equation}
Since $\floor{x} \geq x - 1$ for $x>0$, we have that $q_n \geq 1-\frac{1}{\log^{1/4}(n)} - \frac{1}{2k_n}$. Then Lemma~\ref{lem:bound-pq} implies
\begin{equation}\label{eq:statement-ii-pq}
p_q \leq n^{-1 + \log^{-1/4}(n) + \frac{1}{2k_n} + \frac{\sqrt{6}}{3}\sqrt{\frac{\log(n)}{t}}} \ .
\end{equation}
Using~\eqref{eq:bound-delta} and Lemma~\ref{lem:prob-null-char} we have
\begin{align*}
\tilde p_q &\leq n^{-(1-\abs{\delta_{n,i,q_n}})(1- \log^{-1/4}(n) -\frac{1}{2k_n}) + (1+\delta_{n,i,q_n})^2\frac{\sqrt{6}}{3}\sqrt{\frac{\log(n)}{t}}} \\
&\leq n^{-1 +  \log^{-1/4}(n) + \log^{-1/2}(n) + \frac{1}{2k_n} + (1+\delta_{n,i,q_n})^2\frac{\sqrt{6}}{3}\sqrt{\frac{\log(n)}{t}}} \ . \numberthis\label{eq:statement-ii-tildepq}
\end{align*}
Now, using~\eqref{eq:statement-ii-pq} and~\eqref{eq:statement-ii-tildepq}, recalling the definition of $\tilde k_n$ in~\eqref{eq:def-kn}, and using the bound $(1+\delta_{n,i,q_n})^2\sqrt{6}/6 \leq 1/2$ which holds for $n$ sufficiently large, the convergence in~\eqref{eq:final-req-statii} holds provided
\[
\abs{\cS_{B,n}} = \omega\Big(n^{\frac{1}{2}\log(n)^{-1/4} +\frac{1}{4k_n} + \frac{1}{2}\sqrt{\frac{\log(n)}{t}}}\big(\sqrt{k_n} + n^{\frac{1}{2}\log^{-1/2}(n)}\big)\Big) \ ,
\] 
completing the proof of statement~\ref{en:th1_boundary}.

\myparagraph{Statement~\ref{en:th1_weak}}
Let $q_n \to q > 0$, where we will particularize the value of $q$ depending on the value of $r$ and $\gamma$. First, since $\beta >1/2$ and through our assumption $\abs{\cS_{T,n}} = \bigO(1)$, the bound in~\eqref{eq:bound-sum-mu} and Lemma~\ref{lem:prob-norm-null-alt-approx} imply if $t = o(n)$ and $t=\omega(\log(n))$ we can bound
\begin{equation}\label{eq:prob-norm-null-alt-approx-implied}
n(p_{q_n} - \tilde  p_{q_n}) \leq n^{\frac{1}{2}-q+o(1)} \ ,
\end{equation}
and Lemma~\ref{lem:prob-null-char} and~\eqref{eq:bound-delta} implies that if $t = \omega(\log(n))$ and $t = o(n)$ we have
\[
\tilde p_{q_n} \leq n^{-q + o(1)} \ .
\]
Lemma~\ref{lem:bound-pq} implies for $t=\omega(\log(n))$ that
\begin{equation}\label{eq:prob-purenull-bound}
p_{q_n} \leq n^{-q+o(1)} \ .
\end{equation}
Using all of the above and noting that we have $\frac{1-q}{2} \geq \frac{1}{2} - q$ we find that the two requirements in~\eqref{eq:req-endingth1-2a} and~\eqref{eq:req-endingth1-2b} are satisfied if there exists a sequence $q_n \in Q_n$ with $q_n \to q > 0$ such that:
\begin{equation}\label{eq:final-req}
\frac{n^{\frac{1-q}{2} + o(1)} + 1}{\sum_{i\in\cS} v_{q_n,i}} \to 0 \ .
\end{equation}

We pause to remark that the statement above captures the requirements for the test to have asymptotic power without resorting to considering minimal signal magnitude and heteroskedasticity. However, the implicit nature of the requirement above (through the existence of $q_n \in Q$) calls for a more concrete result, which we give through considering minimal signal magnitude and heteroskedasticity.

To proceed, we need to carefully consider two subcases separately: the subcase where $r > 2\beta - 1$ and $\gamma \geq 0$, and the subcase where $\rho(\beta,\gamma) < r \leq 2\beta - 1$ for $\gamma > 0$. Note that the first subcase implies the result of the theorem for $\gamma=0$, since $\rho(\beta,0) = 2\beta - 1$.

We start with the case where $r > 2\beta - 1 = \rho(\beta,0)$. In that case, we choose $q = \min\{r - \frac{1}{2}(r-(2\beta - 1)), 1\} \in (0,1]$. Now, since $k_n \to \infty$ we have that there exists a $q_n \in Q_n$ such that $q_n \to q$.

Since $r \in (2\beta - 1, 1)$, we have that $q < r$, such that Lemma~\ref{lem:prob-strong-char} implies that $\min_{\cS}\{v_{q_n,i}\} \to 1$ yielding the following lower bound for $n$ sufficiently large:
\[
\sum_{i\in\cS} v_{q_n,i} \geq \frac{s}{2} = n^{1-\beta + o(1)} \ .
\]
Then~\eqref{eq:final-req} holds if:
\[
\max\left\{ \frac{1-q}{2}, 0 \right\} < 1-\beta \ .
\]
Now, since $\beta < 1$ the only binding condition is the first one. For our choice of $q$ the above is trivially true if the minimum is attained at $q=1$. In the other case, the above is true if
\[
\frac{1}{2} - \frac{1}{2}\Bigg(r - \frac{1}{2}\Big(r-(2\beta - 1)\Big)\Bigg)  < 1-\beta \ \Longrightarrow \ r > 2\beta - 1 \ ,
\]
which is true by assumption.

Next, assume $\rho(\beta,\gamma) < r \leq 2\beta - 1$ and $\gamma > 0$. Since $\beta < 1$ we have that $r < 1$. In this case, we choose $q$ to be fixed value from $(r,1]$ which we will particularize later. Once again, since $k_n \to \infty$ and $q \in (0,1]$, for this choice of $q$ there exists a sequence $q_n \in Q_n$ such that $q_n \to q$. However, Lemma~\ref{lem:prob-strong-char} does not apply as $r < q$. The following lemma characterizes the behavior of $v_{q,i}$ in this case:
\begin{lem}\label{lem:prob-alt-char}
Consider the alternative in~\eqref{hyp:general}. Let $q_n \to q > 0$. Let $\delta_{n,i,q} = o(1)$ for $i\in\cS$. If $\mu_{n,i} = \sqrt{2r_i(1+o(1))\log(n)/t}$ for some $r_i < q$ and if $t = \omega(\log(n))$ and if $\sigma_{n,i}$ in~\eqref{eq:def_var} is bounded away from zero, then:
\[
v_{q_n,i} \geq n^{-\frac{1}{\sigma^2_{n,i}}(\sqrt{q}-\sqrt{r_i})^2 + o(1)} \ .
\]
\end{lem}
The proof of this lemma is provided in Section~\ref{sec:prob-alt-char}. Note that the bound in the first result of Lemma~\ref{lem:prob-alt-char} is increasing in $\sigma_{n,i}$ and $r_i$, such that
\[
\min_{i\in\cS_{W,n}}\{v_{q_n,i}\} \geq n^{-\frac{1}{\gamma}(\sqrt{q}-\sqrt{r})^2 + o(1)} \ .
\]
For $n$ sufficiently large, we can now lower bound $v_{q,i}$ as
\begin{align*}
\sum_{i\in\cS} v_{q_n,i} &= \sum_{i\in\cS_{W,n}}v_{q_n,i} + \sum_{i\in\cS_{T,n}}v_{q_n,i} \\
&\geq (s-\abs{\cS_{T,n}})\min_{i\in\cS_{W,n}}\{v_{q_n,i}\} + \frac{\abs{\cS_{T,n}}}{2} \\
&\geq n^{1-\beta-\frac{1}{\gamma^2}(\sqrt{q}-\sqrt{r})^2+o(1)} - \frac{\abs{\cS_{T,n}}}{2} \ ,
\end{align*}
where we bounded $v_{q_n,i} \geq 1/2$ for $i\in\cS_{T,n}$, and trivially bounded $v_{q_n,i} \leq 1$ in the final inequality. Since $\abs{\cS_{T,n}} = \bigO(1)$ the bound above implies the convergence in~\eqref{eq:final-req} holds if
\begin{equation}\label{eq:algebra-system}
\max\left\{ \frac{1-q}{2}, 0 \right\} < 1-\beta-\frac{1}{\gamma^2}(\sqrt{q}-\sqrt{r})^2 \ .
\end{equation}

We then have the following:
\begin{lem}\label{lem:algebra-system}
Consider Equation~\eqref{eq:algebra-system} with $r < 1$.
\begin{enumerate}
\item If $\gamma^2 < 2$ and
\[
r > \rho_1(\beta,\gamma) \equiv \begin{cases} 
(2-\gamma^2)(\beta - \frac{1}{2})\ , & 1/2 < \beta \leq 1-\gamma^2/4\ , \\ 
(1 - \gamma\sqrt{1-\beta})^2\ ,		&  1-\gamma^2/4 < \beta < 1 \ ,
\end{cases}
\]
then the system~\eqref{eq:algebra-system} holds for $q = \min\{\frac{4r}{(\gamma^2-2)^2}, 1 \}$ with $q > r$.
\item If $\gamma^2 \geq 2$ and
\[
r > \rho_2(\beta,\gamma) \equiv \begin{cases}
0 \ , 								& 1/2 < \beta \leq 1-1/\gamma^2 \ ,\\
(1 - \gamma\sqrt{1-\beta})^2 \ ,		& 1-1/\gamma^2 < \beta < 1 \ ,
\end{cases}
\]
then the system~\eqref{eq:algebra-system} holds for $q = 1$. If $\gamma > 2$ and $\beta\in(1/2, 1-1/\gamma^2)$ then the system also holds for $r=0$.
\end{enumerate}
\end{lem}
A similar result was previously derived in \cite{TonyCai2011}. For completeness the proof of the lemma is given in Section~\ref{sec:algebra-system}. The lemma implies that under our assumptions on $r$ and for the particular choices of $q  > r$ specified, the convergence~\eqref{eq:final-req} holds. The proof of the lemma boils down to a careful case separation. The only intricacies lie in arguing that there is no need to consider a value of $q$ larger than 1; while the value of $q = \frac{4r}{\gamma^2-2}$ are asymptotically optimal for a large signal regime, this value can grow large for $\gamma^2$ close to 2. Nevertheless, we can argue that in such cases, the choice $q=1$ also suffices. \qed

\subsection{Proof of Lemma~\ref{lem:bound-p-value}}\label{sec:bound-p-value}

We begin by using a union bound for the max-operator in the permutation statistic, resulting in a sum with $k_n$ terms, i.e. as many as the grid $Q_n$:
\[
\cP_{\text{rank-hc}}(\bR) = \P{ \max_{q \in Q_n} V_q(\bR^\pi) \geq T(\bR) \given \bR } \leq \sum_{q\in Q_n} \P{V_q(\bR^\pi) \geq T(\bR) \given \bR } \ .
\]
In the above and the derivations that follow note that probability operators are defined with respect to the distribution of $\pi$, and all statements are conditional on a given set of ranks~$\bR$. 

Note that the only random quantity inside the final probability above (conditionally on $\bR$) is $N_q(\bR^\pi)$ within the term $V_q(\bR^\pi)$ through the randomness of the permutations. Since $\Eha{N_q(\bR^\pi) \given \bR} = np_q$ we have
\[
\cP_{\text{rank-hc}}(\bR)\leq \sum_{q\in Q_n} \P{N_q(\bR^\pi) - \E{N_q(\bR^\pi)} \geq T(\bR)\sqrt{n p_q(1- p_q)} \given \bR} \ .
\]
One may apply Chebyshev's inequality to the probability term in the expression above, provided the right-hand-side of the inequality inside the probability is positive. In the case that $T(\bR)$ is negative, we trivially bound the $p$-value by one. Therefore, using Chebyshev's inequality we get
\begin{equation}\label{eq:bound-p-value-chebyshev}
\cP_{\text{rank-hc}}(\bR) \leq \ind{T(\bR) \leq 0}+\frac{\ind{T(\bR) > 0}}{T(\bR)^2}\sum_{q\in Q_n} \frac{\Var{N_q(\bR^\pi) \given \bR}}{np_q(1-p_q)}  \ ,
\end{equation}
where we convention that $0/0 = 0$. To continue, we must quantify the conditional variance of $N_q(\bR^\pi)$. For ease of notation, we first define:
\begin{equation}\label{eq:def_zq}
z_q \equiv \Rb + \sqrt{\frac{2q\sigma_R^2 \log(n)}{t}}\ ,
\end{equation}
such that $N_q(\bR^\pi)$ is a sum of the indicator variables $\ind{ Y_i(\bR^\pi) \geq z_q}$, which are dependent through the permutations $\pi$. Nevertheless, this dependency is benign. Resorting to \citet[Theorem 2.11]{Joag-Dev1983} and the properties $\text{P}_7$, $\text{P}_6$ and $\text{P}_4$ in that result, we find that conditionally on $\bR$ the random variables $Y_i(\bR^\pi)$ and $Y_k(\bR^\pi)$ are negatively associated when $i\neq k$. We may now bound the variance of $N_q(\bR^\pi)$ as:
\begin{align*}
\Var{N_q(\bR^\pi) \given \bR} &= \sum_{i\in[n]} \Var{\ind{Y_i(\bR^\pi) \geq z_q}\given \bR}\\
&\quad\qquad + \sum_{i\in[n]}\sum_{k\neq i} \Cov{\ind{Y_i(\bR^\pi) \geq z_q}, \ind{Y_k(\bR^\pi) \geq z_q} \given \bR} \\
& \leq np_q(1-p_q) \ ,
\end{align*}
where we used the definition of negative association \citep[Definition 2.1]{Joag-Dev1983}. Using the bound on the variance above in~\eqref{eq:bound-p-value-chebyshev} gives our result.
\qed 

\subsection{Proof of Lemma~\ref{lem:bound-pq}}\label{sec:bound-pq}
We prove this statement through Bernstein's inequality (see, for example, \citet[Appendix A, Section 4.5]{ShoWel}):

\begin{lem}[Bernstein inequality for bounded random variables]\label{lem:bernstein-bounded}
Let $X_1,\dots,X_n$ be independent zero-mean random variables. Suppose that for all $i$, we have $\abs{X_i} \leq M$ almost surely. Then, for all positive $\zeta$, we have:
\[
\P{\sum_{i\in[n]} X_i \geq \zeta} \leq \exp{-\frac{\frac{1}{2}\zeta^2}{\sum_{i\in[n]}\E{X_i^2} + \frac{1}{3}M\zeta}} \ .
\]
\end{lem}

Let $D_j$ denote i.i.d. random variables with discrete uniform distribution over $[n]$. Note that under the null hypothesis,
\[
Y_1(\bR) \overset{d}{=} \frac{1}{t}\sum_{j\in[t]}D_j \ .
\]
Then since $\abs{\sqrt{\frac{12}{n^2-1}}\left(D_j - \frac{n+1}{2}\right)} \leq \sqrt{3}$, application of Lemma~\ref{lem:bernstein-bounded} implies
\begin{align*}
p_q &\equiv \Phn{\frac{Y_1(\bR) - \Rb}{\sigma_R} \geq \sqrt{\frac{2q \log(n)}{t}} } \\
&= \P{\sum_{j\in[t]} \sqrt{\frac{12}{n^2-1}}\left(D_j - \frac{n+1}{2}\right) \geq \sqrt{2qt\log(n)}} \\
&\leq \exp{\frac{-q\log(n)}{1+ \frac{\sqrt{6}}{3}\sqrt{q\log(n)/t}}} \\
&\leq \exp{-q\left(1-\frac{\sqrt{6}}{3}\sqrt{\frac{n^2}{n^2-1}}\sqrt{\frac{q\log(n)}{t}} \right)\log(n)} \\
&= n^{-q + \frac{\sqrt{6}}{3}\sqrt{\frac{q^3\log(n)}{t}}} \ ,
\end{align*}
where the first inequality is due to Lemma~\ref{lem:bernstein-bounded} and the second inequality holds due to $\frac{1}{1+x} \geq 1-x$ if $x > -1$. \qed

\subsection{Proof of Lemma~\ref{lem:bound-EMiq}}\label{sec:bound-Emiq}
Define $W_{i} \equiv (Y_{i}(\bR) - \overline R)/\sigma_R$ and ${\tau_n \equiv \sqrt{2q_n\log(n)/t}}$ for brevity. Then our probability of interest can be rewritten and bounded as
\begin{align*}
1-\Pha{W_{i} \leq \tau_n} &= 1-\Pha{W_{i} - \Eha{W_{i}} \leq \tau_n - \Eha{W_i}} \\
&= 1-\Pha{-(W_{i} - \Eha{W_{i}}) \geq \Eha{W_{i}}-\tau_n} \\
&\geq 1-\Pha{\abs{W_{i}- \Eha{W_{i}}} \geq \Eha{W_{i}} - \tau_n} \numberthis\label{eq:rawranks-alt-prob-1}\ .
\end{align*}
We now aim to apply Chebyshev's inequality on the final term and show it converges to 0, which implies our result. To do so we first find a lower bound on $\E{R_{ij}}$ and subsequently $\E{W_i}$. Note first that the moments of the individual ranks can be related to the expected value of $U_{ij}$ as:
\begin{align*}
\Eha{R_{ij}} &= \Eha{\lim_{\delta\to 0} \sum_{k\in[n]} \ind{X_{kj}^{(\delta)} \leq X_{ij}^{(\delta)}} } \\
&\geq \Eha{\lim_{\delta\to 0} \sum_{k\not\in\cS} \ind{X_{kj}^{(\delta)} \leq X_{ij}^{(\delta)}} }
= (n-s)\Eha{U_{ij}} \ ,
\end{align*}
where $X_{ij}^{(\delta)}$ is the convolution of the observation $X_{ij}$ and the $\delta$-dilated distribution of $W_{ij}$ resulting from the tie-breaking transformation~\eqref{eq:def_tiebreak}. Using the above and recalling the definition of $\mu_{n,i}$ in~\eqref{eq:def_mu} we can bound
\begin{align*}
\Eha{W_{i}} &\geq \sqrt{\frac{12}{n^2-1}}\left(\frac{n-s}{t}\sum_{j\in[t]}\left(\Eha{U_{{i}j}} -\tfrac{1}{2}\right) - \frac{s+1}{2}\right) \\
&\geq (1 - \tfrac{1}{2n^2})(1-\tfrac{s}{n})\mu_{n,i} - \frac{\sqrt{3}(s+1)}{n} \ . \\
&= (1+o(1))\mu_{n,i}
\end{align*}
where the final order statements hold by $\frac{s}{n}\sqrt{\frac{t}{\log(n)}} = o(1)$. Now, since $q_n \to q$, and $\mu_{n,i} \geq \sqrt{2r\log(n)/t}$ with $r > q$ this implies there exists a constant $c > 0$ such that for sufficiently large $n$ we have
\begin{align*}
\Eha{W_{i}} - \tau_n &\geq \left(\sqrt{r}(1+o(1)) -\sqrt{q_n} \right)\sqrt{\frac{2\log(n)}{t}} \geq c\sqrt{\frac{2\log(n)}{t}} \ .
\end{align*}
By Popoviciu's inequality on variances, $\max_{i\in[n],j\in[t]}\Var{R_{ij}} \leq \frac{(n-1)^2}{4}$ and thus 
\[
\Varha{W_{i}} \leq \frac{12}{n^2-1}\cdot \frac{(n-1)^2}{4} \leq \frac{3}{t} \ .
\]
Then application of Chebyshev's inequality on the final term in~\eqref{eq:rawranks-alt-prob-1} implies:
\[
\Pha{\abs{W_{i} - \Eha{W_{i}}} \geq \Eha{W_{i}} - \tau_n} \leq \frac{\Varha{W_{i}}}{(\Eha{W_{i}} - \tau_n)^2} \leq \frac{3}{t} \cdot \frac{c^2t}{2\log(n)} \to 0 \ ,
\]
which together with the bound in~\eqref{eq:rawranks-alt-prob-1} completes the proof. \qed

\subsection{Proof of Lemma~\ref{lem:bound-Nq}}\label{sec:bound-Nq}

The statement of the lemma suggests that our rank subject means $Y_i(\bR)$ are lower bounded by $\frac{n-s}{t}\sum_{j\in[t]} U_{ij}$, minus a discrepancy term $\varepsilon_{n,1}$ of smaller order. To prove this, the following event will naturally play a critical role:
\[
M_i \equiv \left\{ Y_i(\bR) \geq \frac{n-s}{t}\sum_{j\in[t]}U_{ij} +\frac{1}{t}\sum_{j\in[t]}\sum_{k\in\cS}\E{\lim_{\delta\to 0} F_{kj}^{(\delta)}(X_{ij}^{(\delta)})} - 2\sqrt{\frac{n\log(n)}{t}} -1 \right\} \ .
\]
Letting $x \in \bbR$, the event $M_i$ above comes into play as follows:
\begin{align*}
\P{N_q(\bR) \geq x} &\geq \P{N_q(\bR) \geq x \given \bigcap_{i\in[n]} M_i} \P{\bigcap_{i\in[n]} M_i} \\
&\geq \P{\tilde N_q(\bX) \geq x \given \bigcap_{i\in[n]} M_i} \P{\bigcap_{i\in[n]} M_i} \\
&= \P{\tilde N_q(\bX) \geq x} - \P{\tilde N_q(\bX) \geq x \given \bigcup_{i\in[n]} M_i^c} \P{\bigcup_{i\in[n]} M_i^c} \\
&\geq  \P{\tilde N_q(\bX) \geq x} - \P{\bigcup_{i\in[n]} M_i^c} \\
&\geq \P{\tilde N_q(\bX) \geq x} - \sum_{i\in[n]}\P{M_i^c} \ , \numberthis\label{eq:bound-Nq}
\end{align*}
where the last line is due to a union bound. The bound above implies the statement of the lemma when the second term converges to zero. To show this, we need to lower bound the ranks. Note that we can decompose the value of individual ranks as:
\begin{align*}
R_{ij} &= \lim_{\delta \to 0} \sum_{k\in[n]}\ind{X_{kj}^{(\delta)} \leq X_{ij}^{(\delta)}} \\
&= (n-s)U_{ij} + \lim_{\delta \to 0}  \underbrace{\sum_{k\in\cS} F_{kj}^{(\delta)}(X_{ij}^{(\delta)})}_{\equiv V_{ij}^{(\delta)}} + \sum_{k\in[n]}\left( \ind{X_{kj}^{(\delta)} \leq X_{ij}^{(\delta)}} - F_{kj}^{(\delta)}(X_{ij}^{(\delta)}) \right) \ .
\end{align*}
This decomposition implies the following lower bound for our rank means $Y_i(\bR)$:
\begin{align*}
Y_i(\bR) &= \frac{n-s}{t}\sum_{j\in[t]} U_{ij} + \lim_{\delta\to 0} \frac{1}{t}\sum_{j\in[t]}V_{ij}^{(\delta)} + \frac{1}{t}\sum_{j\in[t]}\sum_{k\in[n]}\ind{X_{kj}^{(\delta)} \leq X_{ij}^{(\delta)}} -  F_{kj}^{(\delta)}(X_{ij}^{(\delta)}) \\
&\geq \frac{n-s}{t}\sum_{j\in[t]} U_{ij} +  \lim_{\delta\to 0} \frac{1}{t}\sum_{j\in[t]}V_{ij}^{(\delta)} - \underbrace{\abs{\frac{1}{t}\sum_{j\in[t]}\sum_{k\in[n]}\ind{X_{kj}^{(\delta)} \leq X_{ij}^{(\delta)}} - F_{kj}^{(\delta)}(X_{ij}^{(\delta)}) }}_{\equiv S_i^{(\delta)}} \ .
\end{align*}
To continue with our bound for $\P{M_i^c}$, we construct two events 
$A_i$ and $B_i$ which (due to the lower bound on our rank means $Y_i(\bR)$ above) jointly imply $M_i^c$, thereby reducing the bound on $\P{M_i^c}$ to bounding the probabilities of those events. Specifically, if we define the two events:
\begin{align*}
A_i &\equiv \left\{\lim_{\delta\to 0} S_i^{(\delta)} \leq \sqrt{\frac{n\log(n)}{t}} + 1 \right\} \ , \\ 
B_i &\equiv \left\{\frac{1}{t}\sum_{j\in[t]}V_{ij}^{(\delta)} \geq \E{\frac{1}{t}\sum_{j\in[t]}V_{ij}^{(\delta)}} - \sqrt{\frac{n\log(n)}{t}} \right\} \ ,
\end{align*}
then $\P{M_i^c}$ can be upper bounded as follows:
\begin{align*}
\P{M_i^c} = 1-\P{M_i} &\leq 1-\P{M_i\given A_i \cap B_i}\P{A_i \cap B_i} \\
&= 1-\P{A_i \cap B_i} \\
&= \P{A_i^c \cup B_i^c} \\
&\leq \P{A_i^c} + \P{B_i^c} \ ,
\end{align*}
where the final line follows from the union bound. To conclude the proof, we must therefore find a sufficiently sharp upper bound for $\P{A_i^c}$ and $\P{B_i^c}$. 

We start with event $A_i^c$. It is critical to observe that, conditional on the observations $\bX_i^{(\delta)} \equiv (X_{i1}^{(\delta)},X_{i2}^{(\delta)},\dots,X_{it}^{(\delta)})$, the (bounded) indicator summands of $S_i^{(\delta)}$ are independent, and their (conditional) expectation is given by $F_{kj}^{(\delta)}(X_{ij}^{(\delta)})$. This suggests that the conditional probability of $A_i^c$ can be upper bounded using Hoeffding's inequality \citep{hoeffding1963}.

Before proceeding with the bound, however, it is important to remark that the conditional sum $S_i^{(\delta)} \mid \bX_i^{(\delta)}$ has $t$ deterministic observations of size bounded by 1, such that we can initially upper bound using the triangle inequality as follows:
\begin{align*}
\big( S_i^{(\delta)} \mid \bX_i^{(\delta)} \big) &\leq \abs{\frac{1}{t}\sum_{j\in[t]}\sum_{k\in[n], k\neq i}\Big(\ind{X_{kj}^{(\delta)} \leq X^{(\delta)}_{ij}} - F_{kj}^{(\delta)}(X^{(\delta)}_{ij}) \Big)} + 1\ .
\end{align*}
Since the indicators are bounded between $0$ and $1$, the bound above and Hoeffding's inequality implies:
\begin{align*}
&\P{S_i^{(\delta)} \leq \sqrt{\frac{n\log(n)}{t}} + 1 \given \bX_i^{(\delta)}} \\
&\leq \P{\abs{\frac{1}{t}\sum_{j\in[t]}\sum_{k\in[n],k\neq i}\ind{X_{kj}^{(\delta)} \leq X_{ij}^{(\delta)}} - F_{kj}^{(\delta)}(X_{ij}^{(\delta)}) } \geq \sqrt{\frac{n\log(n)}{t}}\given \bX_i^{(\delta)}} \\
&\leq 2\exp{-2\cdot\frac{nt\log(n)}{(n-s-1)t}} \\
&= 2 \exp{-2\log(n)\big(1+\bigO\big(\tfrac{s}{n}\big)\big)} = n^{-2 + o(1)} \ ,
\end{align*}
where we have used that $s/n = o(1)$. Since the bound above does not depend on the value of $\bX_i^{(\delta)}$ or $\delta$, this means we can bound the unconditional probability of the event $A_i^c$ by the same quantity as well.

We now move on to the event $B_i^c$. Note that $V_{ij_1}^{(\delta)}$ and $V_{ij_2}^{(\delta)}$ are independent whenever $j_1 \neq j_2$, and have bounded support on $[0,s]$, such that Hoeffding's inequality results in:
\begin{align*}
\P{\sum_{j\in[t]} V_{ij}^{(\delta)} - \E{V_{ij}^{(\delta)}} \leq -\sqrt{nt\log(n)}} &\leq \P{\abs{\sum_{j\in[t]} V_{ij}^{(\delta)} - \E{V_{ij}^{(\delta)}}} \geq \sqrt{nt\log(n)}} \\
&\leq 2\exp{\frac{-2nt\log(n)}{s^2t}} \\
&= n^{-2\frac{n}{s^2}+o(1)} \leq n^{-2 + o(1)} \ ,
\end{align*}
where the final inequality is due to $s \leq \sqrt{n}$ and thus $n/s^2 \geq 1$. Note that the bound does not depend on $\delta$. Summarizing, we obtain that
\[
\sum_{i\in[n]}\P{M_i^c} \leq \sum_{i\in[n]}\P{A_i^c} + \P{B_i^c} \leq \sum_{i\in[n]} n^{-2 + o(1)} = o(1) \ ,
\]
which together with~\eqref{eq:bound-Nq} completes the statement.
\qed

\subsection{Proof of Lemma~\ref{lem:prob-null-char}}\label{sec:prob-null-char}
Note that if $i\not\in\cS$, then $U_{ij}$ are i.i.d. uniformly distributed on $[0,1]$. Furthermore, $\abs{2\sqrt{3}(U_{ij} - \tfrac{1}{2})} \leq \sqrt{3}$ and has mean zero and variance one, such that our desired upper bound can be obtained using the Bernstein inequality of Lemma~\ref{lem:bernstein-bounded} as: 
\begin{align*}
\tilde p_{q} &= \Pha{ \sum_{j\in[t]} 2\sqrt{3}\left(U_{ij} - \frac{1}{2}\right) \geq \sqrt{2qt\log(n)}(1+\delta_{n,i,q})} \\
&\leq \exp{-\frac{q(1+\delta_{n,i,q})\log(n)}{1 + \frac{\sqrt{6}}{3}(1+\delta_{n,i,q})\sqrt{q\log(n)/t}}} \\
&\leq \exp{-q(1+\delta_{n,i,q})\left(1- (1+\delta_{n,i,q})\frac{\sqrt{6}}{3}\sqrt{\frac{q\log(n)}{t}} \right)\log(n)} \\
&= n^{-q(1+\delta_{n,i,q}) + (1+\delta_{n,i,q})^2\frac{\sqrt{6}}{3}\sqrt{\frac{q^3\log(n)}{t}}} \ ,
\end{align*}
where the first inequality is due to Lemma~\ref{lem:bernstein-bounded} and the second inequality is due to $\frac{1}{1+x} \geq 1-x$ if $x > -1$. \qed

\subsection{Proof of Lemma~\ref{lem:prob-norm-null-alt-approx}}\label{sec:prob-norm-null-alt-approx}

We start by upper bounding $p_q$, which subsequently requires us to upper bound $Y_1(\bR) \lvert H_0$. Let $D_j(n)$ be i.i.d. random variables with discrete uniform distribution on $\{1,\dots,n\}$ and let $A_j([a,b])$ be i.i.d. random variables with continuous uniform distribution supported on $[a,b]$. The latter notation explicitly includes the support, which may change at each appearance. Under the null, we can bound:
\begin{align*}
Y_1(\bR) \overset{d}{=} \frac{1}{t}\sum_{j\in[t]} D_j(n) &\leq  \frac{1}{t}\sum_{j\in[t]} D_j(n) + A_j([0,1]) \\
&= \frac{1}{t}\sum_{j\in[t]}A_j([1,n+1]) = 1 + \frac{n}{t}\sum_{j\in[t]} A_j([0,1]) \ .
\end{align*}
The upper bound above subsequently leads to an upper bound for $p_{q_n}$ as:
\begin{align*}
p_{q_n} &= \P{\frac{\frac{1}{t}\sum_{j\in[t]} D_{j}(n) - \overline R}{\sigma_R} \geq \sqrt{\frac{2q_n\log(n)}{t}}} \\
&\leq \P{\frac{1 + \frac{n}{t}\sum_{j=1}^t A_j([0,1]) - \overline R}{\sigma_R} \geq \sqrt{\frac{2q_n\log(n)}{t}}} \\
&= \P{\frac{1}{\sqrt{t}}\sum_{j\in[t]} A_j([-\tfrac{1}{2},\tfrac{1}{2}]) \geq \underbrace{\frac{-\sqrt{t}}{2n} + \frac{\sigma_R}{n}\sqrt{2q_n\log(n)} }_{\equiv g_n}} \ .
\end{align*}

We now lower bound $\tilde p_{q_n}$. Unlike other places within the proof of Theorem~\ref{th:rank_hc_test}, to prove the statement of this lemma we need to carefully consider the higher-order terms abstracted inside $\delta_{n,i,q}$ in the definition~\eqref{eq:def-vqi} to obtain our result. Note that if $i\not\in\cS$, then $U_{ij}$ are i.i.d. with continuous uniform distribution on $[0,1]$. For convenience, denote $S \equiv \frac{1}{\sqrt{t}}\sum_{j=1}^t A_j([-\tfrac{1}{2},\tfrac{1}{2}])$. Recalling the detailed statement from~\eqref{eq:nq-simplified} and letting $i\not\in\cS$, we can restate $\tilde p_{q_n}$ as follows:
\begin{align*}
\tilde p_{q_n} &= \P{S \geq  \frac{\sigma_R}{n-s} \sqrt{2q_n\log(n)} + \sqrt{t}\Big(\frac{\varepsilon_{n,i}+\overline R}{n-s} - \frac{1}{2}\Big)} \\
&= \P{S \geq  g_n + \frac{\sqrt{t}}{2n} + \frac{s\sigma_R}{n(n-s)}\sqrt{2q_n\log(n)} + \sqrt{t}\Big(\frac{\varepsilon_{n,i}+\overline R}{n-s} - \frac{1}{2}\Big)} \ . \numberthis\label{eq:pqtilde-intermediate}
\end{align*}
To continue, we must characterize the final terms inside the probability above. Recalling the definition of $\varepsilon_{n,i}$ in~\eqref{eq:def-varepsilon-1}, to characterize this term the following claim will be important:
\begin{clm}\label{clm:u-inverse}
Consider the definition in~\eqref{eq:def_Uij}. Let $k,i \in [n]$ and $j\in [t]$. Then
\[
\E{ \lim_{\delta\to 0} F_{kj}^{(\delta)}(X_{ij}^{(\delta)})} = 1- \E{\lim_{\delta\to 0} F_{ij}^{(\delta)}(X_{kj}^{(\delta)})} \ .
\]
\end{clm}
We defer the proof until the end of this subsection. The claim is obvious if all observations are continuous; otherwise the validity of comes from the tie-breaking procedure described in Section~\ref{sec:ranking}, which retains the ordering among unique observations and breaks ties uniformly at random. Using Claim~\ref{clm:u-inverse} we find that, if $i\not\in\cS$:
\begin{align*}
\E{\frac{1}{t}\sum_{j\in[t]}\sum_{k\in\cS} \lim_{\delta\to 0} F_{kj}^{(\delta)}(X_{ij}^{(\delta)})} &= \sum_{k\in\cS}\Bigg(1-  \frac{1}{t}\sum_{j\in[t]} \E{\lim_{\delta\to 0} F_{ij}^{(\delta)}(X_{kj}^{(\delta)})}\Bigg) \\ 
&= \sum_{k\in\cS}\Bigg(1-  \frac{1}{t}\sum_{j\in[t]} \E{U_{kj}}\Bigg) \\ 
&= \sum_{k\in\cS}\Bigg(1-  \Bigg(\frac{1}{2} + \frac{\mu_{n,k}}{2\sqrt{3}} \Bigg)\Bigg) \\
&= \frac{s}{2} - \frac{1}{2\sqrt{3}}\sum_{k\in\cS}\mu_{n,k} \ .
\end{align*}
Then, we can simplify for $i\not\in\cS$:
\begin{align*}
\frac{\varepsilon_{n,i}+\overline R}{n-s} - \frac{1}{2} &=\frac{1}{n-s}\Bigg(2\sqrt{\frac{n\log(n)}{t}} +1 - \frac{s}{2} + \frac{1}{2\sqrt{3}}\sum_{k\in\cS}\mu_{n,k} + \frac{n+1}{2} - \frac{n-s}{2}\Bigg) \\
&= \frac{1}{n-s}\left(2\sqrt{\frac{n\log(n)}{t}} + \frac{1}{2\sqrt{3}}\sum_{k\in\cS}\mu_{n,k} + \frac{3}{2}\right) \ .
\end{align*}
Now, following the threshold in Equation~\eqref{eq:pqtilde-intermediate} and the derivation above, define for convenience:
\begin{align*}
h_n \equiv \frac{\sqrt{t}}{2n} + \frac{s\sigma_R}{n(n-s)}\sqrt{2q_n\log(n)} + \frac{\sqrt{t}}{n-s}\left(2\sqrt{\frac{n\log(n)}{t}} +\frac{1}{2\sqrt{3}}\sum_{k\in\cS}\mu_{n,k} + \frac{3}{2}\right) \ .
\end{align*}
Note that 
\begin{equation} \label{eq:hn_bigO}
h_n = \bigO\left(\frac{\sqrt{t}}{n} + \sqrt{\frac{\log(n)}{n}} + \frac{\sqrt{t}}{n}\sum_{k\in\cS}\mu_{n,k}\right),
\end{equation} 
when $s \leq \sqrt{n}$ and $q_n \to q$. Using the above definition, we can characterize $\tilde p_{q_n}$ by:
\begin{align*}
\tilde p_{q_n} &\geq \P{S \geq g_n + h_n} \ .
\end{align*}
Denote with $F_S$ the cumulative distribution function of the Irwin-Hall distribution, i.e.~the distribution of the normalized mean of independent standard uniform random variables $S$. Using all of the above, we can thus upper-bound our quantity of interest as:
\begin{align*}
p_{q_n} - \tilde p_{q_n} &\leq F_S(g_n + h_n) - F_S(g_n) \ .
\end{align*}
Note that $F_S$ is continuous, and denote the probability density function by $f_S(x)$. We can therefore use the mean-value theorem to characterize:
\[
F_S(g_n + h_n) - F_S(g_n) = f_S(c)h_n \ ,
\]
for some $c \in [g_n, g_n+h_n]$. Furthermore, since $A_j([-\tfrac{1}{2},\frac{1}{2}])$ is unimodal and symmetric, the unimodality is preserved under convolution, and as such $F_S$ is unimodal with mode at $0$. Therefore, $f_S(x)$ is nonincreasing for all $x > 0$, leading to the bound:
\[
F_S(g_n + h_n) - F_S(g_n) \leq f_S(g_n)h_n \ .
\]
To continue, we must therefore bound $f_S(h_n)$. We have that:
\begin{align*}
1-F_S(g_n-1) &= \int_{g_n-1}^{g_n} f_S(x)dx + \int_{g_n}^\infty f_S(x)dx \\
&\geq f_S(g_n)\int_{g_n-1}^{g_n}dx + (1-F_S(g_n)) \\
&= f_S(g_n) + (1-F_S(g_n)) \\
&\geq f_S(g_n) \ ,
\end{align*}
where the first inequality follows from $f_S(x)$ unimodal and thus nonincreasing for $x > 0$. Now that we have $f_S(g_n) \leq 1-F_S(g_n-1)$, we can continue our bound for $f_S(g_n)$. First note that for $n$ sufficiently large and $q_n \to q$ we have that:
\begin{align*}
(g_n - 1) &= -\frac{\sqrt{t}}{2n} + \frac{\sigma_R}{n}\sqrt{2q_n\log(n)} - 1 \\
&= \sqrt{\frac{2q_n\log(n)}{12}}\left(\frac{\sqrt{n^2-1}}{n} - \sqrt{\frac{6}{q_n\log(n)}} - \frac{\sqrt{3t}}{n\sqrt{2q_n\log(n)}} \right) \\
&\geq \sqrt{\frac{2q_n\log(n)}{12}}\left(1 + o(1) - \sqrt{\frac{3t}{2n^2q\log(n)}} \right) \ .
\end{align*} 
Furthermore, if $t = \omega(\log(n))$ then $(g_n-1)/\sqrt{t} \to 0$.
We can now use Lemma~\ref{lem:bernstein-bounded} to bound $1-F_S(g_n-1)$ as follows:
\begin{align*}
1-F_S(g_n-1) &= \P{\sum_{j=1}^t A_j([-1/2,1/2]) \geq \sqrt{t}(g_n-1)} \\
&\leq \exp{\frac{-\frac{1}{2}t(g_n-1)^2}{\frac{t}{12} + \bigO(\sqrt{t}(g_n - 1))}} \\
&=  \exp{\frac{-\frac{12}{2}(g_n-1)^2}{1 + \bigO((g_n - 1)/\sqrt{t})}} \\
&= n^{-\frac{6}{\log(n)}(g_n-1)^2(1+o(1))} \\
&\leq n^{-q_n\left(1 +o(1) - \sqrt{3t/(2n^2q\log(n))} \right)(1+o(1))} \\
&= n^{-q + \bigO(\sqrt{t/(n^2\log(n))}) + o(1)} \ .
\end{align*}
Combining all of the above, we have that:
\begin{align*}
p_{q_n} - \tilde p_{q_n} &\leq F_S(g_n + k_n) - F_S(g_n) \\
&\leq f_S(g_n)h_n \\
&\leq (1-F_S(g_n-1))h_n \\
&\leq n^{-q + \bigO(\sqrt{t/(n^2\log(n))}) + o(1)} h_n \ ,
\end{align*}
and we conclude using~\eqref{eq:hn_bigO}.
\qed

\myparagraph{Proof of Claim~\ref{clm:u-inverse}} It can be easily seen (see Lemma~\ref{lem:moments-u} in~\ref{app:moments-u}) that
\[
\E{ \lim_{\delta\to 0} F_{kj}^{(\delta)}(X_{ij}^{(\delta)})} = \frac{1}{2}\P{X_{ij} = X_{kj}} + \P{X_{kj} < X_{ij}} \ .
\]
Using this fact
\begin{align*}
\frac{1}{2}\P{X_{ij} = X_{kj}} + \P{X_{kj} < X_{ij}}
&= \frac{1}{2}\P{X_{ij} = X_{kj}} + 1 - \P{X_{kj} \geq X_{ij}} \\
&= -\frac{1}{2}\P{X_{ij} = X_{kj}} + 1 - \P{X_{kj} > X_{ij}} \\
&= 1 - \E{ \lim_{\delta\to 0} F_{ij}^{(\delta)}(X_{kj}^{(\delta)})} \ .
\end{align*} \qed

\subsection{Proof of Lemma~\ref{lem:prob-strong-char}}
We suppress the subscript $i$ for readability. For brevity throughout the remainder define $\tau_n \equiv \sqrt{\frac{2q_n\log(n)}{t}}(1+\delta_{n,q_n})$ and $\overline Z \equiv \frac{2\sqrt{3}}{t}\sum_{j\in[t]}U_{ij} - \tfrac{1}{2}$. 

To show this lower bound we use Chebyshev's inequality. We start with the first condition. Our probability of interest can be bounded as:
\begin{align*}
1- v_{q_n} &= \P{\overline Z - \mu_n \leq \tau_n - \mu_n} \\
&= \P{- \left(\overline Z - \mu_n\right) \geq \mu_n - \tau_n} \\
&\leq \P{\abs{\overline Z - \mu_n} \geq \mu_n - \tau_n} \ . \numberthis\label{eq:vqlem-cheb1}
\end{align*}
To apply Chebyshev's inequality we must ensure $\mu_n - \tau_n$ is positive. Note that:
\[
\mu_{n} - \tau_n \geq \sqrt{\frac{2\log(n)}{t}}\left(\sqrt{r} - \sqrt{q} + o(1)\right) \ ,
\]
such that for $n$ sufficiently large there exists a constant $c > 0$ such that $\mu_n - \tau_n > c\sqrt{\frac{2\log(n)}{t}}$. Note that Popoviciu's inequality on variances implies $\Var{\overline Z} \leq 3/t$. Then continuing from~\eqref{eq:vqlem-cheb1} we find that for $n$ sufficiently large:
\[
1- v_{q_n} \leq \frac{3/t}{2c\log(n)/t} \to 0 \ ,
\]
which implies our result.

Next, assume $q_n \leq r - \xi_n$ with $\xi_n > 0$ and $\xi_n = \omega\left(\sqrt{\frac{1}{\log(n)}}\right)$. Then:
\begin{align*}
\mu_n - \tau_n &\geq \sqrt{\frac{2\log(n)}{t}}\left(\sqrt{r} - \sqrt{r-\xi_n}(1+\delta_{n,q_n})\right) \\
&\geq \sqrt{\frac{2\log(n)}{t}}\left(\sqrt{r} - \sqrt{r}(1-\frac{\xi_n}{2r})(1+\delta_{n,q_n})\right) \\
&\geq \sqrt{\frac{2\log(n)}{t}}\xi_n\left(\frac{1}{2\sqrt{r}} - \frac{\abs{\delta_{n,q_n}}}{2\sqrt{r}} - \frac{\abs{\delta_{n,q_n}}\sqrt{r}}{\xi_n} \right) \ \\
&\geq \sqrt{\frac{2\log(n)}{t}}\xi_n\left(\frac{1}{2\sqrt{r}}+o(1)\right) \ ,
\end{align*}
where the final order statement holds due to the assumption $\xi_n = \omega\left(\sqrt{\frac{1}{\log(n)}}\right)$ and~\eqref{eq:bound-delta}. Then, analogous to the arguments for the first setting we have
\[
1-v_q \leq \frac{3/t}{2\log(n)\xi_n^2(1+o(1))/t} \to 0 \ ,
\]
where the convergence is implied by the assumption on $\xi_n$. \qed

\subsection{Proof of Lemma~\ref{lem:prob-alt-char}}\label{sec:prob-alt-char}

Throughout the proof, we assume $i\in\cS$ and suppress the subscript $i$ for readability. For convenience, define $\tau_n \equiv \sqrt{2q_n(1+\delta_{n,q_n})^2\log(n)/t} = \sqrt{2q(1+o(1))\log(n)/t}$ by assumption on $q_n$ and $\delta_{n,q_n}$, and define:
\begin{equation}\label{eq:def_Z}
Z_{j} \equiv 2\sqrt{3}\left(U_{j} -\frac{1}{2}\right)  \ .
\end{equation}
Denote by $F_{Z_j}$ the distribution of $Z_j$. Note that $\abs{Z_{j}} \leq \sqrt{3}$. Denote $\overline Z \equiv \frac{1}{t}\sum_{j\in[t]} Z_{j}$. Using these definitions, we can restate our probability of interest $v_{q_n}$ as
\[
v_{q_n} = \P{\frac{1}{t}\sum_{j\in[t]}Z_{j} \geq \tau_n } = \P{\overline Z \geq \tau_n} \ .
\]
Denote $\mu_{n,j} \equiv \E{Z_j}$ and $\sigma_{n,j}^2 \equiv \Var{Z_j}$. We have that $\mu_n = \sqrt{2r(1+o(1))\log(n)/t}$ for some value $r < q$. Then, using the definition of $\mu_n$ and $\sigma_n$ in~\eqref{eq:def_mu} and~\eqref{eq:def_var} we have that:
\[
\E{\overline Z} = \frac{1}{t}\sum_{j\in[t]} \mu_{n,j} = \mu_n \text{ , and }
\Var{\overline Z} = \frac{1}{t^2}\sum_{j\in[t]} \sigma_{n,j}^2 = \frac{\sigma_n^2}{t} \ .
\]
We show a lower bound for $v_q$ using a tilting argument. To that end, first note that since $Z_{j}$ is bounded, all moments are finite, and subsequently define the moment generating function of $Z_j$ as
\[
\varphi_j(x) \equiv \E{\exp{xZ_j}} = 1 + \mu_{n,j}x + (\sigma_{n,j}^2+\mu_{n,j}^2)\frac{x^2}{2} + \bigO(x^3) \ .
\]
Now, define $\tilde Z_j$ as an exponentially tilted random variable with respect to $Z_j$ with parameter~$\xi_n$ potentially depending on $n$. In detail, the density of $\tilde Z_j$ is given by
\[
f_{\tilde Z_j}(x) = \exp{\xi_n x - \log{\varphi_j(\xi_n)}}
\]
with respect to the (reference measure) distribution $F_{Z_j}$ for some deterministic $\xi_n > 0$. Now, we aim to choose $\xi_n$ such that the average of our tilted random variables is equal to $\tau_n$:
\[
\E{\frac{1}{t}\sum_{j\in[t]} \tilde Z_j} = \tau_n \ .
\]
We now show such a choice of $\xi_n$ exists and provide its characterization. We first note that the moment generating function of $\tilde Z_j$ can be expressed as:
\begin{align*}
\tilde\varphi_j(x) \equiv \E{\exp{x\tilde Z_j}} &= \int \exp{xz}f_{\tilde Z_j}(z){\rm d} F_{Z_j}(z) \\
&= \int \exp{xz}\exp{\xi_n z - \log(\varphi_j(\xi_n))}{\rm d}F_{Z_j}(z) \\
&= \frac{\varphi_j(x + \xi_n)}{\varphi_j(\xi_n)} \ .
\end{align*}
The mean of $\tilde Z_j$ as a function of $\xi_n$ is therefore given by:
\begin{align*}
\E{\tilde Z_j} &= \frac{\frac{d}{dx}\varphi_j(x+\xi_n)\Bigr|_{x=0}}{\varphi_j(\xi_n)} = \frac{\mu_{n,j} + (\sigma_{n,j}^2 + \mu_{n,j}^2)\xi_n + \bigO(\xi_n^2)}{1 + \mu_{n,j}\xi_n + \bigO(\xi_n^2)} \\
&= \bigg(\mu_{n,j} + (\sigma_{n,j}^2 + \mu_{n,j}^2)\xi_n + \bigO(\xi_n^2)\bigg)\bigg(1-\mu_{n,j}\xi_n + \bigO(\xi_n^2)\bigg) \\ &= \mu_{n,j} + \sigma_{n,j}^2\xi_n + \bigO(\xi_n^2) \ .
\end{align*}
Therefore, the expected value of the average of our tilted random variables as a function of $\xi_n$ is given by
\begin{align*}
\E{\frac{1}{t}\sum_{j\in[t]} \tilde Z_j} &= \frac{1}{t}\sum_{j\in[t]}\mu_{n,j} + \frac{1}{t}\sum_{j\in[t]}\sigma_{n,j}^2\xi_n + \bigO(\xi_n^2) \\
&= \mu_n + \sigma^2_n\xi_n + \bigO\left(\xi_n^2\right) \ .
\end{align*}
Note that $\sigma_n^2$ is bounded away from zero by assumption, such that $\frac{1}{\sigma_n^2} = \bigO(1)$. Now, setting the expectation above to $\tau_n$, and solving for $\xi_n$ yields:
\begin{align*}
& \mu_n + \sigma_n^2\xi_n + \bigO\left(\xi_n^2\right) = \tau_n \\
\Longrightarrow \hspace{0.5cm} & \xi_n = \frac{\tau_n - \mu_n}{\sigma_n^2} + \bigO\left(\xi_n^2\right) \ .
\end{align*} 
At this point it is useful to remark on the asymptotic order of the terms involved. Note that since $r < q$ we have for $n$ sufficiently large $\mu_n < \tau_n$, and $\bigO(\mu_n) = \bigO(\tau_n)$. Since $\sigma_n = \bigO(1)$, the expression above subsequently implies that $\bigO(\xi_n) = \bigO(\tau_n) = \bigO(\mu_n)$.

We now move on to characterize the average variance of our tilted random variables. We have:
\[
\E{\tilde Z_j^2} = \frac{\frac{d^2}{dx^2}\varphi_j(x+\xi_n)\Bigr|_{x=0}}{\varphi_j(\xi_n)} = \frac{\sigma_{n,j}^2 + \mu_{n,j}^2 + \bigO(\xi_n)}{1 +\bigO(\xi_n)} = \sigma_{n,j}^2 + \mu_{n,j}^2 + \bigO(\xi_n) \ .
\]
Since $\E{\tilde Z_j}^2 = \mu_{n,j}^2 + \bigO(\mu_{n,j}\sigma_{n,j}^2\xi_n +\xi_n^2)$ we obtain:
\[
\Var{\tilde Z_j^2} = \E{\tilde Z_j^2} - \E{\tilde Z_j}^2 = \sigma_{n,j}^2 + \bigO(\xi_n + \mu_{n,j}\sigma^2_{n,j}\xi_n) \ .
\]
We subsequently obtain:
\[
\sigma_{\xi_n}^2 \equiv \frac{1}{t}\sum_{j\in[t]}\Var{\tilde Z_j} = \frac{1}{t}\sum_{j\in[t]} \big(\sigma_{n,j}^2 + \bigO(\xi_n + \mu_{n,j}\sigma_{n,j}^2\xi_n) \big) = \sigma_n^2 + \bigO\left(\xi_n\right) \ .
\]
We are now ready to use a tilting argument. Denote by $\tilde F_{\tilde Z_j}$ the distributions of $\tilde Z_j$. Now, one can use the tilted random variables $\tilde Z_j$ to lower-bound $v_q$ as follows:
\begin{align*}
v_{q} &= \P{\overline Z \geq \tau_n} \\
&= \int \ind{\frac{1}{t}\sum_{j\in[t]}z_{j} \geq \tau_n} {\rm d}F_{Z_1}(z_{1})\dotsb{\rm d}F_{Z_t}(z_{t})\\
&= \left(\prod_{j\in[t]} \varphi_j(\xi_n)\right) \int \ind{\frac{1}{t}\sum_{j\in[t]}z_{j} \geq \tau_n} \exp{-\xi_n\sum_{j\in[t]}z_{j}} {\rm d} F_{\tilde Z_1}(z_{1})\dotsb{\rm d} F_{\tilde Z_t}(z_{t}) \\
&= \left(\prod_{j\in[t]} \varphi_j(\xi_n)\right) \E{\ind{\frac{1}{t}\sum_{j\in[t]}\tilde Z_{j} \geq \tau_n}\exp{-\xi_n\sum_{j\in[t]}\tilde Z_{j}}} \\
&\geq \left(\prod_{j\in[t]} \varphi_j(\xi_n)\right) \E{\ind{0 \leq \frac{\frac{1}{t}\sum_{j\in[t]} \tilde Z_j - \tau_n}{\sqrt{\sigma^2_{\xi_n}/t}} \leq 1} \exp{-\xi_n \sum_{j\in[t]}\tilde Z_j}} \\
&\geq \left(\prod_{j\in[t]} \varphi_j(\xi_n)\right) \P{0 \leq \frac{\frac{1}{t}\sum_{j\in[t]} \tilde Z_j - \tau_n}{\sqrt{\sigma^2_{\xi_n}/t}} \leq 1} \exp{-\xi_n t(\tau_n+ \sqrt{\sigma^2_{\xi_n}/t})} \ . \numberthis\label{eq:vq-lb-1}
\end{align*}
To continue, we first find suitable characterizations for the first two terms in the expression above. We start with the first term in~\eqref{eq:vq-lb-1}. Using a Taylor expansion to approximate the logarithm as $\log(1+x) = x - x^2/2 + \bigO(x^3)$ for $x \geq 0$, we have:
\begin{align*}
\prod_{j\in[t]}  \varphi_{j}(\xi_n) &= \exp{\sum_{j\in[t]}\log{\varphi_j(\xi_n)}} \\
&= \exp{\sum_{j\in[t]} \mu_{n,j} \xi_n + (\sigma_{n,j}^2+\mu_{n,j}^2)\frac{\xi_n^2}{2} - \frac{\mu_{n,j}^2\xi_n^2}{2} + \bigO(\xi_n^3)} \\
&= \exp{ t\mu_n\xi_n + \frac{t\sigma_n^2\xi_n^2}{2} + \bigO\left(t\xi_n^3\right)}
\ . \numberthis\label{eq:vq-lb-2}
\end{align*}

We now turn our attention to the second term in~\eqref{eq:vq-lb-1}. To show our results, it suffices to lower-bound this probability term by a constant. This can be done conveniently through Lindeberg's central limit theorem (e.g., see \citet[Chapter VIII, Section 4, Theorem 3]{Feller1991} for example):

\begin{theorem}[Lindeberg's Central Limit Theorem]\label{th:clt-lindeberg}
Let $\{X_1,\dots,X_n\}$ be a sequence of independent random variables with distribution $F_i$ with zero mean and finite variance $\sigma_i^2$. Define $s_n^2 = \sum_{i\in[n]} \sigma_i^2$. Now, if Lindeberg's condition holds:
\begin{equation}\label{eq:lindeberg-condition}
\frac{1}{s_n^2} \sum_{i\in[n]} \int_{\abs{x} \geq ts_n} x^2 \d F_i(x) \to 0 \ ,
\end{equation}
then the sum $S_n = \frac{1}{s_n}\sum_{i\in[n]} X_i$ weakly converges to the normal distribution with zero mean and unit variance.
\end{theorem}

We apply the above theorem to the random variables $(\tilde Z_j - \E{\tilde Z_j})$. In the notation of the above theorem, we have that $s_n^2 = t\sigma_{\xi_n}^2$. In \citet[Chapter VIII, Section 4, Example (e)]{Feller1991}, it is remarked that condition~\eqref{eq:lindeberg-condition} holds provided that $s_n \to \infty$ and the observations are bounded. First, since $\sigma_n$ is assumed to be bounded away from zero and $t=\omega(\log(n))$, we easily have
\[
s_n = \sqrt{t}\sigma_{\xi_n} = \sqrt{t}\sqrt{\sigma_n^2 + \bigO\left(\xi_n\right)} \to \infty \ .
\]
Since $\P{|Z_j|\leq\sqrt{3}}=1$ then naturally the tilted version $\tilde Z_j$ is equally bounded almost surely, and therefore $\tilde Z_j - \E{\tilde Z_j}$ is bounded as well. Subsequently,~\eqref{eq:lindeberg-condition} holds such that Theorem~\ref{th:clt-lindeberg} implies
\begin{align*}
\P{0 \leq \frac{\frac{1}{t}\sum_{j\in[t]} \tilde Z_j - \tau_n}{\sqrt{\sigma^2_{\xi_n}/t}} \leq 1} &= \P{0 \leq \frac{\sum_{j\in[t]} \tilde Z_j - \tau_n}{\sqrt{t\sigma^2_{\xi_n}}} \leq 1} \\ &\to \Phi(1) - \Phi(0) \geq \frac{1}{4} \ , \numberthis\label{eq:vq-lb-3}
\end{align*}
and thus for sufficiently large $n$, the probability term in~\eqref{eq:vq-lb-1} is larger than $1/4$.

Note that $\sigma_n^2 \leq 12$ and thus $\sigma_{\xi_n}^2 = \bigO(\sigma_n^2) = \bigO(1)$. Also note that 
\[
(\tau_n-\mu_n)\xi_n = ((\tau_n-\mu_n)/\sigma_n)^2 + \bigO(\xi_n^3) \ ,
\]
and $\sigma_n^2\xi_n^2 = (\tau_n-\mu_n)^2/\sigma_n^2 + \bigO(\xi_n^3)$. Finally, $\tau_n - \mu_n = \sqrt{2\log(n)/t}(\sqrt{q} - \sqrt{r} + o(1))$. Now, we can combine the results~\eqref{eq:vq-lb-1},~\eqref{eq:vq-lb-2} and~\eqref{eq:vq-lb-3} and obtain that:
\begin{align*}
v_{q} &\geq \frac{1}{4}\exp{t\mu_n\xi_n + \frac{t\sigma_n^2\xi_n^2}{2} -t\tau_n\xi_n - \sqrt{t}\xi_n\sigma_{\xi_n} + \bigO\left(t\xi_n^3\right)} \\
&= \frac{1}{4}\exp{-t(\tau_n-\mu_n)\xi_n + \frac{t\sigma_n^2\xi_n^2}{2} + \bigO\left(\sqrt{t}\xi_n + t\xi_n^3\right)} \\
&= \frac{1}{4}\exp{-\frac{t}{2}\frac{(\tau_n-\mu_n)^2}{\sigma_n^2} + \bigO\left(\sqrt{t}\xi_n + t\xi_n^3\right)} \\
&= \frac{1}{4}\exp{-\log(n)\left(\frac{\sqrt{q}-\sqrt{r} + o(1)}{\sigma_n}\right)^2 + \bigO\left(\sqrt{t}\xi_n + t\xi_n^3 \right)} \\
&= \frac{1}{4}\exp{-\log(n)\left(\left(\frac{\sqrt{q}-\sqrt{r}}{\sigma_n}\right)^2 + \bigO\left(\tfrac{1}{\sqrt{\log(n)}} + \sqrt{\tfrac{\log(n)}{t}}\right) + o(1) \right)} \\
&= \frac{1}{4}\exp{-\log(n)\left(\left(\frac{\sqrt{q}-\sqrt{r}}{\sigma_n}\right)^2 + o(1) \right)} \\
&= n^{-\frac{1}{\sigma_n^2}(\sqrt{q}-\sqrt{r})^2 + o(1)} \ ,
\end{align*}
where we have used that $t = \omega(\log(n))$. \qed

\section{Proof of results for particular distributional families}\label{sec:particularizations}

\subsection{Proof of Proposition~\ref{prp:rank-perf-exp}}\label{sec:rank-perf-exp}

The proof follows from particularizing Theorem~\ref{th:rank_hc_test} to the setting of the exponential family. This require the characterization of $\mu_{n,i}$ and $\sigma_{n,i}$ from Definition~\ref{def:rank_trans_char}, which in turn hinges on the characterization of the first two moments of $U_{ij}$. Specifically
\[
\E{U_{ij}} = \frac{1}{2} + \frac{\theta\sigma_0}{2\sqrt{3}\Upsilon_0} + \bigO(\theta^2) \ ,
\]
as $\theta\to 0$, as shown in in the proof of Proposition~2 of \cite{arias-castro2018a}. Furthermore
\[
\E{U_{ij}^2} = \frac{1}{3} + \bigO(\theta) \ ,
\]
as $\theta\to 0$. Both statements follow a Taylor expansion of the above expectations around $\theta=0$, where the first-order term corresponds naturally to the first two moments of the standard uniform distribution. Since $t = \omega(\log(n))$ and thus $\theta \to 0$ we have that
\[
\mu_{n,i} = \frac{\theta\sigma_0}{\Upsilon_0} + \bigO(\theta^2) = \sqrt{\frac{2(\tau^2\rho(\beta,1)/\Upsilon_0^2)(1+o(1)))\log(n)}{t}} \ ,
\]
and
\[
\sigma_{n,i}^2 = 1 + o(1) \ .
\]
Based on this, a direct application of Theorem~\ref{th:rank_hc_test} implies the result in the proposition.
\qed

\subsection{Proof of Proposition~\ref{prp:means-perf-conv}}\label{sec:prp:means-perf-conv_proof}

Denote the standardized sums by
\[
Z_i \equiv \frac{1}{\sqrt{t}}\sum_{j=1}^t X_{ij} \ .
\]
Then the hypothesis test~\eqref{hyp:conv} collapses to:
\begin{align*}
H_0:\qquad &\forall_{i\in[n]} \quad Z_{i} \overset{\text{i.i.d.}}{\sim} \cN(0,1) \ , \numberthis \label{hyp:conv-1dim}\\
H_1:\qquad &\exists_{\cS \subset [n]} : \forall_{i\in\cS} \quad Z_{i} = W_{i} + \theta\sqrt{t} \bar Q_{i} \ \text{ with } W_{i} \overset{\text{i.i.d.}}{\sim} \cN(0,\sigma^2) \text{ and } \bar Q_{i} \overset{\text{i.i.d.}}{\sim} \bar G \ ,
\\ &\text{ and } \forall_{i\notin\cS} \quad  Z_{i} \overset{\text{i.i.d.}}{\sim} \cN(0,1) \ .
\end{align*}
where $\bar G$ is the distribution of the average of $t$ i.i.d. random variables from $G$. For simplicity, denote by $F_1$ the distribution of $Z_i$ if $i\in\cS$ under the alternative. Let $\Phi$ denote the cumulative distribution function of the standard normal distribution and $\phi$ its density.

Our aim is to use the results from \cite{Cai2014} to prove our statements, by reasoning in an analogous way as in examples B and C in Section~5 of that paper. To do so, we start by characterizing the log-likelihood ratio:
\[
\ell(y) \equiv \log\left(\frac{ {\rm d}F_1(y)}{ \phi(y)}\right) \ .
\]
To start, letting $i\in\cS$, note that the distribution $F_1$ can be written as:
\begin{align*}
F_1(y) &= \P{\cN(\theta\sqrt{t}\bar Q_i, \sigma^2) \leq y} \\
&= \int_0^1\P{\cN(\theta\sqrt{t}q, \sigma^2) \leq y \given \bar Q_i = q} {\rm d} \overline G(q) \\
&= \E{\Phi\left(\frac{y-\theta\sqrt{t}\bar Q_i}{\sigma}\right)} \ .
\end{align*}
Since $\bar Q_i$ is bounded the above implies
\[
\frac{{\rm d}F_1(y)}{{\rm d} y} = \E{\frac{1}{\sigma}\phi\left(\frac{y-\theta\sqrt{t}\bar Q_i}{\sigma}\right)} \ ,
\]
such that the log-likelihood ratio can be written as:
\begin{align*}
\ell(y) &= \E{\frac{1}{\sigma}\exp{-\frac{1}{2}\left(\frac{y-\theta\sqrt{t}\bar Q_i}{\sigma}\right)^2 + \frac{y^2}{2}}} \ .
\end{align*}
Following Theorem 1 from \cite{Cai2014}, we characterize the log-likelihood at the point $u\sqrt{2\log(n)}$. For convenience, define:
\[
H(q,m) \equiv \exp{m\left[u^2-\left(\frac{u-s(G)^{-1}\sqrt{r}q}{\sigma}\right)^2\right]} \ .
\]
Using the parameterization of $\theta$ and function $H$ above, we can rewrite the log-likelihood as:
\begin{align*}
\ell(u\sqrt{2\log(n)}) &= \log\left(\E{\frac{1}{\sigma}\exp{-\left(\frac{u-s(G)^{-1}\sqrt{r}\bar Q_i}{\sigma}\right)^2\log(n) + u^2\log(n)}}\right) \\
&= \log\left(\E{n^{u^2-\left(\frac{u-s(G)^{-1}\sqrt{r}\bar Q_i}{\sigma}\right)^2}}\right) - \log(\sigma) \\
&= \log\bigg(\E{H(\bar Q_i, \log(n) )}\bigg) - \log(\sigma)\ . \numberthis\label{eq:prp3-int}
\end{align*}
To characterize the term above, we apply Lemma 3 from \cite{Cai2014}. Note that
\[
\frac{\log(H(q,m))}{m} = u^2-\left(\frac{u-s(G)^{-1}\sqrt{r}q}{\sigma}\right)^2 \equiv f(q) \ ,
\]
and since $\bar Q_i$ is bounded there exists an $m_0$ such that $\int_0^1 \exp{m_0f(q)} {\rm d}\bar G(q)) < \infty$, Lemma~3 from \cite{Cai2014} then implies that
\[
\lim_{m\to\infty} \frac{1}{m}\log\bigg(\E{H(\overline Q_i, m )}\bigg) = \text{ess sup}_{\bar Q_i}\left\{u^2-\left(\frac{u-s(G)^{-1}\sqrt{r}\bar Q_i}{\sigma}\right)^2 \right\} \ .
\]
Note that the function within the essential supremum above is increasing in the argument $\bar Q_i$ and $\text{ess sup}_{\bar Q_i}\{\bar Q_i\} = s(G)$. Now, noting that the final term in~\eqref{eq:prp3-int} is constant and letting $m = \log(n)$, we find:
\[
\ell(u\sqrt{2\log(n)}) = \log(n)(1+o(1))\left(u^2-\left(\frac{u-\sqrt{r}}{\sigma}\right)^2 \right) \ .
\]
Since we have now arrived at the same expression for the log-likelihood ratio as in (62) in \cite{Cai2014} (apart from the extra vanishing term $1+o(1)$), we can continue analogously to ultimately conclude that any test is asymptotically powerless when
\[
r < \rho(\beta,\sigma) \ .
\]
\qed

\subsection{Proof of Proposition~\ref{prp:rank-perf-conv}}\label{sec:rank-perf-conv}

Like in the proof of Proposition~\ref{prp:rank-perf-exp}, we must characterize the two quantities of interest for asymptotic power, $\mu_{n,i}$ and $\sigma_{n,i}$ from Definition~\ref{def:rank_trans_char}, in terms of the parameters $\theta$ and $\sigma$. This requires us to characterize (for $i\in\cS$) the quantities $\E{U_{ij}}$ and $\E{U_{ij}^2}$ in terms of $\theta$ and $\sigma$.

To start, letting $i\in\cS$ and $k\not\in\cS$, we have that:
\begin{align*}
\E{U_{ij}} = \P{X_{ij} > X_{kj}} &= \P{\cN(0,1) > -\frac{\theta Q_{ij}}{\sqrt{\sigma^2+1}}} \\
&= \E{\Phi\left( \frac{\theta Q_{ij}}{\sqrt{\sigma^2+1}} \right)} \\
&= \frac{1}{2} + \frac{\theta \mu(G)}{\sqrt{2\pi(\sigma^2+1)}} + \bigO(\theta^2)  \ ,
\end{align*}
where in the final equality we used a Taylor expansion around $\theta = 0$ and used the fact that $Q_{ij}$ has bounded support. Noting that $t = \omega(\log(n))$ and thus $\theta = o(1)$, we can use the result above to characterize $\mu_{n,i}$ as
\[
\mu_{n,i} = 2\sqrt{3}\left(\frac{\theta \mu(G)}{\sqrt{2\pi(\sigma^2+1)}} + \bigO(\theta^2)\right) = \sqrt{\left(\frac{6\mu(G)^2r}{s(G)^2\pi(\sigma^2+1)}\right)(1+o(1))\frac{2\log(n)}{t}} \ .
\]
We now turn our attention to $\E{U_{ij}^2}$. Let $i\in\cS$ and $k_1,k_2\not\in\cS$. Since $G$ has support bounded on $[0,1]$, we can lower bound as
\begin{align*}
\E{U_{ij}^2} &= \PP\bigg(X_{ij} > \max\{X_{k_1j},X_{k_2j}\}\bigg) \geq \PP\bigg(\cN(0,\sigma^2) > \max\{X_{k_1j},X_{k_2j}\} \bigg) \\
&= \E{\Phi(\cN(0,\sigma^2))^2} \ ,
\end{align*}
and upper bound as:
\begin{align*}
\E{U_{ij}^2} &= \PP\bigg(X_{ij} > \max\{X_{k_1j},X_{k_2j}\}\bigg) \leq \PP\bigg(\cN(\theta,\sigma^2) > \max\{X_{k_1j},X_{k_2j}\} \bigg) \\
&= \E{\Phi(\cN(\theta,\sigma^2))^2} \ ,
\end{align*}
Noting that the upper bound above converges to the value of the lower bound as $\theta \to 0$, the results above thus imply that 
\[
\sigma_{n,i}^2 \to 12\left(\E{\Phi(\cN(0,\sigma^2))^2} - \frac{1}{4}\right) = \xi_\sigma^2 \ .
\]
Therefore, asymptotic power is implied by Theorem~\ref{th:rank_hc_test} provided
\[
\frac{6\mu(G)^2r}{s(G)^2\pi(\sigma^2+1)} > \rho(\beta, \xi_\sigma) \ ,
\]
which implies the result.
\qed

%%%%%%%%%%%%%%%%%%%%%%%%%%%%%%%%%%%%%%%%%%%%%%
%% Multiple Appendixes:                     %%
%%%%%%%%%%%%%%%%%%%%%%%%%%%%%%%%%%%%%%%%%%%%%%

\appendix
\section{Auxiliary results and proofs}\label{app:math}

\subsection{Expressing the moments of \texorpdfstring{$U_{ij}$}{Uij}}\label{app:moments-u}
In this section relate the parameters in Definition~\ref{def:rank_trans_char} to simple probabilities of ``correctly ranking'' anomalous observations with higher ranks than nominal observations, leading to the following lemma:

\begin{lem}[Characterization of the rank signal moments]\label{lem:moments-u}
Let $k_1, k_2 \not\in\cS$. Then
\begin{equation}\label{eq:nuj}
\E{U_{ij}} = \P{X_{ij} > X_{k_1j}} + \tfrac{1}{2}\P{X_{ij} = X_{k_1j}}\ ,
\end{equation}
and
\begin{align*}
\E{U_{ij}^2} &= \P{X_{ij} > \max\{X_{k_1j}, X_{k_2j}\}}\\
&\qquad+ \P{X_{ij} = X_{k_2j} > X_{k_1j}} + \tfrac{1}{3}\P{X_{ij} = X_{k_1j} = X_{k_2j}} \numberthis\label{eq:lamj}\ .
\end{align*}
\end{lem}

\begin{proof}
Consider the notation in Section~\ref{sec:ranking} and leading to the statement of the lemma, and let $X_{0j} \sim F_{0j}$ and $W_{0j}\sim \text{Unif}([-1,1])$ independently. We must characterize
\begin{align*}
\E{U_{ij}} &= \E{\lim_{\delta\to 0} F_{0j}^{(\delta)} (X_{ij}^{(\delta)})}\\
&= \E{\lim_{\delta\to 0} \P{X_{0j}+\delta W_{0j}\leq X_{ij}+\delta W_{ij} \given X_{ij}^{(\delta)}}}\\
&= \E{\lim_{\delta\to 0} \P{W_{0j}\leq W_{ij}+\frac{X_{ij}-X_{0j}}{\delta} \given X_{ij},W_{ij}}}\ .
\end{align*}
Let $\delta>0$ be fixed. Clearly
\begin{align*}
\lefteqn{\P{W_{0j}\leq W_{ij}+\frac{X_{ij}-X_{0j}}{\delta} \given X_{ij},W_{ij}}}\\
&=\P{W_{0j}\leq W_{ij}+\frac{X_{ij}-X_{0j}}{\delta} \given X_{ij}=X_{0j},W_{ij}}\P{X_{ij}=X_{0j}}\\
&\qquad +\P{W_{0j}\leq W_{ij}+\frac{X_{ij}-X_{0j}}{\delta} \given X_{ij}>X_{0j},W_{ij}}\P{X_{ij}>X_{0j}}\\
&\qquad +\P{W_{0j}\leq W_{ij}+\frac{X_{ij}-X_{0j}}{\delta} \given X_{ij}<X_{0j},W_{ij}}\P{X_{ij}<X_{0j}}\\
&=\P{W_{0j}\leq W_{ij}\given X_{ij}=X_{0j},W_{ij}}\P{X_{ij}=X_{0j}}\\
&\qquad +\P{W_{0j}\leq W_{ij}+\frac{X_{ij}-X_{0j}}{\delta} \given X_{ij}>X_{0j},W_{ij}}\P{X_{ij}>X_{0j}}\\
&\qquad +\P{W_{0j}\leq W_{ij}+\frac{X_{ij}-X_{0j}}{\delta} \given X_{ij}<X_{0j},W_{ij}}\P{X_{ij}<X_{0j}}\ .
\end{align*}
Therefore
\begin{multline*}
\lim_{\delta\downarrow 0} \P{W_{0j}\leq W_{ij}+\frac{X_{ij}-X_{0j}}{\delta} \given X_{ij},W_{ij}} \\=\P{W_{0j}\leq W_{ij}\given W_{ij}}\P{X_{ij}=X_{0j}} +\P{X_{ij}>X_{0j}}\ .
\end{multline*}
Therefore
\begin{align*}
\E{U_{ij}}&=\P{W_{0j}\leq W_{ij}}\P{X_{ij}=X_{0j}}+\P{X_{ij}>X_{0j}}\\
&=\frac{1}{2}\P{X_{ij}=X_{0j}}+\P{X_{ij}>X_{0j}}\ .
\end{align*}
The second statement in the lemma follows by an analogous argument.
\end{proof}

\subsection{Proof of Lemma~\ref{lem:algebra-system}}\label{sec:algebra-system}
For convenience, we label the the two equations in the system as:
\begin{enumerate}[label=(\roman*)]
\item $1-\beta - \frac{1}{\gamma^2}(\sqrt{q} - \sqrt{r})^2 > \frac{1-q}{2}$\ ,
\item $1-\beta - \frac{1}{\gamma^2}(\sqrt{q} - \sqrt{r})^2 > 0$ \ .
\end{enumerate}
A visualization of the upcoming relevant functions is provided in Figure~\ref{fig:overview-system}. We treat three cases separately to cover all cases in the lemma.

\begin{figure}[htbp]
\centering
\begin{subfigure}{0.95\textwidth}
\hspace{1cm}\includegraphics[width=0.95\textwidth]{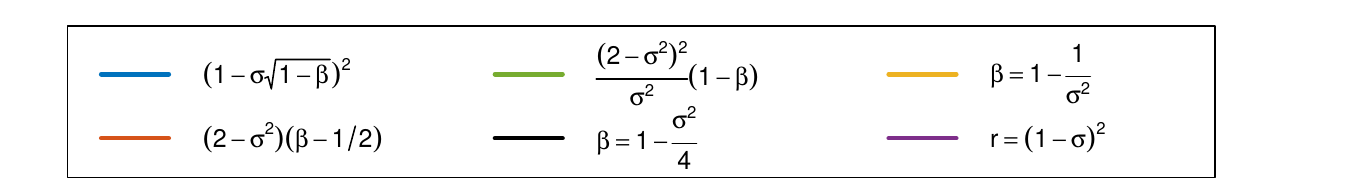}\vspace{0.5cm}
\end{subfigure}

\begin{subfigure}{0.48\textwidth}
\includegraphics[width=\textwidth]{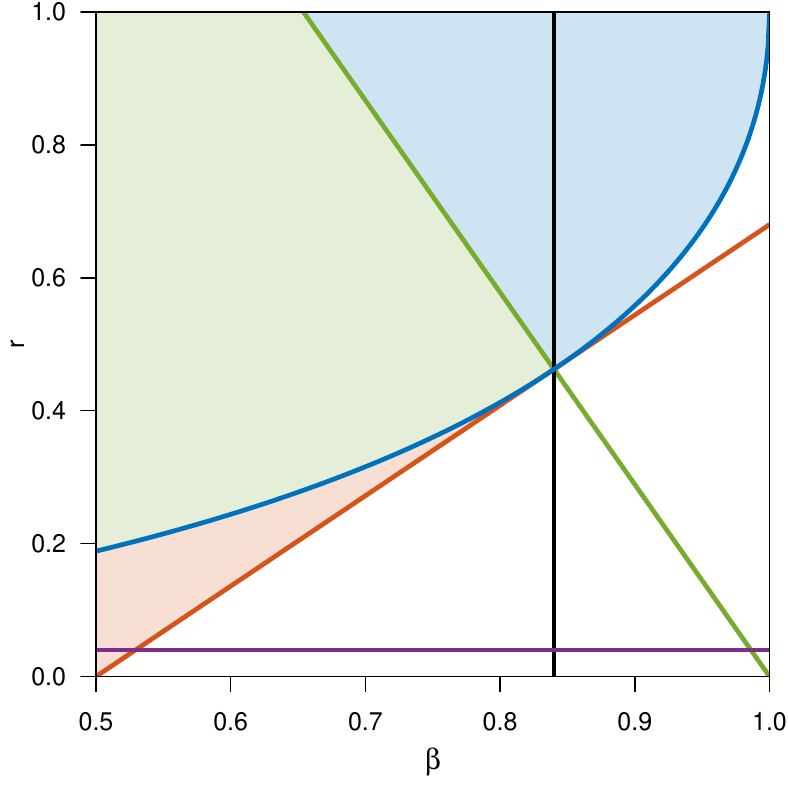}\caption{System when $\sigma=0.8$.}
\end{subfigure}
\begin{subfigure}{0.48\textwidth}
\includegraphics[width=\textwidth]{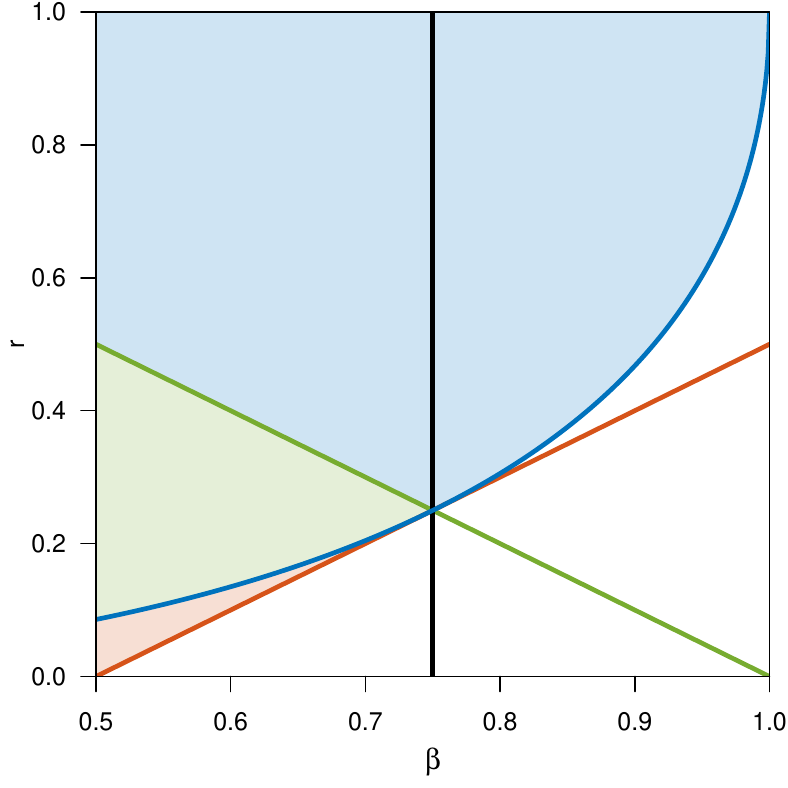}\caption{System when $\sigma=1$.}
\end{subfigure}
\begin{subfigure}{0.48\textwidth}
\includegraphics[width=\textwidth]{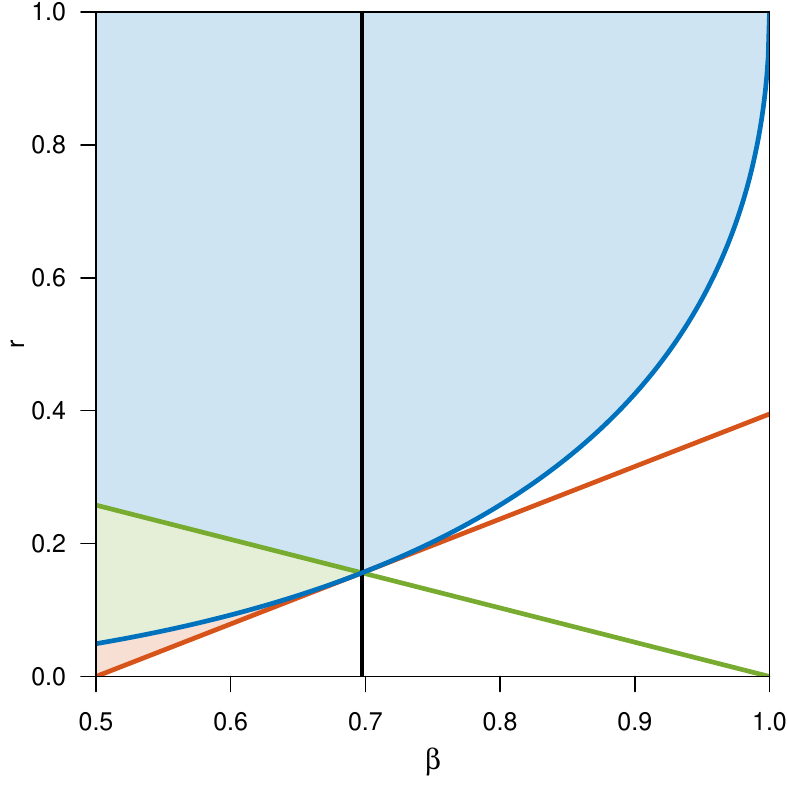}\caption{System when $\sigma=1.1$.}
\end{subfigure}
\begin{subfigure}{0.48\textwidth}
\includegraphics[width=\textwidth]{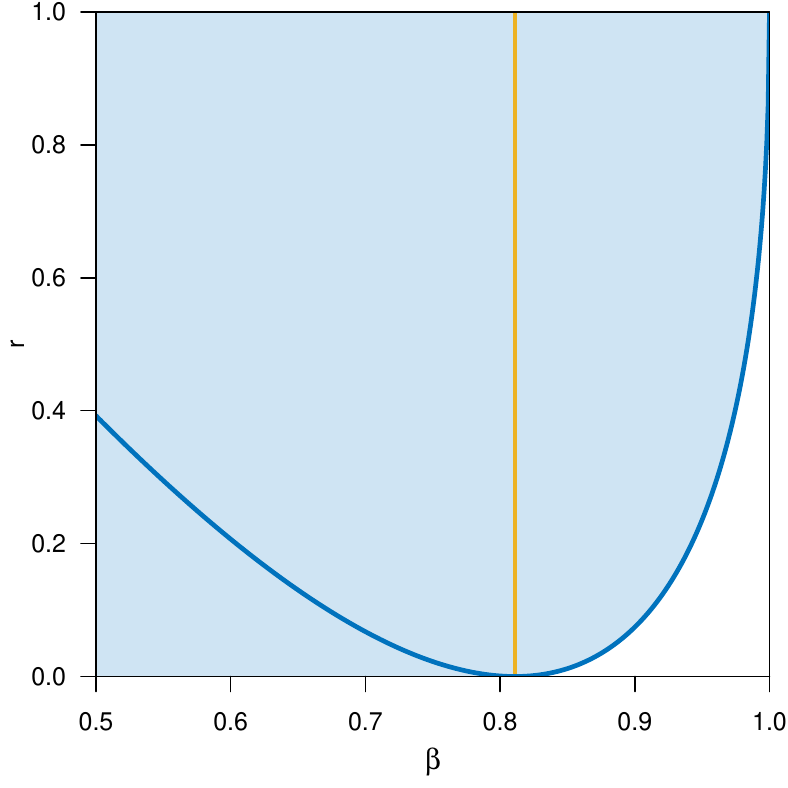}\caption{System when $\sigma=2.3$.}
\end{subfigure}
\caption{Visualization of the system of equations for different values of $\gamma$. Shaded areas correspond to $(\gamma,\beta,r)$ values for which the system is satisfied (relevant functions are derived in the proof of Lemma~\ref{lem:algebra-system}). In the blue shaded region, the system is satisfied provided $q=1$. In the red shaded region, the system is satisfied provided $q=\frac{4r}{(2-\gamma^2)^2}$, and for those $(\gamma,\beta,r)$ values, it holds that $q \leq 1$. In the green shaded region, both choices of $q$ suffice, but the latter choice results in $q \geq 1$.}\label{fig:overview-system}
\end{figure}

\myparagraph{Case $\gamma^2 < 2$ and $r < (1-\gamma\sqrt{1-\beta})^2$ and $\beta \in (\frac{1}{2}, 1-\frac{\gamma^2}{4})$} Plugging in 
\[
q = \min\{\frac{4r}{(\gamma^2-2)^2} , 1\} = \frac{4r}{(\gamma^2-2)^2}
\] 
results in constraints following equation (i)-(ii) as:
\begin{enumerate}[label=(\roman*)]
\item $r > (2-\gamma^2)(\beta - \frac{1}{2})$ \ ,
\item $r < \frac{(2-\gamma^2)^2}{\gamma^2}(1-\beta)$ \ .
\end{enumerate}
Constraint (i) is satisfied since $r > \rho_1^*(\beta,\gamma)$. Since $r < (1-\gamma\sqrt{1-\beta})^2 < \frac{(2-\gamma^2)^2}{\gamma^2}(1-\beta)$ for $\beta \in (\frac{1}{2}, 1-\frac{\gamma^2}{4})$, constraint (ii) is satisfied as well.

\myparagraph{Case $\gamma^2 < 2$ and $(1-\gamma\sqrt{1-\beta})^2 < r < 1$ and $\beta \in (\frac{1}{2}, 1)$} Plugging in 
\[
q = \min\{\frac{4r}{(\gamma^2-2)^2} , 1\} = 1
\] 
results in equations (i) and (ii) being equal, such that we obtain the single constraint: 
\begin{enumerate}[label=(\roman*)]
\item $1-\gamma\sqrt{1-\beta} < \sqrt{r}$ \ ,
\end{enumerate}
which is true by assumption under this case.

\filbreak
\myparagraph{Case $\gamma^2 \geq 2$ and $\beta \in (\frac{1}{2}, 1)$} Plugging in $q = 1$, equations (i) and (ii) are equal such that we again obtain the single constraint:
\begin{enumerate}[label=(\roman*)]
\item $1-\gamma\sqrt{1-\beta} < \sqrt{r}$ \ .
\end{enumerate}
If $\beta\in(1/2, 1-1/\gamma^2)$ then the left-hand side is negative, such that $r=0$ suffices. If $\beta > 1-1/\gamma^2$ then the constraint is satisfied since in that case we have $\sqrt{r} > \sqrt{\rho_2(\beta,\gamma)} = 1-\gamma\sqrt{1-\beta}$.

\qed

\subsection{Discussion regarding Remark~\ref{rem:variance-zero}}\label{app:variance-zero}
Let $F_{0j}$ be a mixture distribution between a continuous uniform distribution on $[0,1]$ with weight $p$ and a continuous uniform distribution on $[2,3]$ with weight $1-p$. Let $F_i$ for $i\in\cS$ be the continuous uniform distribution on $[1,2]$. Now, let $i\in\cS$ and $k,m\not\in\cS$. Then:
\[
\E{U_{ij}} = \P{X_{ij} > X_{kj}} = p \ ,
\]
and 
\[
\E{U_{ij}^2} = \P{X_{ij} \geq X_{kj}, X_i \geq X_{mj}} = \P{X_{ij} \geq X_{kj} \given X_{ij} \geq X_{mj}}\P{X_i \geq X_{mj}} = p \cdot p \ .
\]
Now, if $p = \frac{1}{2}+\frac{1}{2\sqrt{3}}\sqrt{\frac{2r\log(n)}{t}}$ for some constant $r$, then
\[
\mu_{n,i} = \sqrt{\frac{2r\log(n)}{t}} \text{ and } \sigma_{n,i} = 0 \ .
\]
Now, Theorem~\ref{th:rank_hc_test} implies the test has power converging to one if $r > \rho(\beta,0) = 2\beta - 1$.

\subsection{Discussion regarding Remark~\ref{rem:nonconverging_sigma}}\label{app:nonconverging-sigma}
We suppress the subscript $j$ for legibility. Let $F_0$ be a continuous uniform distribution on $[-1,1]$. Let $F_i$ for $i\in\cS$ be the continuous uniform distribution on $[-2-\sin(n),2+\sin(n)]$. By symmetry of the distributions, it is clear that $\mu_{n,i} = 0$. Now, let $i\in\cS$ and $k,m\not\in\cS$. Then:
\begin{align*}
\E{U_{ij}^2} &= \P{X_i > \max\{X_k, X_m\}} \\
&= \P{X_i > 1} + \P{X_i >\max\{X_k, X_m\} \given X_i\in[-1,1]}\P{X_i\in[-1,1]} \\
&= \frac{2+\sin(n)-1}{4+2\sin(n)} + \frac{1}{3}\cdot\frac{2}{4+2\sin(n)} \\
&= \frac{5+3\sin(n)}{12+6\sin(n)} \ .
\end{align*}
Clearly, $\E{U_{ij}^2}$ does not converge, since $\lim\inf \E{U_{ij}^2} = \frac{1}{3}$ and $\lim\sup \E{U_{ij}^2} = \frac{4}{9}$, and thus $\sigma_{n,i}$ does not converge either.

\subsection{Remark on covariance of \texorpdfstring{$N_q(\bR)$}{Nq(R)}}\label{app:cov-nq}
Note that the variance of $N_q(\bR)$ can be expanded as:
\begin{multline*}
\Varha{N_q(\bR)} = \sum_{i\in[n]} \Varha{\ind{Y_i(\bR) \geq z_q}} \\
+ \sum_{i\in[n]}\sum_{k\neq i} \Covha{\ind{Y_i(\bR) \geq z_q},\ind{Y_k(\bR) \geq z_q}} \ .
\end{multline*}
Under the null hypothesis, one can show the set of rank subject means $\{Y_i(\bR)\}_{i\in[n]}$ is negatively associated, which implies each covariance term in the sum above is negative.
Under the alternative, however, this negative association is not present in general. Nevertheless, at first glance it may still seem that it should at least be possible to show that:
\begin{equation}\label{eq:false-conjecture}
\Varha{N_q(\bR)} \leq d_n\sum_{i\in[n]} \Varha{\ind{Y_i(\bR) \geq z_q}} \ ,
\end{equation}
for some $d_n = n^{o(1)}$. However, one can construct a counterexample such that the inequality above is contradicted. Specifically, let $t = 2$ and $\cS = \{1,\dots,s\}$. For all of the alternative subjects $i \in \cS$, let $X_{i1} = i$. For all but the first of the alternative subjects $i \in \{2,\dots,s\}$, let $X_{i2} = s-i$. Finally, for the first anomalous subject, $X_{12}$ is either $s-1$ or $-1$, both with probability $1/2$. Finally, the null subjects $i  \geq s+1$ and $j = 1,2$, let the observations $X_{ij}$ be continuously uniformly distributed on $[-2,-1]$. A sketch of these densities is given in Figure~\ref{fig:sketch_densities}.

\begin{figure}[htb]
\centering
\begin{subfigure}{0.48\textwidth}
\includegraphics[width=\textwidth]{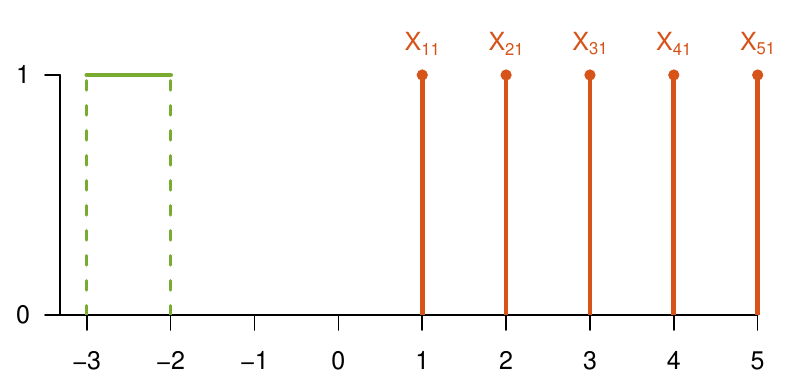}\caption{Density for $t=1$.}
\end{subfigure}
\hfill
\begin{subfigure}{0.48\textwidth}
\includegraphics[width=\textwidth]{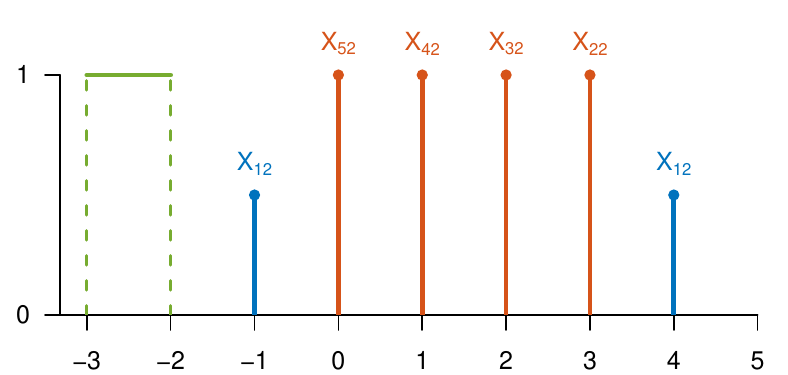}\caption{Density for $t=2$.}
\end{subfigure}
\caption{Sketch of the (degenerate) densities involved in the example discussed in Section~\ref{app:cov-nq}}\label{fig:sketch_densities}
\end{figure} 

Denote $\bX_{0j}$ the vector of observations $\{X_{ij}\}_{i\not\in\cS}$, and $\bR_{0j}$ their ranks, and $Y_0(\bR)$ their rank means. Note that due to the degeneracy of many of the anomalous observations, the resulting data $\bX$ takes two forms:
\[
\bX_1 = \begin{bmatrix}
1 & s-1\\ 
2 & s-2\\ 
3 & s-3 \\
\vdots & \vdots \\ 
s-1 & 1 \\
s &  0 \\
\bX_{01} & \bX_{02} \\
\end{bmatrix}  \ , \quad\text{ or }\quad
\bX_2 = \begin{bmatrix}
1 & -1\\ 
2 & s-2\\ 
3 & s-3 \\
\vdots & \vdots \\ 
s-1 & 1 \\
s &  0 \\
\bX_{01} & \bX_{02} \\
\end{bmatrix} \ ,
\]
each with probability $1/2$. Note that, since all null observations are smaller than $-1$ with probability 1, the set of the observation ranks takes two forms:
\[
\bR_1 = \begin{bmatrix}
n-s-1 & n\\ 
n-s & n-1\\ 
n-s+1 & n-2 \\
\vdots & \vdots \\ 
n-1 & n-s \\
n &  n-s-1 \\
\bR_{01} & \bR_{02} \\
\end{bmatrix} \ , \quad\text{ or }\quad
\bR_1 = \begin{bmatrix}
n-s-1 & n-s-1\\ 
n-s & n\\ 
n-s+1 & n-1 \\
\vdots & \vdots \\ 
n-1 & n-s+1 \\
n &  n-s \\
\bR_{01} & \bR_{02} \\
\end{bmatrix} \ ,
\]
each with probability $1/2$. The rank subject averages of these two sets of ranks correspond to:
\[
\overline \bR_1 = \frac{1}{2}\begin{bmatrix}
2n-s-1 \\
2n-s-1 \\
2n-s-1 \\
\vdots \\
2n-s-1 \\
2n-s-1 \\
Y_0(\bR)
\end{bmatrix}  \ , \quad\text{ or }\quad
\overline \bR_2 = \frac{1}{2}\begin{bmatrix}
2n-2s-2 \\
2n-s \\
2n-s \\
\vdots \\
2n-s \\
2n-s \\
Y_0(\bR)
\end{bmatrix} \ .
\]
It is critical to note that $Y_0(\bR) \leq n-s-2$, such that if $z_q = (2n-s)/2$ (corresponding to $q=3(n-s-1)^2/((n^2-1)\log(n))$) we have that if $i\not\in\cS$ or $i=1$, then $\ind{Y_i(\bR) \geq z_q} = 0$ with probability one. This means that the vector of indicators $\{ \ind{Y_i(\bR) \geq z_q} \}_{i\in[n]}$ is either all zero with probability $1/2$, or
\[
\{ \ind{Y_i(\bR) \geq z_q} \}_{i\in[n]} = \{0,\underbrace{1,\dots,1}_{s-1 \text{ entries}},0,\dots,0\} \ ,
\]
with probability $1/2$. Therefore
\[
\Var{\ind{Y_i(\bR) \geq z_q}} = \begin{cases} \frac{1}{4} \text{ if } i \in \cS \text{ and } i \geq 2 \ , \\
0 \text{ otherwise } \ , \end{cases}
\]
and
\[
\Cov{\ind{Y_i(\bR) \geq z_q}, \ind{Y_k(\bR) \geq z_q} } = \begin{cases} \frac{1}{4} \text{ if } i,k \in \cS \text{ and } i,k \geq 2 \ , \\
0 \text{ otherwise.}\end{cases} 
\]
This implies that
\[
\sum_{i\in[n]}\Var{\ind{Y_i(\bR) \geq z_q}} = \frac{s-1}{4} \ ,
\]
and subsequently
\begin{align*}
\sum_{i\in[n]}\sum_{k\neq i} \Cov{\ind{Y_i(\bR) \geq z_q},\ind{Y_k(\bR) \geq z_q}} &= \frac{(s-1)(s-2)}{4} \\
&= (s-2)\sum_{i\in[n]}\Var{\ind{Y_i(\bR) \geq z_q}} \ .
\end{align*}
The above shows that Equation~\eqref{eq:false-conjecture} cannot be true in general. 

\subsection{Derivation of \texorpdfstring{$\Upsilon_0$}{Y0} in Table~\ref{tb:sim-settings1}}\label{app:upsilon}
First recall that if $X,Y \sim F_0$ i.i.d. continuous with density $f_0$, then
\begin{align*}
\E{\max\left\{\frac{X_1-\mu_0}{\sigma_0}, \frac{X_2-\mu_0}{\sigma_0}\right\}} &= \frac{1}{\sigma_0}\bigg(\E{\max\{X,Y\}} - \mu_0\bigg) \\
&= \frac{1}{\sigma_0}\bigg(2\E{X\ind{X \geq Y}} -\mu_0 \bigg) \\
&= \frac{2}{\sigma_0} \int_{y=0}^\infty \int_{x=y}^\infty xf_0(x)f_0(y) \d x \d y - \frac{\mu_0}{\sigma_0} \ . \numberthis\label{eq:max-integral}
\end{align*}
\myparagraph{Uniform} Let $F_0$ be the continuous uniform distribution on $[0,1]$ and $f_0$ its corresponding density. We then directly obtain that
\[
\int_{y=0}^1 \int_{x=y}^1 xf_0(x)f_0(y) \d x \d y = \int_{y=0}^1 \int_{x=y}^1 x \d x \d y = \int_{y=0}^1 \frac{1}{2} - \frac{y^2}{2}\d y = \frac{1}{3} \ ,
\]
and using that $\sigma_0 = 1/\sqrt{12}$ and $\mu_0 = 1/2$, this implies
\[
\Upsilon_0 = \frac{1}{\sqrt{3}} \left(\frac{2}{3\sigma_0} - \frac{\mu_0}{\sigma_0} \right)^{-1} = 1 \ .
\]
\myparagraph{Exponential} Let $F_0$ be the exponential distribution with rate parameter $\lambda$ and $f_0$ its corresponding density. Using integration by parts we then have that
\begin{align*}
\int_{x=y}^\infty xf_0(x)\d x &=  \int_{x=y}^\infty x\lambda \exp{-\lambda x} \d x = \bigg[ -x\exp{-\lambda x} \bigg ]_{x=y}^\infty + \int_{x=y}^\infty \exp{-\lambda x} \d x \\
&= \bigg[ -\frac{1}{\lambda}\left(1 + x\lambda)\right)\exp{-\lambda x} \bigg ]_{x=y}^\infty = \frac{1}{\lambda}(1+y\lambda)\exp{-\lambda y} \ .
\end{align*}
Now, letting $Z$ an exponential random variable with rate parameter $2\lambda$, we can use the above to write the integral from~\eqref{eq:max-integral} as:
\begin{align*}
\int_{y=0}^\infty \int_{x=y}^\infty xf_0(x)f_0(y) \d x \d y &= \int_{y=0}^\infty (1+y\lambda)\exp{-2\lambda y} \d y \\
&= \bigg[\frac{-1}{2\lambda}\exp{-2\lambda y} \bigg]_{y=0}^\infty + \frac{1}{2}\E{Z} \\
&= \frac{1}{2\lambda} + \frac{1}{4\lambda} = \frac{3}{4\lambda} \ .
\end{align*}
Using $\mu_0 = \sigma_0 = \frac{1}{\lambda}$, we thus obtain that for the exponential model:
\[
\Upsilon_0 = \frac{1}{\sqrt{3}}\left(\frac{6}{4\lambda\sigma_0} - \frac{\mu_0}{\sigma_0} \right)^{-1} = \frac{2}{\sqrt{3}} \ .
\]

\myparagraph{Normal} The case of the standard normal distribution is identical to the derivations provided in the appendix of \cite{arias-castro2018a} (apart from the use of a different notation). We include it here for completeness. For $f_0$ corresponding to the standard normal distribution, we have that
\[
\int_{x=y}^\infty xf_0(x){\rm d}x = \frac{1}{\sqrt{2\pi}}\int_{x=y}^\infty x\exp{-\frac{x^2}{2}}{\rm d}x = \frac{1}{\sqrt{2\pi}}\left[-\exp{-\frac{x^2}{2}}\right]_y^\infty = f_0(y) \ .
\]
Therefore, the expression in~\eqref{eq:max-integral} with $\mu_0 = 0$ and $\sigma_0 = 1$ simplifies to
\[
2\int_{y=0}^\infty f_0(y)^2 {\rm d}y = \frac{1}{\pi}\int_{y=0}^\infty \exp{-y^2} {\rm d}y = \frac{1}{\sqrt{\pi}} \ ,
\]
such that $\Upsilon_0 = \sqrt{\pi/3}$.

\subsection{Derivations regarding normalization of Cauchy setting in Section~\ref{sec:sim-settings}}\label{app:sim-cauchy}
Let $X \sim F_1 = C(\theta,1)$ and $Y \sim F_0 = C(0,1)$. Then for $i\in\cS$:
\begin{align*}
\E{U_{ij}} &\equiv \P{X \geq Y_1} = \P{C(\theta,1) \geq C(0,1)} = \P{C(\theta,2) \geq 0} = \frac{1}{2} - \frac{1}{\pi}\arctan\left(\frac{-\theta}{2}\right) \ ,
\end{align*}
such that for $i\in\cS$:
\begin{align*}
\mu_{n,i} &= 2\sqrt{3}\left(\E{U_{ij}} - \tfrac{1}{2}\right) = -\frac{2\sqrt{3}}{\pi}\arctan\left(\tfrac{-\theta}{2}\right) \\
&=  \frac{2\sqrt{3}}{\pi}\arctan\left(\tau\pi\sqrt{\frac{\rho(\beta,1)\log(n)}{6t}}\right) \\
&= \tau\sqrt{\frac{2\rho(\beta,1)\log(n)}{t}} + \bigO\left(\sqrt{\frac{\log(n)}{t}}\right) \ ,
\end{align*}
where the final inequality follows by a Taylor series. If $t=\omega(\log(n)$ the above implies that
\[
\frac{\mu^2t}{2\log(n)} \to \tau^2\rho(\beta,1) \text{ and } \sigma^2 \to 1 \ ,
\]
such that Theorem~\ref{th:rank_hc_test} implies the rank test from Definition~\ref{def:rank_hc_stat} has power converging to one if $\tau > 1$. 

\section{Supplemental simulation results}\label{app:supp-sim}

\FloatBarrier
\subsection{Grid choice}\label{app:sim-grid}
This section contains supplemental results to the discussion in Section~\ref{sec:exp-grid}. All results in this section depict simulated power at $5\%$ significance for the rank test of Theorem~\ref{th:rank_hc_test} under different choices of $k_n$ as in~\eqref{eq:extended-grid}, as a function of the signal strength. The subcaptions refer to the settings of Table~\ref{tb:sim-settings1}, Table~\ref{tb:sim-settings2}, and the text of Section~\ref{sec:sim-settings} with relevant parameters further parameterized using $\tau$ as described in~\eqref{eq:theta-param-exp},~\eqref{eq:theta-param-conv}, and~\eqref{eq:theta-param-cauchy}. In the simulations below, when $n=10^3$ we have $\lvert\cS\rvert = \lceil n^{1-0.85}\rceil= 3$ and if $n=10^4$ then $\lvert\cS\rvert = \lceil n^{1-0.85}\rceil= 4$. Monte-Carlo simulation for the null statistics is based on $10^5$ samples. Each signal level was repeated $10^4$ times, leading to the $95\%$ confidence bars depicted.

\begin{figure}[htb]
\centering
\begin{subfigure}{0.65\textwidth}
    \includegraphics[width=\textwidth]{IMG/figure_grid_legend}\vspace{0.3cm}
\end{subfigure}

\centering
\begin{subfigure}{0.32\textwidth}
    \includegraphics[width=\textwidth]{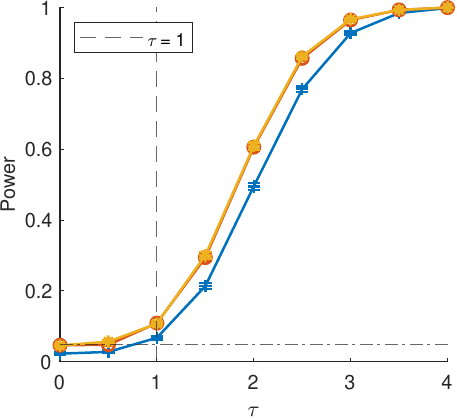}
    \caption{Uniform, $n=10^3$, $t=7$.}
\end{subfigure}
\begin{subfigure}{0.32\textwidth}
    \includegraphics[width=\textwidth]{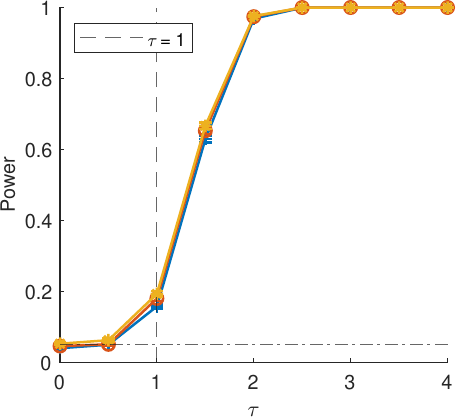}
    \caption{Uniform, $n=10^3$, $t=48$.}
\end{subfigure}
\begin{subfigure}{0.32\textwidth}
    \includegraphics[width=\textwidth]{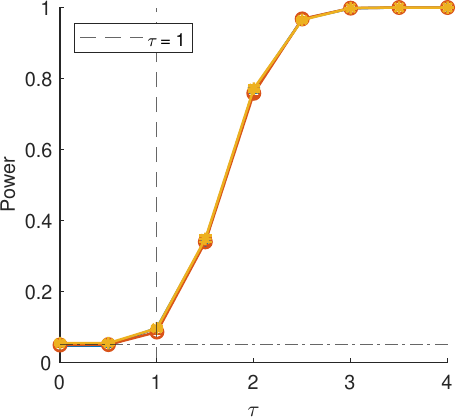}
    \caption{Uniform, $n=10^4$, $t=10$.}
\end{subfigure}
\caption{Supplemental simulation results to Section~\ref{sec:sim} for the uniform setting of Section~\ref{sec:sim-settings}. See Appendix~\ref{app:sim-grid} for details on this simulation study.}
\end{figure}

\begin{figure}[htb]
\centering
\begin{subfigure}{0.65\textwidth}
    \includegraphics[width=\textwidth]{IMG/figure_grid_legend}\vspace{0.3cm}
\end{subfigure}

\begin{subfigure}{0.32\textwidth}
    \includegraphics[width=\textwidth]{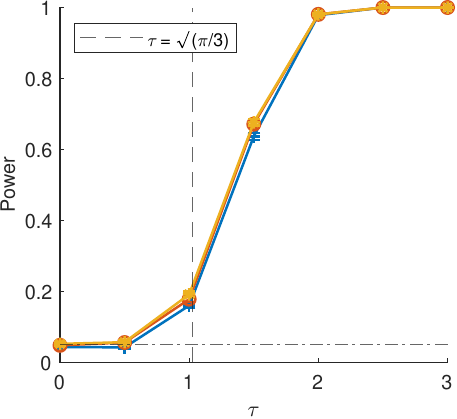}
    \caption{Normal, $n=10^3$, $t=48$.}
\end{subfigure}
\begin{subfigure}{0.32\textwidth}
    \includegraphics[width=\textwidth]{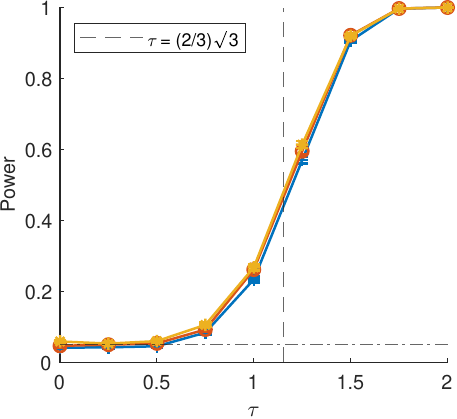}
    \caption{Exp, $n=10^3$, $t=48$.}
\end{subfigure}
\begin{subfigure}{0.32\textwidth}
    \includegraphics[width=\textwidth]{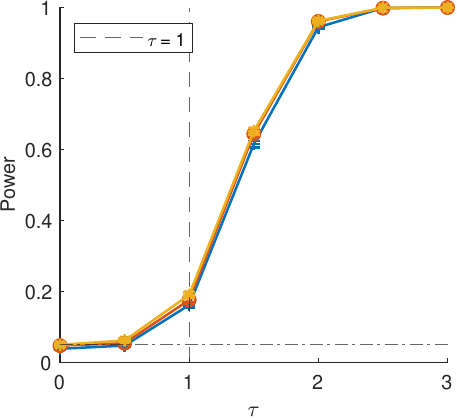}
    \caption{Cauchy, $n=10^3$, $t=48$.}
\end{subfigure}
\hfill
\begin{subfigure}{0.32\textwidth}
    \includegraphics[width=\textwidth]{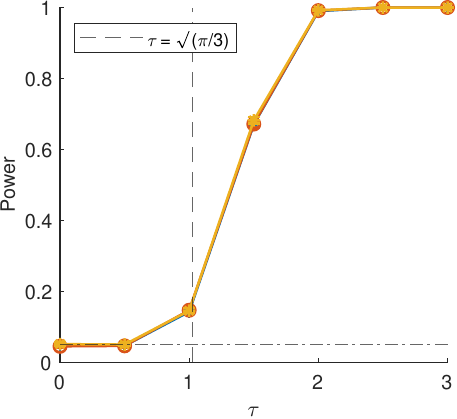}
    \caption{Normal, $n=10^4$, $t=10$.}
\end{subfigure}
\begin{subfigure}{0.32\textwidth}
    \includegraphics[width=\textwidth]{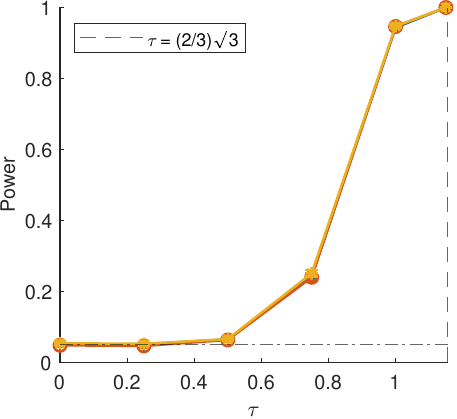}
    \caption{Exp, $n=10^4$, $t=10$.}
\end{subfigure}
\begin{subfigure}{0.32\textwidth}
    \includegraphics[width=\textwidth]{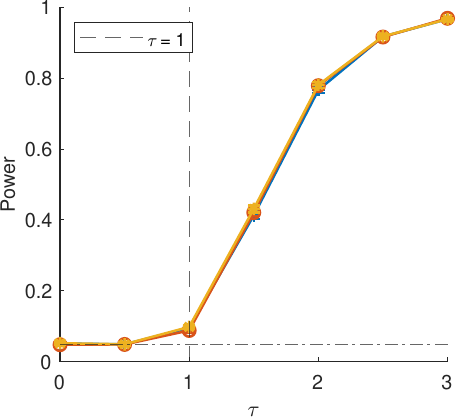}
    \caption{Cauchy, $n=10^4$, $t=10$.}
\end{subfigure}
\caption{Supplemental simulation results to Section~\ref{sec:exp-grid} for the settings described in Section~\ref{sec:sim-settings}. See Appendix~\ref{app:sim-grid} for details on this simulation study.}
\end{figure}

\begin{figure}[htb]
\centering
\begin{subfigure}{0.65\textwidth}
    \includegraphics[width=\textwidth]{IMG/figure_grid_legend}\vspace{0.3cm}
\end{subfigure}

\begin{subfigure}{0.32\textwidth}
    \includegraphics[width=\textwidth]{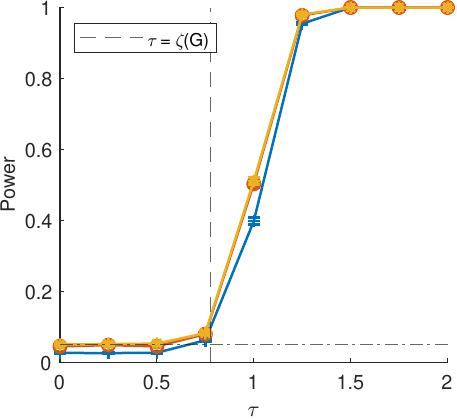}
    \caption{Normal$(\tfrac{1}{2})$, $n=10^3$, $t=7$.}
\end{subfigure}
\begin{subfigure}{0.32\textwidth}
    \includegraphics[width=\textwidth]{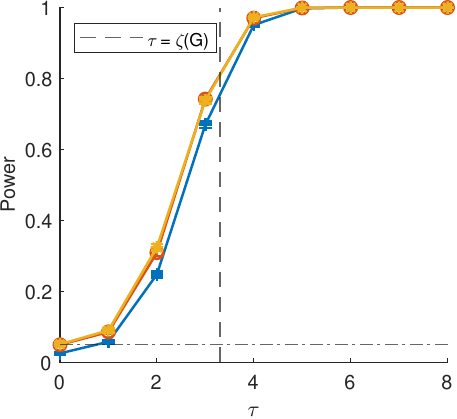}
    \caption{Normal$(\tfrac{3}{2})$, $n=10^3$, $t=7$.}
\end{subfigure}
\begin{subfigure}{0.32\textwidth}
    \includegraphics[width=\textwidth]{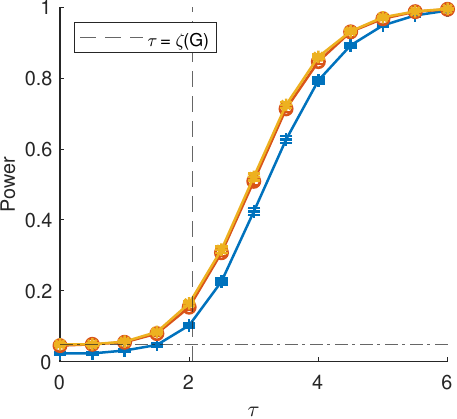}
    \caption{Triang.(1), $n=10^3$, $t=7$.}
\end{subfigure}
\hfill
\begin{subfigure}{0.32\textwidth}
    \includegraphics[width=\textwidth]{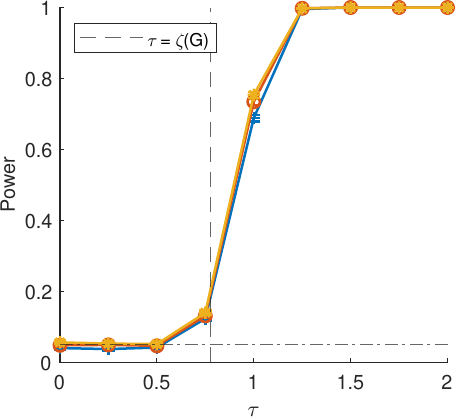}
    \caption{Normal$(\tfrac{1}{2})$, $n=10^3$, $t=48$.}
\end{subfigure}
\begin{subfigure}{0.32\textwidth}
    \includegraphics[width=\textwidth]{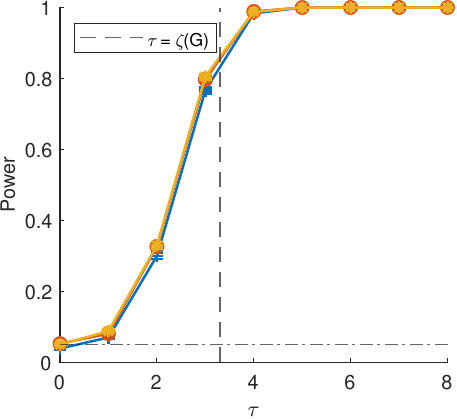}
    \caption{Normal$(\tfrac{3}{2})$, $n=10^3$, $t=48$.}
\end{subfigure}
\begin{subfigure}{0.32\textwidth}
    \includegraphics[width=\textwidth]{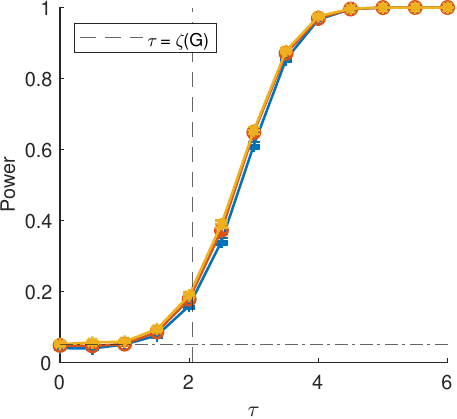}
    \caption{Triang.(1), $n=10^3$, $t=48$.}
\end{subfigure}
\hfill
\begin{subfigure}{0.32\textwidth}
    \includegraphics[width=\textwidth]{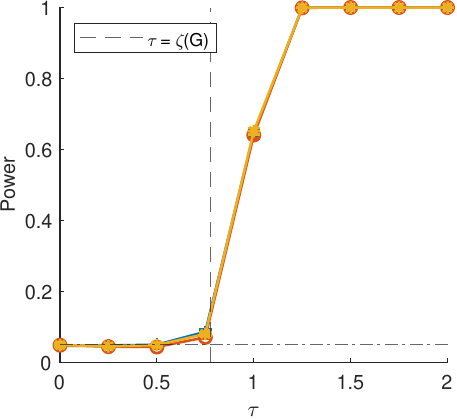}
    \caption{Normal$(\tfrac{1}{2})$, $n=10^4$, $t=10$.}
\end{subfigure}
\begin{subfigure}{0.32\textwidth}
    \includegraphics[width=\textwidth]{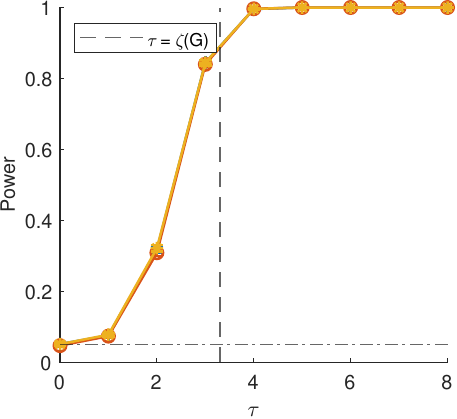}
    \caption{Normal$(\tfrac{3}{2})$, $n=10^4$, $t=10$.}
\end{subfigure}
\begin{subfigure}{0.32\textwidth}
    \includegraphics[width=\textwidth]{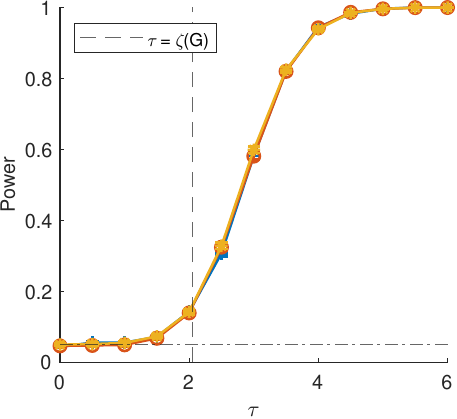}
    \caption{Triang.(1), $n=10^4$, $t=10$.}
\end{subfigure}
\caption{Supplemental simulation results to Section~\ref{sec:exp-grid} for the settings described in Section~\ref{sec:sim-settings}. See Appendix~\ref{app:sim-grid} for details on this simulation study.}
\end{figure}

\FloatBarrier
\subsection{Performance for varying signal strengths}\label{app:sim-perf}
This section contains supplemental results to the discussion in Section~\ref{sec:sim}. All results in this section depict simulated power at $5\%$ significance for the rank test of Theorem~\ref{th:rank_hc_test} and the two tests described in Section~\ref{sec:sim}, as a function of the signal strength. The subcaptions refer to the settings of Table~\ref{tb:sim-settings1}, Table~\ref{tb:sim-settings2}, and the text of Section~\ref{sec:sim-settings} with relevant parameters further parameterized using $\tau$ as described in~\eqref{eq:theta-param-exp},~\eqref{eq:theta-param-conv}, and~\eqref{eq:theta-param-cauchy}. In the simulations below, when $n=10^3$ we have $\lvert\cS\rvert = \lceil n^{1-0.85}\rceil= 3$ and if $n=10^4$ then $\lvert\cS\rvert = \lceil n^{1-0.85}\rceil= 4$. Monte-Carlo simulation for the null statistics is based on $10^5$ samples. Each signal level was repeated $10^3$ times, leading to the $95\%$ confidence bars depicted.

\begin{figure}[htb]
\centering
\begin{subfigure}{0.6\textwidth}
    \includegraphics[width=\textwidth]{IMG/figure_pwr_legend}\vspace{0.3cm}
\end{subfigure}

\begin{subfigure}{0.32\textwidth}
    \includegraphics[width=\textwidth]{IMG/fig_power_cauchy_n=1000_t=7}
    \caption{Cauchy, $n=10^3$, $t=7$.}
\end{subfigure}
\begin{subfigure}{0.32\textwidth}
    \includegraphics[width=\textwidth]{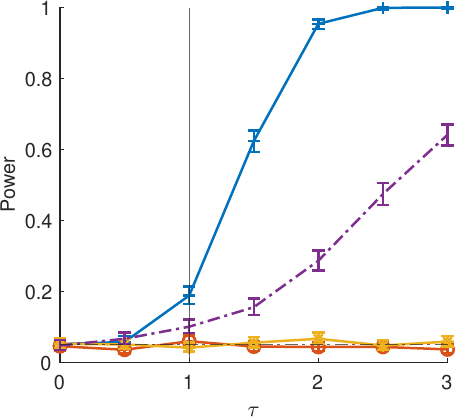}
    \caption{Cauchy, $n=10^3$, $t=48$.}
\end{subfigure}
\begin{subfigure}{0.32\textwidth}
    \includegraphics[width=\textwidth]{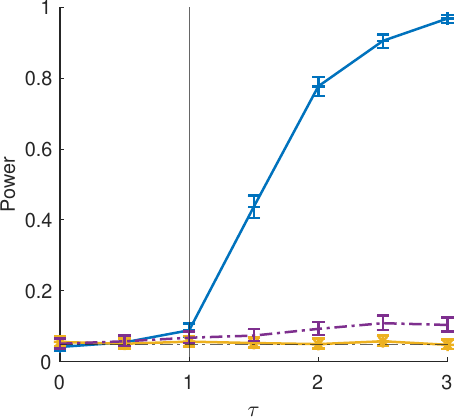}
    \caption{Cauchy, $n=10^4$, $t=10$.}
\end{subfigure}
\caption{Supplemental simulation results to Section~\ref{sec:sim} for the Cauchy setting of Section~\ref{sec:sim-settings}. See Appendix~\ref{app:sim-perf} for details on this simulation study.}
\end{figure}

\begin{figure}[htb]
\centering
\begin{subfigure}{0.6\textwidth}
    \includegraphics[width=\textwidth]{IMG/figure_pwr_legend}\vspace{0.3cm}
\end{subfigure}

\begin{subfigure}{0.32\textwidth}
    \includegraphics[width=\textwidth]{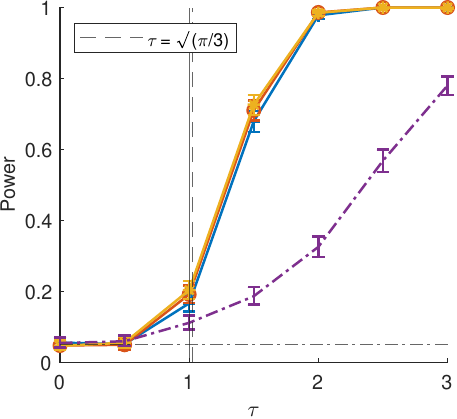}
    \caption{Normal, $n=10^3$, $t=48$.}
\end{subfigure}
\begin{subfigure}{0.32\textwidth}
    \includegraphics[width=\textwidth]{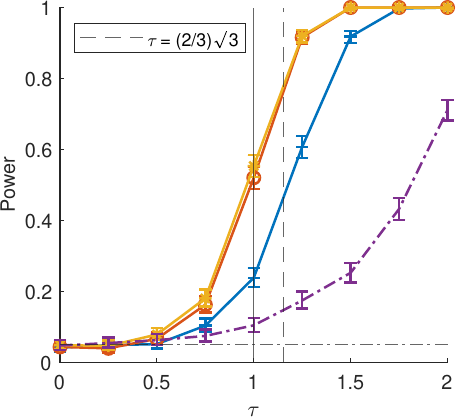}
    \caption{Exp, $n=10^3$, $t=48$.}
\end{subfigure}
\begin{subfigure}{0.32\textwidth}
    \includegraphics[width=\textwidth]{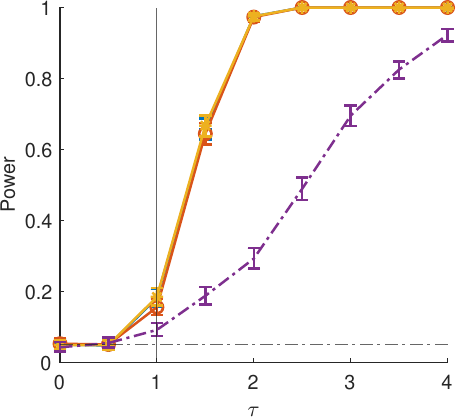}
    \caption{Uniform, $n=10^3$, $t=48$.}
\end{subfigure}
\hfill
\begin{subfigure}{0.32\textwidth}
    \includegraphics[width=\textwidth]{IMG/fig_power_norm_n=10000_t=10}
    \caption{Normal, $n=10^4$, $t=10$.}
\end{subfigure}
\begin{subfigure}{0.32\textwidth}
    \includegraphics[width=\textwidth]{IMG/fig_power_exp_n=10000_t=10}
    \caption{Exp, $n=10^4$, $t=10$.}
\end{subfigure}
\begin{subfigure}{0.32\textwidth}
    \includegraphics[width=\textwidth]{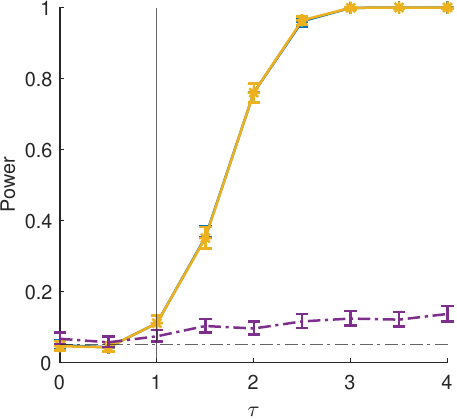}
    \caption{Uniform, $n=10^4$, $t=10$.}
\end{subfigure}
\hfill
\begin{subfigure}{0.32\textwidth}
    \includegraphics[width=\textwidth]{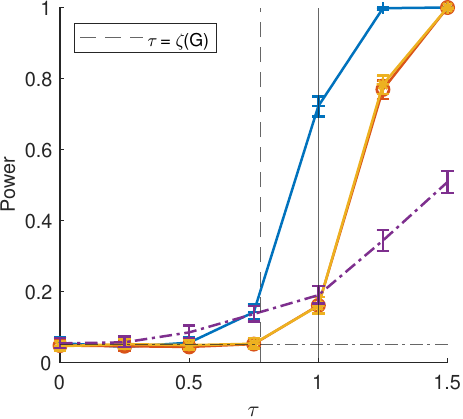}
    \caption{Normal$(\tfrac{1}{2})$, $n=10^3$, $t=48$.}
\end{subfigure}
\begin{subfigure}{0.32\textwidth}
    \includegraphics[width=\textwidth]{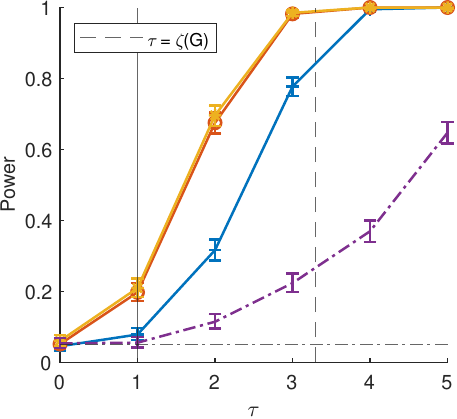}
    \caption{Normal$(\tfrac{3}{2})$, $n=10^3$, $t=48$.}
\end{subfigure}
\begin{subfigure}{0.32\textwidth}
    \includegraphics[width=\textwidth]{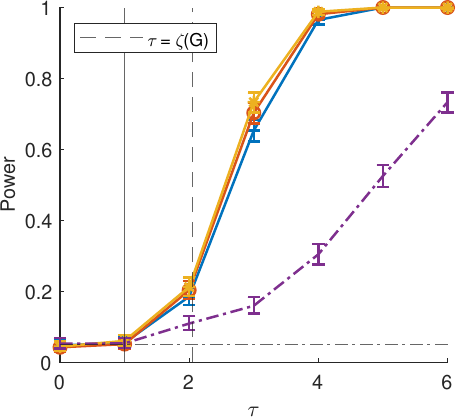}
    \caption{Triang.(1), $n=10^3$, $t=48$.}
\end{subfigure}
\hfill
\begin{subfigure}{0.32\textwidth}
    \includegraphics[width=\textwidth]{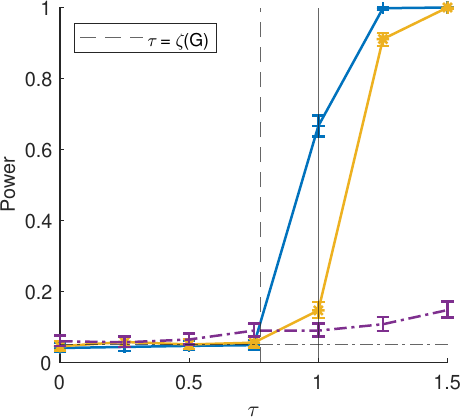}
    \caption{Normal$(\tfrac{1}{2})$, $n=10^4$, $t=10$.}
\end{subfigure}
\begin{subfigure}{0.32\textwidth}
    \includegraphics[width=\textwidth]{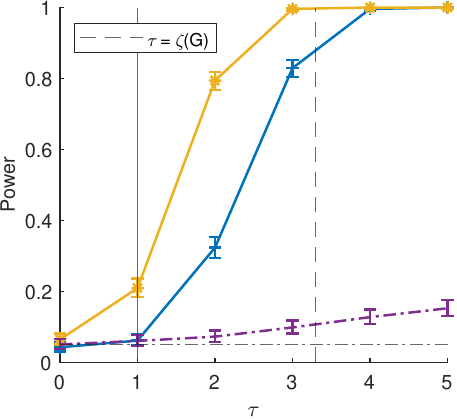}
    \caption{Normal$(\tfrac{3}{2})$, $n=10^4$, $t=10$.}
\end{subfigure}
\begin{subfigure}{0.32\textwidth}
    \includegraphics[width=\textwidth]{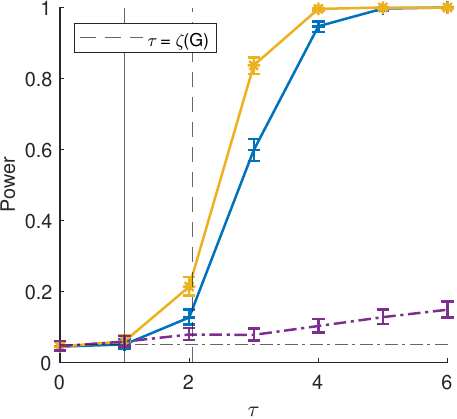}
    \caption{Triang.(1), $n=10^4$, $t=10$.}
\end{subfigure}
\caption{Supplemental simulation results to Section~\ref{sec:sim} for the settings described in Section~\ref{sec:sim-settings}. See Appendix~\ref{app:sim-perf} for details on this simulation study.}
\end{figure}

\FloatBarrier
\clearpage
\subsection{Performance for varying values of \texorpdfstring{$t$}{t}}\label{app:streamlengths}
This section contains supplemental results to the discussion in Section~\ref{sec:discussion}. All results in this section depict simulated power at $5\%$ significance for the rank test of Theorem~\ref{th:rank_hc_test} as a function of the signal strength, for various values of $t$. The normal location setting of Table~\ref{tb:sim-settings1} is considered with signal magnitude parameters further parameterized using $\tau$ as described in~\eqref{eq:theta-param-exp}. Note that for this setting, the value of $t$ is irrelevant for the means-based oracle test and as such only the case of $t=1$ is reported. In the simulations below $n=100$ and $\lvert\cS\rvert = 6$. Monte-Carlo simulation for the null statistics is based on $10^5$ samples. Each signal level was repeated $10^4$ times, leading to the $95\%$ confidence bars depicted.

\begin{figure}[htb]
\centering
\includegraphics[width=0.7\textwidth]{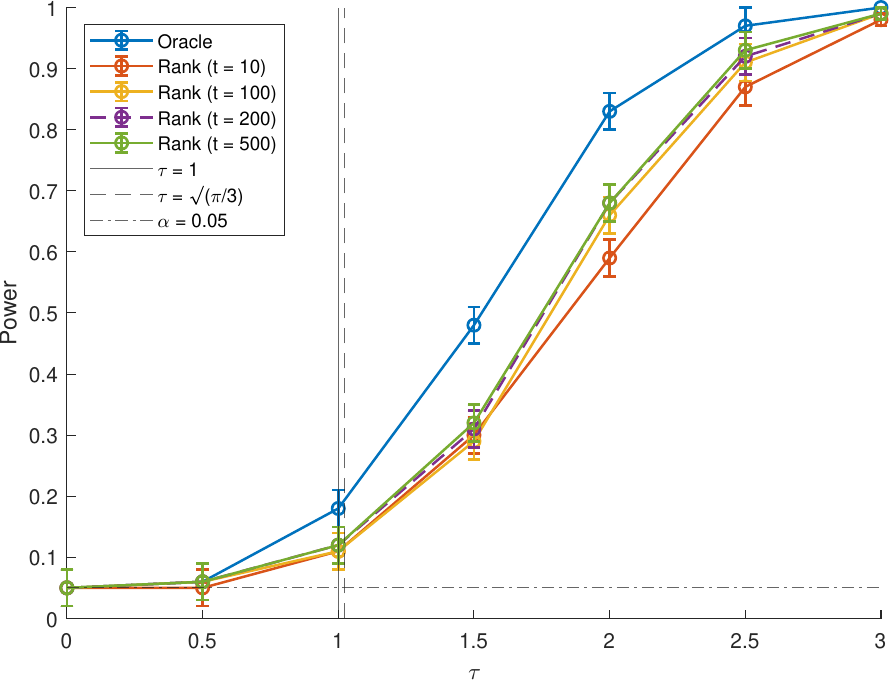}
\caption{Supplemental simulation results to Section~\ref{sec:discussion}. See Appendix~\ref{app:streamlengths} for details on this simulation study.}\label{fig:streamlengths}
\end{figure}

\FloatBarrier

\newpage
\section{Supplemental results for Section~\ref{sec:application}}\label{app:supp-app}
Supplemental results to Section~\ref{sec:application} are provided in Figure~\ref{fig:supp-application}.

\begin{figure}[htb]
\centering
\begin{subfigure}{0.3\textwidth}
    \includegraphics[width=\textwidth]{IMG/fig_application_histogram_j=1k=1}
    \caption{Subfamily 1, intermediate}
\end{subfigure}
\begin{subfigure}{0.3\textwidth}
    \includegraphics[width=\textwidth]{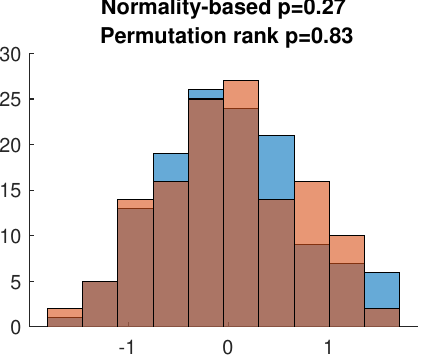}
    \caption{Subfamily 13, intermediate}
\end{subfigure}
\begin{subfigure}{0.3\textwidth}
    \includegraphics[width=\textwidth]{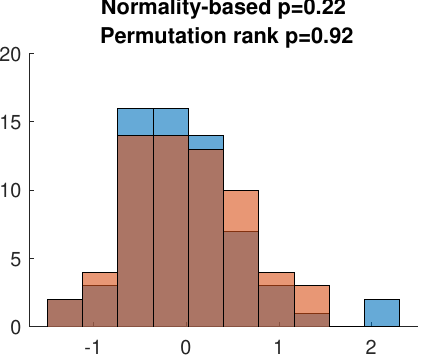}
    \caption{Subfamily 15, intermediate}
\end{subfigure}
\hfill
\begin{subfigure}{0.3\textwidth}
    \includegraphics[width=\textwidth]{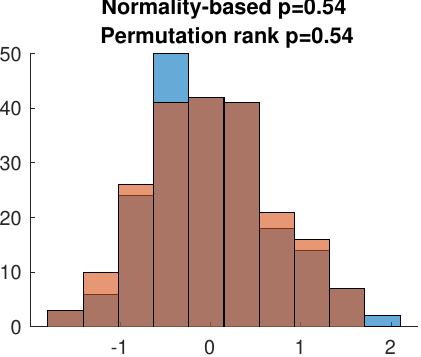}
    \caption{Subfamily 17, intermediate}
\end{subfigure}
\begin{subfigure}{0.3\textwidth}
    \includegraphics[width=\textwidth]{IMG/fig_application_histogram_j=5k=1}
    \caption{Subfamily 21, intermediate}
\end{subfigure}
\begin{subfigure}{0.3\textwidth}
    \includegraphics[width=\textwidth]{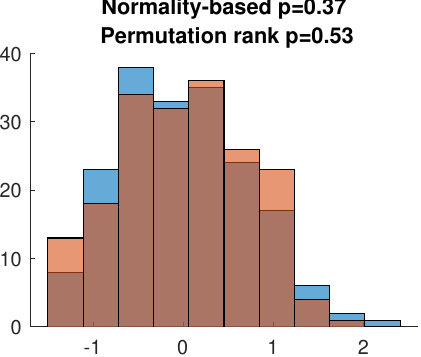}
    \caption{Subfamily 23, intermediate}
\end{subfigure}
\hfill
\begin{subfigure}{0.3\textwidth}
    \includegraphics[width=\textwidth]{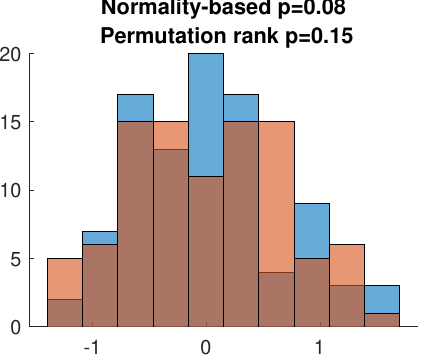}
    \caption{Subfamily 1, final product}
\end{subfigure}
\begin{subfigure}{0.3\textwidth}
    \includegraphics[width=\textwidth]{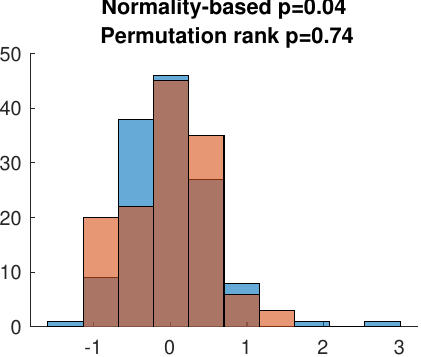}
    \caption{Subfamily 13, final product}
\end{subfigure}
\begin{subfigure}{0.3\textwidth}
    \includegraphics[width=\textwidth]{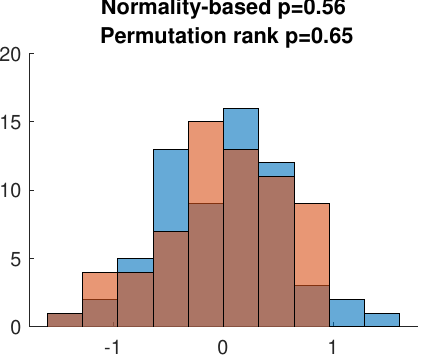}
    \caption{Subfamily 15, final product}
\end{subfigure}
\hfill
\begin{subfigure}{0.3\textwidth}
    \includegraphics[width=\textwidth]{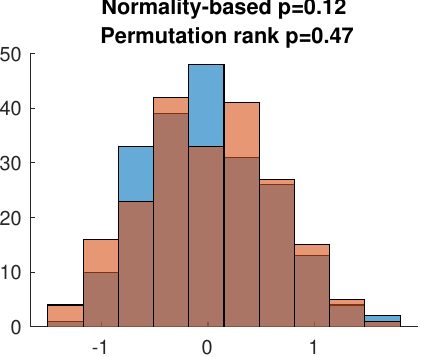}
    \caption{Subfamily 17, final product}
\end{subfigure}
\begin{subfigure}{0.3\textwidth}
    \includegraphics[width=\textwidth]{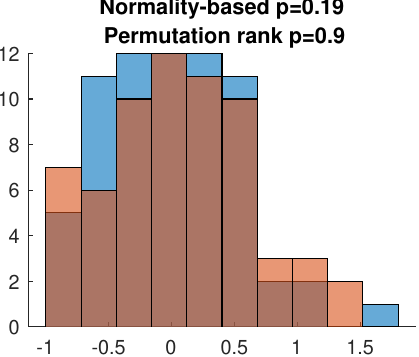}
    \caption{Subfamily 21, final product}
\end{subfigure}
\begin{subfigure}{0.3\textwidth}
    \includegraphics[width=\textwidth]{IMG/fig_application_histogram_j=6k=2}
    \caption{Subfamily 23, final product}
\end{subfigure}
\caption{Supplemental results to Figure~\ref{fig:application-histogram} of Section~\ref{sec:application}.}\label{fig:supp-application}
\end{figure}

\FloatBarrier

\newpage

%%%%%%%%%%%%%%%%%%%%%%%%%%%%%%%%%%%%%%%%%%%%%%
%% Support information, if any,             %%
%% should be provided in the                %%
%% Acknowledgements section.                %%
%%%%%%%%%%%%%%%%%%%%%%%%%%%%%%%%%%%%%%%%%%%%%%
\phantomsection\addcontentsline{toc}{section}{Acknowledgements.}
\textbf{Acknowledgements. }
I.V.S. is grateful for a discussion with Pieta C. IJzerman-Boon on the topic of pharmaceutical production processes, which helped shape the usage of the proposed methodology in the case study of Section \ref{sec:application}.

\bibliographystyle{imsart-nameyear} % Style BST file (imsart-number.bst or imsart-nameyear.bst)
\bibliography{rank-hc}       % Bibliography file (usually '*.bib')

\end{document}